%% file: main.tex
\documentclass[sigconf]{acmart}

\def\fullversion

\usepackage{caption}
\usepackage{bm}
\usepackage{tikz}
\usepackage{style}
\usepackage{bigstrut}
\usepackage{microtype}
\usepackage{tabularx}
\usepackage{multirow}
\usepackage{multicol}
\usepackage{rotating}
\usepackage{balance}
\usepackage{titlecaps}
\usepackage{soul}

\usepackage{pifont}
\usepackage{titlesec}

\usepackage[font=small,labelfont=bf,skip=2pt]{caption}

\usepackage{etoolbox}
\makeatletter
\patchcmd{\@algocf@start}
  {-1.5em}
  {0pt}
  {}{}

\patchcmd{\maketitle}{\@copyrightpermission}{
	\begin{minipage}{0.3\columnwidth}
		\href{https://creativecommons.org/licenses/by/4.0/}{\includegraphics[width=0.90\textwidth]{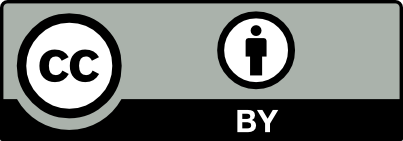}}
	\end{minipage}\hfill
	\begin{minipage}{0.7\columnwidth}
		\href{https://creativecommons.org/licenses/by/4.0/}{This work is licensed under a Creative Commons Attribution International 4.0 License.}
	\end{minipage}
	
	\vspace{5pt}
}{}{}

\makeatother

\hide{
\binoppenalty=700
\brokenpenalty=0 
\clubpenalty=0   
\displaywidowpenalty=0   
\exhyphenpenalty=50
\floatingpenalty=20000
\hyphenpenalty=50
\interlinepenalty=0
\linepenalty=10
\postdisplaypenalty=0
\predisplaypenalty=0 
\relpenalty=500
\widowpenalty=0  
}

\input{macro}
\copyrightyear{2023}
\acmYear{2023}
\setcopyright{rightsretained}
\acmConference[SPAA '23]{Proceedings of the 35th ACM Symposium on Parallelism in Algorithms and Architectures}{June 17--19, 2023}{Orlando, FL, USA}
\acmBooktitle{Proceedings of the 35th ACM Symposium on Parallelism in Algorithms and Architectures (SPAA '23), June 17--19, 2023, Orlando, FL, USA}
\acmDOI{10.1145/3558481.3591069}
\acmISBN{978-1-4503-9545-8/23/06}

\usepackage{etoolbox}
\makeatletter
\patchcmd{\maketitle}{\@copyrightpermission}{
   \begin{minipage}{0.3\columnwidth}
     \href{https://creativecommons.org/licenses/by/4.0/}{\includegraphics[width=0.90\textwidth]{figures/cc_by4acm.png}}
   \end{minipage}\hfill
   \begin{minipage}{0.7\columnwidth}
     \href{https://creativecommons.org/licenses/by/4.0/}{This work is licensed under a Creative Commons Attribution International 4.0 License.}
   \end{minipage}

   \vspace{5pt}
}{}{}

\makeatother

\usepackage{array}
\begin{document}
\fancyhead{}
\title{Parallel Longest Increasing Subsequence and \\van Emde Boas Trees}
  \author{Yan Gu}
  \affiliation{\institution{UC Riverside}\country{}}
  \email{ygu@cs.ucr.edu}
  \author{Ziyang Men}
  \affiliation{\institution{UC Riverside}\country{}}
  \email{zmen002@ucr.edu}
   \author{Zheqi Shen}
  \affiliation{\institution{UC Riverside}\country{}}
  \email{zshen055@ucr.edu}
  \author{Yihan Sun}
  \affiliation{\institution{UC Riverside}\country{}}
  \email{yihans@cs.ucr.edu}
    \author{Zijin Wan}
  \affiliation{\institution{UC Riverside}\country{}}
  \email{zwan019@ucr.edu}




\input{abstract.tex}
\maketitle

\makeatletter
\newcommand{\removelatexerror}{\let\@latex@error\@gobble}
\makeatother

\input{intro.tex}

\input{prelim.tex}

\input{algo.tex}

\input{parallel-veb}

\input{exp.tex}

\input{related.tex}

\input{conclusion.tex}

\section*{Acknowledgement}

This work is supported by NSF grants CCF-2103483, IIS-2227669, NSF CAREER award CCF-2238358, and UCR Regents Faculty Fellowships.
We thank Huacheng Yu for some initial discussions on this topic, especially for the vEB tree part.
We thank anonymous reviewers for the useful feedbacks.

\bibliographystyle{ACM-Reference-Format}
\balance
\bibliography{../bib/strings,../bib/main}

\iffullversion{
\clearpage
\appendix

\input{appendix.tex}
}

\end{document}

%% file: macro.tex
  \newcommand{\codeskip}{{\vspace{.03in}}}
  
  \newcommand{\nodecircle}[1]{{\textcircled{\raisebox{-.05em}{\footnotesize{#1}}}}}

  \newcommand{\pre}{\mathit{pre}\xspace}

  \newcommand{\dg}{dependence graph\xspace}
  \newcommand{\DG}{DG\xspace}

  \newcommand{\val}{\mathit{value}}
  
  \newcommand{\ff}{\mathcal{F}}

  \newcommand{\phaseparallel}{phase-parallel}

  \newcommand{\rank}{\mathit{rank}}

  \newcommand{\mathdp}{\mathit{dp}\xspace}
  \newcommand{\dpvalue}{$\mathdp$ value}
  \newcommand{\DPvalue}{\dpvalue}

  \newcommand{\leftmin}{\mathit{LMin}\xspace}

  \newcommand{\lislength}{k}
  \newcommand{\prefixminfunc}{\textsc{PrefixMin}}
  \newcommand{\processfrontier}{\textsc{ProcessFrontier}}
  \newcommand{\prefixmin}{{prefix-min}}
  \newcommand{\swgs}{SWGS}
  \newcommand{\relevant}{relevant}
  \newcommand{\ranklist}{\mathcal{F}}

  \newcommand{\vebfull}{\textsc{van Emde Boas}\xspace}
  \newcommand{\veb}{\textsc{vEB}\xspace}
  \newcommand{\monoveb}{\textsc{Mono-vEB}\xspace}
  \newcommand{\rangeveb}{\textsc{Range-vEB}\xspace}
  \newcommand{\dommax}{dominant-max\xspace}
  \newcommand{\Dommax}{DominantMax\xspace}
  \newcommand{\updatefunc}{\textsc{Update}\xspace}
  \newcommand{\domby}{\textsc{CoveredBy}\xspace}
  \newcommand{\dominate}{cover\xspace}
  \newcommand{\dominated}{covered\xspace}
  \newcommand{\dominates}{covers\xspace}
  \newcommand{\batchinsert}{\textsc{BatchInsert}\xspace}
  \newcommand{\batchdelete}{\textsc{BatchDelete}\xspace}
  \newcommand{\batchdeletehelper}{\textsc{BatchDeleteRecursive}}
  \newcommand{\insertop}{\textsc{Insert}\xspace}
  \newcommand{\deleteop}{\textsc{Delete}\xspace}
  \newcommand{\minop}{\textsc{Min}\xspace}
  \newcommand{\maxop}{\textsc{Max}\xspace}
  \newcommand{\predop}{\textsc{Predecessor}\xspace}
  \newcommand{\succop}{\textsc{Successor}\xspace}
  \newcommand{\memberop}{\textsc{Member}}
  \newcommand{\twin}{\mathcal{T}}
  \newcommand{\trange}{\mathcal{R}}
  \newcommand{\rangequery}{\textsc{Range}}

  \newcommand{\vt}{\mathcal{V}}
  \newcommand{\tveb}{\vt}
  \newcommand{\setveb}{\mathit{\{\vt\}}}
  \newcommand{\high}{\mathit{high}}
  \newcommand{\low}{\mathit{low}}
  \newcommand{\summary}{\mathit{summary}}
  \newcommand{\cluster}{\mathit{cluster}}
  \newcommand{\idx}{\mathit{index}}
  \newcommand{\predessor}{\textsc{Pred}}
  \newcommand{\successor}{\textsc{Succ}}
  \newcommand{\vebpred}{\predessor}
  \newcommand{\vebsucc}{\successor}
  
  \newcommand{\univ}{\mathcal{U}}
  \newcommand{\bin}{B_{\mathit{in}}}
  \newcommand{\bout}{B_{\mathit{out}}}
  \newcommand{\treenode}{\tau}
  \newcommand{\midd}{\mathit{mid}}
  \newcommand{\kl}{k_L}
  \newcommand{\kr}{k_R}
  \newcommand{\mmin}{\mathit{min}}
  \newcommand{\mmax}{\mathit{max}}
  \newcommand{\aboldsymbol}[1]{{#1}}

  \newcommand{\forkins}{\texttt{fork}}

  \newcommand{\thread}{thread}

  \newtheorem{lemma}{Lemma}[section]
  \newtheorem{definition}{Definition}[section]

  \newcommand{\ifconference}[1]{{{\ifx\fullversion\undefined{#1}\fi}}}
  \newcommand{\iffullversion}[1]{{{\ifx\conference\undefined{#1}\fi}}}

\crefname{section}{Sec.}{Sec.}
\crefname{theorem}{Thm.}{Thm.}
\crefname{algorithm}{Alg.}{Alg.}
\crefname{table}{Tab.}{Tab.}
\crefname{figure}{Fig.}{Fig.}

%% file: abstract.tex
\begin{abstract}
	This paper studies parallel algorithms for the longest increasing subsequence (LIS) problem.
	Let $n$ be the input size and $k$ be the LIS length of the input.
	Sequentially, LIS is a simple problem that can be solved using dynamic programming (DP) in $O(n\log n)$ work.
	However, parallelizing LIS is a long-standing challenge.
	We are unaware of any parallel LIS algorithm that has optimal $O(n\log n)$ work
	and non-trivial parallelism (i.e., $\tilde{O}(k)$ or $o(n)$ span).
	
	This paper proposes a parallel LIS algorithm that costs $O(n\log k)$ work, $\tilde{O}(k)$ span, and $O(n)$ space,
	and is much simpler than the previous parallel LIS algorithms.
	We also generalize the algorithm to a weighted version of LIS, which maximizes the weighted sum for all objects in an increasing subsequence.
	To achieve a better work bound for the weighted LIS algorithm, we designed parallel algorithms for the \vebfull{} (\veb) tree, which
	has the same structure as the sequential \veb{} tree, and supports work-efficient parallel batch insertion, deletion, and range queries.
	
	We also implemented our parallel LIS algorithms.  Our implementation is light-weighted, efficient, and scalable.
	On input size $10^9$, our LIS algorithm outperforms a highly-optimized sequential algorithm (with $O(n\log k)$ cost) on inputs with $k\le 3\times 10^5$.
	Our algorithm is also much faster than the best existing parallel implementation by Shen et al.~(2022) on all input instances.
\end{abstract}


%% file: intro.tex
\section{Introduction}\label{sec:intro}

This paper studies parallel algorithms for classic and weighted longest increasing subsequence problems (LIS and WLIS, see definitions below).
We propose a \defn{work-efficient parallel LIS algorithm} with $\tilde{O}(k)$ span, where $k$ is the LIS length of the input.
Our WLIS algorithm is based on a new data structure that \defn{parallelizes the famous \vebfull{} (\veb) tree}~\cite{van1977preserving}.
Our new algorithms improve existing theoretical bounds on the parallel LIS and WLIS problem, as well as enable simpler and more efficient implementations.
Our parallel \veb{} tree supports work-efficient batch insertion, deletion and range query with polylogarithmic span.

Given a sequence~$A_{1..n}$ and a comparison function on the objects in $A$,
the LIS of $A$ is the longest subsequence (not necessarily contiguous) in $A$ that is strictly increasing (based on the comparison function).
In this paper, we use LIS to refer to both the longest increasing subsequence of a sequence, and the problem of finding such an LIS.
LIS is one of the most fundamental primitives and has extensive applications (e.g.,~\cite{delcher1999alignment,gusfield1997algorithms,crochemore2010fast,schensted1961longest,oprimer,zhang2003alignment,altschul1990basic,deift2000integrable}). 
In this paper, we use $n$ to denote the input size and $\lislength$ to denote the LIS length of the input.
LIS can be solved by dynamic programming (DP) using the following DP recurrence (more details in \cref{sec:prelim}).
\vspace{-.05in}
\begin{equation}\label{eqn:lis}
	\mathdp[i]=\max(1,{\max}_{j<i,A_j<A_i}\mathdp[j]+1)
\end{equation}

Sequentially, LIS is a straightforward textbook problem ~\cite{dasgupta2008algorithms,goodrich2015algorithm}.
We can iteratively compute $\mathdp[i]$ using a search structure to find ${\max}_{j<i,A_j<A_i}\mathdp[j]$, which gives $O(n\log n)$ work.
However, in parallel, LIS becomes challenging both in theory and in practice.
In theory, we are unaware of parallel LIS algorithms with $O(n\log n)$ work and non-trivial parallelism ($o(n)$ or $\tilde{O}(\lislength)$ span).
In practice, we are unaware of parallel LIS implementations
that outperform the sequential algorithm on general input distributions.
\emph{We propose new LIS algorithms with improved work and span bounds in theory, which also lead to
	a more practical parallel LIS implementation. }

Our work follows some recent research~\cite{blelloch2012internally,BFS12,fischer2018tight,hasenplaugh2014ordering,pan2015parallel,shun2015sequential,blelloch2016parallelism,blelloch2018geometry,blelloch2020randomized,gu2022parallel,shen2022many,blelloch2020optimal} 
that directly parallelizes sequential iterative algorithms.
Such algorithms are usually simple and practical, given their connections to sequential algorithms.
To achieve parallelism in a ``sequential'' algorithm,
the key is to identify the \defn{dependences}~\cite{shun2015sequential,blelloch2016parallelism,blelloch2020randomized,shen2022many} among the objects.
In the DP recurrence of LIS,
processing an object $x$ \defn{depends on} all objects $y<x$ before it, but does not need to wait for objects before it with a larger or equal value.

An ``ideal'' parallel algorithm should process all objects in a proper order based on the dependencies---it should 1) process as many objects as possible in parallel (as long as they do not depend on each other), and 2) process an object only when it is \defn{ready} (all objects it depends on are finished) to avoid redundant work.
More formally, we say an algorithm is \defn{round-efficient}~\cite{shen2022many} if its span is $\tilde{O}(D)$ for a computation with the longest logical dependence length~$D$.
In LIS, the logical dependence length given by the DP recurrence is the LIS length $\lislength$.
We say an algorithm is \defn{work-efficient} if its work is asymptotically the same as the best sequential algorithm.
Work-efficiency is \emph{crucial in practice},
since nowadays, the number of processors on one machine (tens to hundreds) is much smaller than the problem size.
A parallel algorithm is less practical if it significantly blows up the work of a sequential algorithm. 

Unfortunately, there exists no parallel LIS algorithm with both work-efficiency and round-efficiency.
Most existing parallel LIS algorithms are not work-efficient~\cite{galil1994parallel,krusche2009parallel,seme2006cgm,thierry2001work,nakashima2002parallel,nakashima2006cost,krusche2010new,shen2022many},
or have $\tilde{\Theta}(n)$ span~\cite{alam2013divide}.
We review more related work in~\cref{sec:related}. 

Our algorithm is based on the parallel LIS algorithm and the \defn{\phaseparallel{} framework} by Shen et al.~\cite{shen2022many}.
We refer to it as the \swgs{} algorithm, and review it in \cref{sec:prelim}.
The \phaseparallel{} framework defines a \defn{rank} for each input object as the length of LIS ending at it (the \dpvalue{} in \cref{eqn:lis}).
Note that an object only depends on lower-rank objects.
Hence, the \phaseparallel{} LIS algorithm processes all objects based on the increasing order of ranks.
However, the \swgs{} algorithm takes $O(n\log^3 n)$ work \whp{}, $\tilde{O}(\lislength)$ span, and $O(n\log n)$ space, and is quite complicated.
In the experiments, the overhead in work and space limits the performance.
On a 96-core machine and input size of $10^8$, \swgs{} becomes slower than a sequential algorithm when the LIS length $k>100$.

In this paper, we propose a parallel LIS algorithm that is \defn{work-efficient} ($O(n\log \lislength)$ work), \defn{round-efficient} ($\tilde{O}(\lislength)$ span) and \defn{space-efficient} ($O(n)$ space), and is \defn{much simpler than previous parallel LIS algorithms}~\cite{shen2022many,krusche2010new}.
Our result is summarized in \cref{thm:main}.

\begin{theorem}[LIS]\label{thm:main}
	Given a sequence $A$ of size $n$ and LIS length $\lislength$,
	the longest increasing subsequence (LIS) of $A$ can be computed in parallel with $O(n\log \lislength)$ work, $O(\lislength\log n)$ span, and $O(n)$ space.
\end{theorem}

We also extend our algorithm to the \defn{weighted LIS (WLIS)} problem, which has a similar DP recurrence as LIS but maximizes the weighted sum instead of the number of objects in an increasing subsequence.
\vspace{-.05in}
\begin{equation}\label{eqn:lisdpweighted}
	\mathdp[i]=w_i + \max(0,{\max}_{j<i,A_j<A_i}\mathdp[j])
\end{equation}

\noindent where $w_i$ is the weight of the $i$-th input object. 
We summarize our result in \cref{thm:mainweighted}.

\hide{
	\begin{theorem}[WLIS]\label{thm:mainweighted}
		Given a sequence $A$ of size $n$ and LIS length $\lislength$, the weighted LIS of $A$ can be computed using $O(n\log n\log \log n)$ work, $O(\lislength\log^2 n)$ span, and $O(n\log n)$ space.
	\end{theorem}
}

\begin{theorem}[WLIS]\label{thm:mainweighted}
	Given a sequence $A$ of size $n$ and LIS length $\lislength$, the weighted LIS of $A$ can be computed using $O(n\log n\log \log n)$ work, $O(\lislength\log^2 n)$ span, and $O(n\log n)$ space.
\end{theorem}

Our primary techniques to support both LIS and WLIS rely on better data structures for 1D or 2D \defn{prefix min/max queries} in the \phaseparallel{} framework.
For the LIS problem, our algorithm efficiently identifies all objects with a certain rank using a \defn{parallel tournament tree}
that supports 1D dynamic prefix-min queries, i.e., given an array of values, find the minimum value for each prefix of the array.
For WLIS, we design efficient data structures for 2D dynamic ``prefix-max'' queries,
which we refer to as \dommax{} queries (see more details in \cref{sec:wlis}).
Given a set of 2D points associated with values, which we refer to as their \emph{scores},
a \dommax{} query returns the largest score to the bottom-left of a query point.
Using \dommax{} queries, given an object $x$ in WLIS, we can find the maximum \dpvalue{} among all objects that $x$ depends on.
We propose two solutions focusing on theoretical and practical efficiency, respectively.
In practice, we use a \defn{parallel range tree} similar to that in \swgs{}, which results in $O(n\log^2 n)$ work and $\tilde{O}(k)$ span for WLIS.
In theory, we parallelize the \defn{\vebfull{} (\veb) tree}~\cite{van1977preserving} and integrate it into range trees to achieve a better work bound for WLIS.

\begin{figure}
	\centering
	\includegraphics[width=\columnwidth]{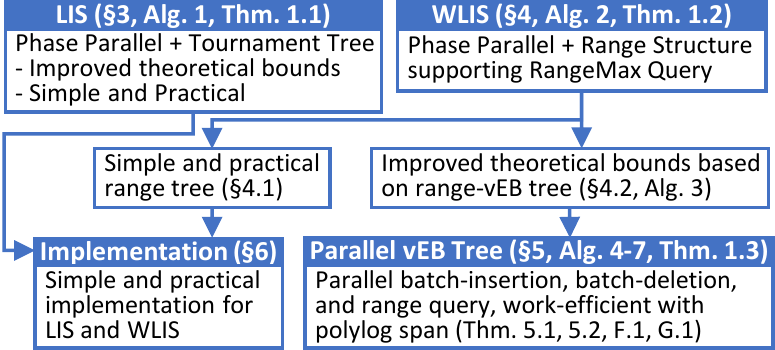}
	\caption{\textbf{Outline and contributions of this paper.}}\label{fig:outline}
	\vspace{-.5em}
\end{figure}

The van Emde Boas (vEB) tree~\cite{van1977preserving} is a famous data structure for priority queues and ordered sets on integer keys, and is introduced in many textbooks (e.g.,~\cite{CLRS}).
To the best of our knowledge, our algorithm is the \textbf{first parallel version of \veb{} trees}. We believe our algorithm is of independent interest in addition to the application in WLIS.
We note that it is highly non-trivial to redesign and parallelize \veb{} trees because the classic \veb tree interface and algorithms are inherently sequential.
Our parallel \veb{} tree supports a general ordered set abstract data type on integer keys in $[0,U)$ with bounds stated below.
We present more details in \cref{sec:veb}.

\begin{theorem}[Parallel \veb{} Tree]\label{thm:veb}
	Let $\univ$ be a universe of all integers in range $[0,U)$.
	Given a set of integer keys from $\univ$, there exists a data structure
	that has the same organization as the sequential \veb{} tree, and supports:
	\begin{itemize}[leftmargin=*]
		\item single-point insertion, deletion, lookup, reporting the minimum (maximum) key, and reporting the predecessor and successor of an element, all in $O(\log\log U)$ work, using the same algorithms for sequential \veb{} trees;
		\item \batchinsert{}$(B)$ and \batchdelete{}$(B)$ that insert and delete a sorted batch $B\subseteq \univ$ in the \veb tree with $O(|B|\log\log U)$ work 
		and $O(\log U)$ span;
		\item \rangequery{}$(k_L,k_R)$ that reports all keys in range $[k_L,k_R]$ in $O((1+m)\log \log U)$ work and $O(\log U\log \log U)$ span, where $m$ is the output size.
	\end{itemize}
\end{theorem}

Our LIS algorithm and the WLIS algorithm based on range trees are simple to program, and we expect them to be the algorithms of choice in implementations in the parallel setting.
We tested our algorithms on a 96-core machine.
Our implementation is \emph{light-weighted, efficient and scalable}.
Our LIS algorithm outperforms \swgs{} in all tests, and is faster than highly-optimized sequential algorithms ~\cite{Knuth73vol3} on reasonable LIS lengths (e.g., up to $k= 3\times 10^5$ for $n=10^9$).
To the best of our knowledge, this is the \emph{first parallel LIS implementation that can outperform the efficient sequential algorithm in a large input parameter space}.
On WLIS, our algorithm is up to 2.5$\times$ faster than \swgs{} and 7$\times$ faster than the sequential algorithm for small $k$ values.
We believe the performance is enabled by the \emph{simplicity} and \emph{theoretical-efficiency} of our new algorithms.

We note that there exist parallel LIS algorithms~\cite{krusche2010new,cao2023nearly} with better worst-case span bounds than our results in theory.
We highlight the \emph{simplicity, practicality, and work-efficiency} of our algorithms.
We also highlight our contributions on \emph{parallel \veb{} trees and the extension to the WLIS problem}.
We believe this paper has mixed contributions of both theory and practice, summarized as follows.

\noindent \textbf{Theory: } 1) Our LIS and WLIS algorithms improve the existing bounds. Our LIS algorithm is the first work- and space-efficient parallel algorithm with non-trivial parallelism ($\tilde{O}(k)$ span). 2) We design the first parallel version of \veb{} trees, which supports work-efficient batch-insertion, batch-deletion and range queries with polylogarithmic span.

\noindent \textbf{Practice: } Our LIS and WLIS algorithms are highly practical and simple to program. Our implementations outperform the state-of-the-art parallel implementation \swgs{} on all tests, due to better work and span bounds. We plan to release our code.


%% file: prelim.tex
\section{Preliminaries} \label{sec:prelim}

\hide{
	\begin{figure*}[t]
		\centering
		\includegraphics[width=\textwidth]{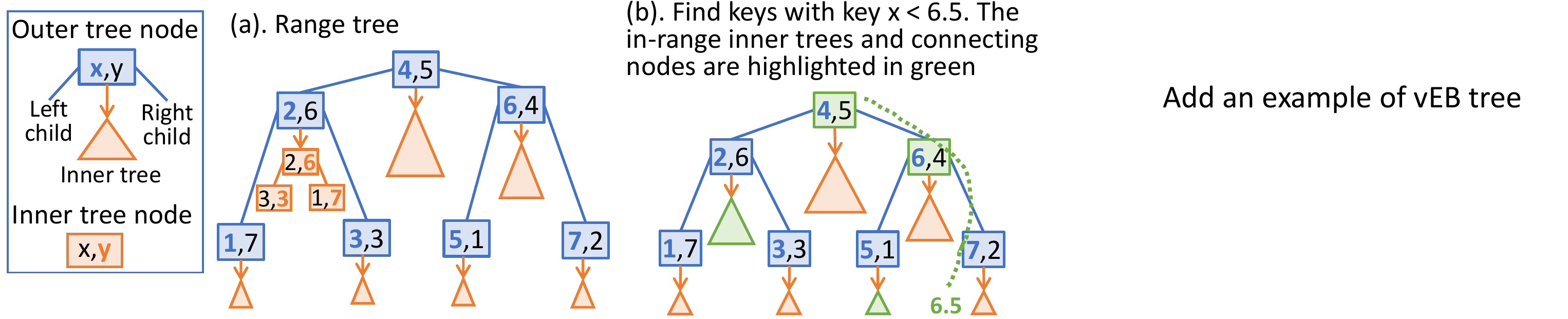}
		\caption{\small\textbf{Illustration of the range tree and \veb{} tree. }
			(a). range tree.
			(b). range tree and in-range inner trees and connecting nodes.
			(c). \veb{} tree.
			\vspace{-.2in}
		}\label{fig:range}
	\end{figure*}
}
\myparagraph{Notation and Computational Model.}
We use $O(f(n))$ \emph{with high probability (\whp{})} (in $n$) to mean $O(cf(n))$ with probability at least $1-n^{-c}$ for $c \geq 1$.
$\tilde{O}(f(n))$ means $O(f(n)\cdot\mathit{polylog}(n))$.
We use $\log n$ as a short form for $1+\log_2(n+1)$.
For an array or sequence $A$, we use $A_i$ and $A[i]$ interchangeably as the $i$-th object in $A$, and
use $A[i..j]$ or $A_{i..j}$ to denote the $i$-th to the $j$-th objects $A$.

We use the \defn{work-span model} in the classic multithreaded model with \defn{binary-forking}~\cite{BL98,arora2001thread,blelloch2020optimal}. 
We assume a set of \thread{}s that share the memory.
Each \thread{} acts like a sequential RAM plus a \forkins{} instruction
that forks two child threads running in parallel.
When both child threads finish, the parent thread continues. 
A parallel-for is simulated by \forkins{} for a logarithmic number of steps.
A computation can be viewed as a DAG (directed acyclic graph).
The \defn{work $\boldsymbol{W}$} of a parallel algorithm is the total number of operations in this DAG,
and the \defn{span (depth) $\boldsymbol{S}$} is the longest path in the DAG.
An algorithm is \defn{work-efficient} if its work is asymptotically the same as the best sequential algorithm.
The randomized work-stealing scheduler can execute such a computation in $W/P+O(S)$ time \whp{} in $W$ on $P$ processor cores~\cite{BL98,arora2001thread,gu2022analysis}.
Our algorithms can also be analyzed on PRAM and have the same work and span bounds.

\hide{Different from the PRAM~\cite{SV81} model, this model assumes loose synchronization and is a practical model supported by most existing libraries~\cite{Openmp,Cilk11,TBB,Java-fork-join,X10}.
	It is used in recent papers for shared-memory parallel algorithms (a short list:~\cite{agrawal2014batching,bender2004fly,Acar02,blelloch2010low,BlellochFiGi11,Cole17,dhulipala2020semi,BBFGGMS18,blelloch2020randomized,dhulipala2019low,dhulipala2022pac,goodrich2021atomic,ahmad2021low}).
	Our algorithms can also be analyzed on PRAM and have the same work and span bounds.
}

\myparagraph{Longest Increasing Subsequence (LIS).} Given a sequence $A_{1..n}$ of $n$ input objects and a comparison function $<$ on objects in $A$,
$A_{1\ldots m}'$ is a subsequence of $A$ if $A'_i=A_{s_i}$, where $1\le s_1< s_2<\dots s_m\le n$.
The \defn{longest increasing subsequence} (LIS) of $A$ is the longest subsequence $A^{*}$ of $A$ where $\forall i<n, A^{*}_i<A^{*}_{i+1}$.
Throughout the paper, we use $n$ to denote the input size, and $\lislength$ to denote the LIS length of the input.

LIS can be solved using dynamic programming (DP) with the DP recurrence in \cref{eqn:lis}.
Here $\mathdp[i]$ (called the \defn{$\boldsymbol{\mathit{dp}}$ value} of object $i$) is the LIS length of $A_{1\ldots i}$ ending with $A_i$.

The LIS problem generalizes to the \defn{weighted LIS (WLIS) problem} with DP recurrence in \cref{eqn:lisdpweighted}.
Sequentially, both LIS and weighted LIS 
can be solved in $O(n\log n)$ work.
This is also the lower bound~\cite{fredman1975computing} w.r.t.\ the number of comparisons.
For (unweighted) LIS, there exists an $O(n\log \lislength)$ sequential algorithm~\cite{Knuth73vol3}.
When the input sequence only contains integers in range $[1,n]$, one can compute the LIS in $O(n\log\log n)$ work using a \veb{} tree.
In our work, we assume general input and only use comparisons between input objects.
Note that although we use \veb{} trees in WLIS,
we will only use it to organize the indexes of the input sequence (see details in \cref{sec:veb}).
Therefore, our algorithm is still comparison-based and work on any input type.

\hide{
	\subsection*{Data Structures} 
	In this section, we present the three data structures used in our algorithms: the \defn{tournament tree}, the \defn{range tree},
	and the \defn{van Emde Boas (\veb{})} tree.
	
	\myparagraph{Tournament tree.} A tournament tree $T$ on $n$ records is a complete binary tree with $2n-1$ nodes (see \cref{fig:tree}).
	It can be represented implicitly as an array $T[1..(2n-1)]$.
	The last $n$ elements are the leaves, where $T[i]$ stores the $(i-n+1)$-th record in the dataset.
	The first $n-1$ elements are internal nodes, each storing the minimum value of its two children.
	The left and right children of $T[i]$ are $T[2i]$ and $T[2i+1]$, respectively.
	We will use a parallel tournament tree to maintain a dynamic prefix-min structure.
	We will use the following theorem about the tournament tree.
	
	\begin{theorem}(Parallel Tournament Trees~\cite{rhostepping,blelloch2020optimal})
		\label{thm:tourtree}
		A tournament tree can be constructed from $n$ elements in $O(n)$ work and $O(\log n)$ span.
		Given a set $S$ of $m$ leaves in the tournament tree with size $n$, the number of ancestors of all the nodes in $S$ is $O(m\log (n/m))$.
	\end{theorem}
	
	Implementing a parallel tournament tree is straightforward.  For instance, construction can be performed by recursively constructing the left and right trees in parallel and updating the root value.
}

\myparagraph{Dependence Graph~\cite{shun2015sequential,blelloch2016parallelism,blelloch2020randomized,shen2022many}. } In a sequential iterative algorithm, we can analyze the logical \emph{dependences} between iterations (objects) to achieve parallelism.
Such dependences can be represented in a DAG, called a \emph{\dg{}} (\DG{}).
In a \DG{}, each vertex is an object in the algorithm.
An edge from $u$ to $v$ means that $v$ can be processed only when $u$ has been finished.
We say $v$ \defn{depends} on $u$ in this case.
\cref{fig:lis-previous} illustrates the dependences in LIS.
We say an object is \defn{ready} when all its predecessors have finished.
When executing a \DG{} with depth $D$, we say an algorithm is \defn{round-efficient} if its span is $\tilde{O}(D)$.
In LIS, the dependence depth given by the DP recurrence is the LIS length $\lislength$.
We note that round-efficiency does not guarantee optimal span, since round-efficiency is with respect to a given \DG{}. One can design a different algorithm with a shallower \DG{} and get a better span.

\begin{figure}[t]
	\centering
	\includegraphics[width=\columnwidth]{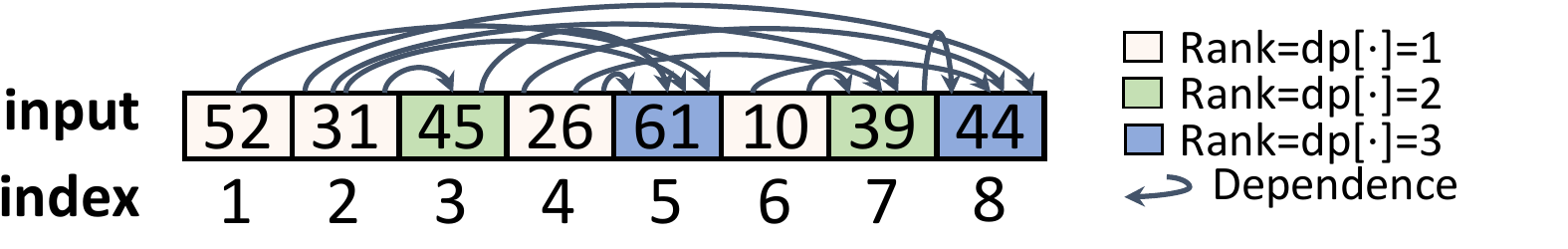}
	\caption{\small
		\textbf{An input for LIS, the dependences and ranks.} An object depends on all objects before it and is smaller than it.
		The rank of an object is the LIS length ending at it, which is also its \dpvalue{}.
	}\label{fig:lis-previous}
\vspace{-.5em}
\end{figure}

\myparagraph{Phase-Parallel Algorithms and \swgs{} Algorithm~\cite{shen2022many}.}
The high-level idea of the \phaseparallel{} algorithm is to assign each object $x$ a \defn{rank}, denoted as $\rank(x)$,
indicating the earliest phase when the object can be processed.
In LIS, the rank of each object is the length of the LIS ending with it (the \DPvalue{} computed by \cref{eqn:lis}).
We also define the \defn{rank of a sequence} $A$ as the LIS length of $A$.
An object only depends on other objects with lower ranks.
The \phaseparallel{} LIS algorithm~\cite{shen2022many} processes all objects with rank $i$ (in parallel) in round $i$.
We call the objects processed in round $i$ the \defn{frontier} of this round.
An LIS example is given in \cref{fig:lis-previous}. 

\newcommand\mycommfont[1]{\textit{\textcolor{blue}{#1}}}
\SetCommentSty{mycommfont}
\hide{
	\begin{algorithm}
		\caption{\small \phaseparallel{} LIS algorithm framework\label{algo:phaseparallel}}
		\DontPrintSemicolon \small
		$i\gets 1$\\
		\While(\tcp*[f]{\emph{$A$ is the input sequence}}){$A\ne \emptyset$\label{step:frameworkwhile}} {
			$\ff_i\gets$ \{$x\in A:$ the LIS ending at $x$ has length $i$\} \\
			Process all $x\in\ff_i$ in parallel (compute the DP values) and set them as finished. \\
			$A\gets A\setminus\ff_i$\\
			$i\gets i+1$
		}
	\end{algorithm}
}

The \swgs{} algorithm uses a \emph{wake-up scheme}, where each object can be processed $O(\log n)$ times \whp{}.
It also uses a range tree to find the frontiers both in LIS and WLIS.
In total, this gives $O(n\log^3 n)$ work \whp{}, $O(\lislength\log^2 n)$ span, and $O(n\log n)$ space.
\hide{
	An object $A_i$ is viewed as a 2D point $(i,A_i)$,
	and it only depends on objects in its lower-left corner (\cref{fig:lis-previous}(c)).
	Hence, the readiness of an object $x$ can be checked using a range tree in $O(\log^2 n)$ cost.
	Each object $x$ is attached to a \emph{pivot} $p_x$, which is a random unfinished ancestor (\cref{fig:lis-previous}(b)).
	When $p_x$ is finished, it attempts to wake up $x$ by checking the readiness of $x$.
	If $x$ is not ready, we assign a new pivot to $x$ and let it sleep again.
}
Our algorithm is also based on the \phaseparallel{} framework but avoids the wake-up scheme to achieve better bounds and performance.


%% file: algo.tex

\section{Longest Increasing Subsequence}
\label{sec:lis}

We start with the (unweighted) LIS problem. Our algorithm is also based on the \phaseparallel{} framework~\cite{shen2022many} but uses a much simpler idea to make it work-efficient.
The work overhead in the \swgs{} algorithm comes from two aspects: range queries on a range tree and the wake-up scheme.
The $O(\log n)$ space overhead comes from the range tree.
Therefore, we want to 1) use a more efficient (and simpler) data structure than the range tree to reduce both work and space, and 2)
wake up and process an object only when it is ready to avoid the wake-up scheme.

Our algorithm is based on a simple observation in \cref{lem:prefixmin} and the concept of \defn{\prefixmin}
objects (\cref{def:prefixmin}).
Recall that the \emph{rank} of an object $A_i$ is exactly its \emph{\dpvalue{}},
which is the \emph{length of LIS ending at $A_i$}.

\begin{definition}[Prefix-min Objects]\label{def:prefixmin}
	Given a sequence $A_{1..n}$, we say $A_i$ is a \defn{\prefixmin} object if for all $j<i$, we have $A_i\le A_j$,
	i.e., $A_i$ is (one of) the smallest object among $A_{1..i}$.
\end{definition}


\begin{lemma}
	\label{lem:prefixmin}
	In a sequence $A$, an object $A_i$ has rank $1$ iff. $A_i$ is a \prefixmin{} object.
	An object $A_i$ has rank $r$  iff.  $A_i$ is a \prefixmin{} object after removing all objects with ranks smaller than $r$.
\end{lemma}

\input{lis-code.tex}

\iffullversion{We use \cref{fig:lis} to illustrate the intuition of \cref{lem:prefixmin}, and prove it in \cref{app:prefixminproof}.}
\ifconference{We use \cref{fig:lis} to illustrate the intuition of \cref{lem:prefixmin}, and the proof is illustrated in the full version of this paper.}
Based on \cref{lem:prefixmin}, we can design an efficient yet simple \phaseparallel{} algorithm for LIS (\cref{algo:lis}).
For simplicity, we first focus on computing the \DPvalue{s} (ranks) of all input objects.
\iffullversion{We show how to output a specific LIS for the input sequence in \cref{app:outputlis}.}
\ifconference{How to output a specific LIS for the input sequence is shown in the full version of this paper.}
The main loop of \cref{algo:lis} is in Lines~\ref{line:whilebegin}--\ref{line:whileend}.
In round $r$, we identify the frontier $\ff_r$ as all the \prefixmin{} objects and set their \dpvalue{s} to $r$.
We then remove the objects in $\ff_r$ and repeat.
\cref{fig:lis} illustrates \cref{algo:lis} by showing the ``prefix-min'' value $\pre_i$ for each object,
which is the smallest value up to each object. Note that this sequence $\pre_i$ is not maintained in our algorithm
but is just used for illustration.
In each round, we find and remove all objects $A_i$ with $A_i=\pre_i$.
Then we update the prefix-min values $\pre_i$, and repeat.
In round $r$, all identified \prefixmin{} objects have rank $r$.
\hide{
	\begin{figure*}[t]
		\centering
		\includegraphics[width=\textwidth]{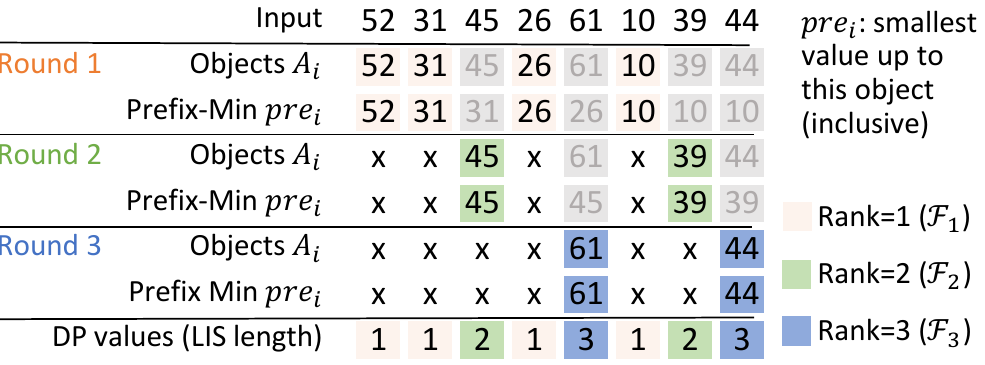}
		\caption{\small \textbf{An illustration of \cref{algo:lis}.}
			The figure also shows $\pre_i$ for each object, which is the smallest object up to this object (inclusive).
			If $A_i=\pre_i$, it is a \prefixmin{} object.
			In round $r$, \cref{algo:lis} finds all \prefixmin{} objects, sets their DP values as $r$, removes them, and updates the $\pre_i$ values.
			The objects identified in round $r$ all have rank $r$ (i.e., the LIS length ending at this object is $r$).
		}
		\label{fig:lis}
	\end{figure*}
}

To achieve work-efficiency, we cannot re-compute the prefix-min values of the entire sequence after each round.
Our approach is to design a parallel \emph{tournament tree} to help identify the frontiers.
Next, we briefly overview the tournament tree and then describe how to use it to find the \prefixmin{} objects efficiently.

\myparagraph{Tournament tree.} A tournament tree $\twin$ on $n$ records is a complete binary tree with $2n-1$ nodes (see \cref{fig:tree}).
It can be represented implicitly as an array $\twin[1..(2n-1)]$.
The last $n$ elements are the leaves, where $\twin[i]$ stores the $(i-n+1)$-th record in the dataset.
The first $n-1$ elements are internal nodes, each storing the minimum value of its two children.
The left and right children of $\twin[i]$ are $\twin[2i]$ and $\twin[2i+1]$, respectively.
We will use the following theorem about the tournament tree.

\vspace{-.5em}

\begin{theorem}(Parallel Tournament Trees~\cite{dong2021efficient,blelloch2020optimal})
	\label{thm:tourtree}
	A tournament tree can be constructed from $n$ elements in $O(n)$ work and $O(\log n)$ span.
	Given a set $S$ of $m$ leaves, in the tournament tree with size $n$, the number of ancestors of all the nodes in $S$ is $O(m\log (n/m))$.
\end{theorem}
\vspace{-.5em}

A tournament tree can be constructed by recursively constructing the left and right trees in parallel, and updating the root value.

\hide{\begin{figure*}[t]
		\begin{minipage}[c]{0.68\textwidth}
			\includegraphics[width=\textwidth]{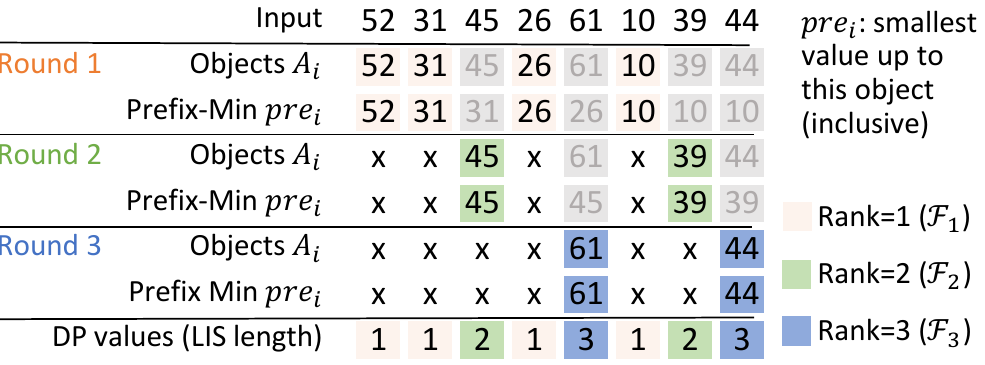}
		\end{minipage}\hfill
		\begin{minipage}[c]{0.30\textwidth}
			\caption{\small \textbf{An illustration of \cref{algo:lis}.}
				The figure also shows $\pre_i$ as the smallest object up to this object (inclusive).
				If $A_i=\pre_i$, it is a \prefixmin{} object.
				In round $r$, \cref{algo:lis} finds all \prefixmin{} objects, sets their DP values as $r$, removes them, and updates the $\pre_i$ values.
			} \label{fig:lis}
		\end{minipage}
		\vspace{-1.5em}
	\end{figure*}
}

\begin{figure}[t]
	\vspace{-1em}\includegraphics[width=\columnwidth]{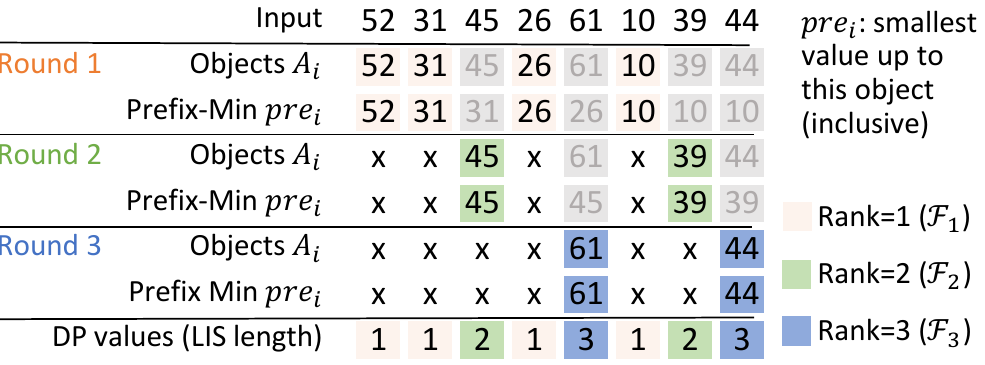}
	\caption{\small \textbf{An illustration of \cref{algo:lis}.}
		The figure also shows $\pre_i$ as the smallest object up to this object (inclusive).
		If $A_i=\pre_i$, it is a \prefixmin{} object.
		In round $r$, \cref{algo:lis} finds all \prefixmin{} objects, sets their DP values as $r$, removes them, and updates the $\pre_i$ values.
		\label{fig:lis}}
	\vspace{-.5em}
\end{figure}

\myparagraph{Using Tournament Tree for LIS.} We use a tournament tree $\twin$ to efficiently identify the frontier and dynamically remove objects (see \cref{algo:lis}).
$\twin$ stores all input objects in the leaves.
We always round up the number of leaves to a power of 2 to make it a full binary tree.
Each internal node stores the minimum value in its subtree.
When we traverse the tree at $\twin[i]$, if the smallest object to its left is smaller than $\twin[i]$,
we can skip the entire subtree.
Using the internal nodes, we can maintain the minimum value before any subtree and skip irrelevant subtrees to save work.

In particular, the function \processfrontier{} finds all \prefixmin{} objects from $\twin$ by calling \prefixminfunc{} starting at the root.
\prefixminfunc{}$(i, \leftmin{})$ traverses the subtree at node $i$,
and finds all leaves $v$ in this subtree s.t. 1) $v$ is no more than any leaf before $v$ in this subtree, and 2) $v$ is no more than $\leftmin$.
The argument $\leftmin$ records the smallest value in $\twin$ before the subtree at $\twin[i]$.
If the smallest value in subtree $\twin[i]$ is larger than $\leftmin$, we can skip the entire subtree (\cref{line:skip}),
because no object in this subtree can be a \prefixmin{} object (they are all larger than $\leftmin$).
Otherwise, there are two cases.
The first case is when $\twin[i]$ is a leaf (Lines~\ref{line:leaf:begin}--\ref{line:leaf:remove}).
Since $\twin[i]\le \leftmin$, it must be a \prefixmin{} object.
Therefore, we set its \DPvalue{} as the current round number $r$ (\cref{line:leaf:setdp}) and remove it by setting its value as $+\infty$ (\cref{line:leaf:remove}).
In second case, when $\twin[i]$ is an internal node (Lines~\ref{line:internal:begin}--\ref{line:compmin}), we can recurse on both subtrees in parallel to find the desired objects (\cref{line:recurse}).
For the left subtree, we directly use the current $\leftmin$ value.
For the right subtree, we need to further consider the minimum value in the left subtree.
Therefore, we take the minimum of the current $\leftmin$ and the smallest
value in the left subtree ($\twin[2i]$), and set it as the $\leftmin$ value of the right recursive call.
After the recursive calls return, we update $\twin[i]$ (\cref{line:compmin}) because some values in the subtree may have been removed (set to $+\infty$).
We present an example in \cref{fig:tree}, which illustrates finding the first frontier for the input in \cref{fig:lis}.

\hide{
	\begin{figure*}[t]
		\centering
		\begin{minipage}[c]{0.58\textwidth}
			\includegraphics[width=\columnwidth]{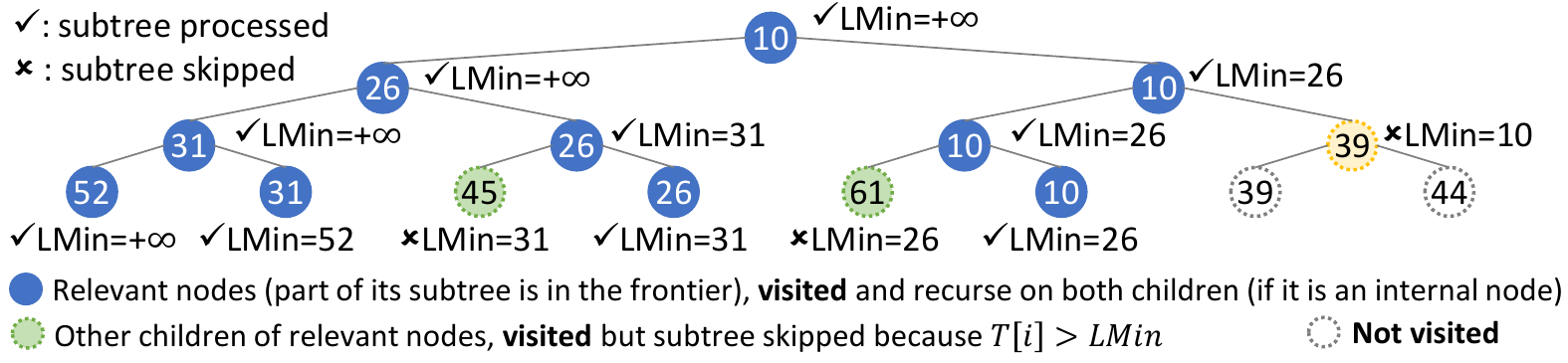}
		\end{minipage}
		\begin{minipage}[c]{0.4\textwidth}
			\caption{\small \textbf{The parallel tournament tree used in \cref{algo:lis}.}
				The leaf nodes store the input sequence ($A_{1..n}$). Each internal node stores the minimum value in its subtree.
				The figure illustrates finding the first frontier in the example in \cref{fig:lis}.
				The algorithm recursively traverses the tree from the root (two subtrees are visited in parallel),
				and maintains a $\leftmin$ value for each subtree, which is the smallest value before this subtree.
				If the $\leftmin$ value is smaller than the value stored at the subtree root, the subtree is skipped.
				For example, at the yellow node with value 39, the smallest value $\leftmin$ before it is 10.
				This means that no leaves in this subtree can be a \prefixmin{} object, so this subtree is skipped.
				\label{fig:tree}
			}
		\end{minipage}
	\end{figure*}
}

We now prove the cost of \cref{algo:lis} in \cref{thm:liscost}.

\vspace{-.5em}
\begin{theorem}\label{thm:liscost}
	\cref{algo:lis} computes the LIS of the input sequence $A$ in $O(n\log \lislength)$ work and $O(\lislength \log n)$ span, where $n$ is the length of the input sequence $A$, and $\lislength$ is the LIS length of $A$.
\end{theorem}
\vspace{-.5em}

\begin{proof}
	Constructing $\twin$ takes $O(n)$ work and $O(\log n)$ span.
	We then focus on the main loop (Lines~\ref{line:whilebegin}--\ref{line:whileend}) of the algorithm.
	The algorithm runs in $\lislength$ rounds.
	In each round, \processfrontier{} recurses for $O(\log n)$ steps.
	Hence, the algorithm has $O(k\log n)$ span.
	
    Next we show that the work of \processfrontier{} in round $r$ is $O(m_r\log (n/m_r))$ work,
	where $m_r=|\ff_r|$ is the number of \prefixmin{} objects identified in this round. 	
	First, note that visiting a tournament tree node has a constant cost, so the work is asymptotically the number of nodes visited in the algorithm.
	We say a node is \defn{\relevant{}} if at least one object in its subtree is in the frontier.
	Based on \cref{thm:tourtree}, there are $O(m_r\log(n/m_r))$ \relevant{} nodes.
	
	If \cref{line:leaf:begin} is executed (i.e., \cref{line:skip} does not return),
	the smallest object in this subtree is no more than $\leftmin$ and must be a \prefixmin{} object,
	and this node is \relevant{}.
	Other nodes are also visited but skipped by \cref{line:skip}.
	Executing~\cref{line:skip} for subtree $i$
	means that $i$'s parent executed \cref{line:internal:begin}, so $i$'s parent is \relevant{}.
	This indicates that a node is visited either because it is \relevant{}, or its parent is \relevant{}.
	Since every node has at most two children, the number of visited nodes is asymptotically the same as all \relevant{} nodes.
	which is $O(m_r\log(n/m_r))$.
	Hence, the total number of visited nodes is:
	\hide{
		\begin{align*}
			&~\sum_{r=1}^{\lislength} m_r\log(n/m_r)\\
			\le &~\sum_{i=1}^{\lislength} (n/\lislength)\log(n/(n/\lislength)) \\
			=&~n\log \lislength
		\end{align*}
	}
\vspace{-.5em}
	\begin{align*}
		\sum_{r=1}^{\lislength} m_r\log(n/m_r)\le\sum_{i=1}^{\lislength} (n/\lislength)\log(n/(n/\lislength))=n\log \lislength
	\end{align*}
	
	The last step uses the concavity of the function $f(x)=x\log_2(1+\frac{n}{x})$. This proves the work bound of the algorithm.
\end{proof}

\begin{figure}[t]
\vspace{-1em}
	\includegraphics[width=\columnwidth]{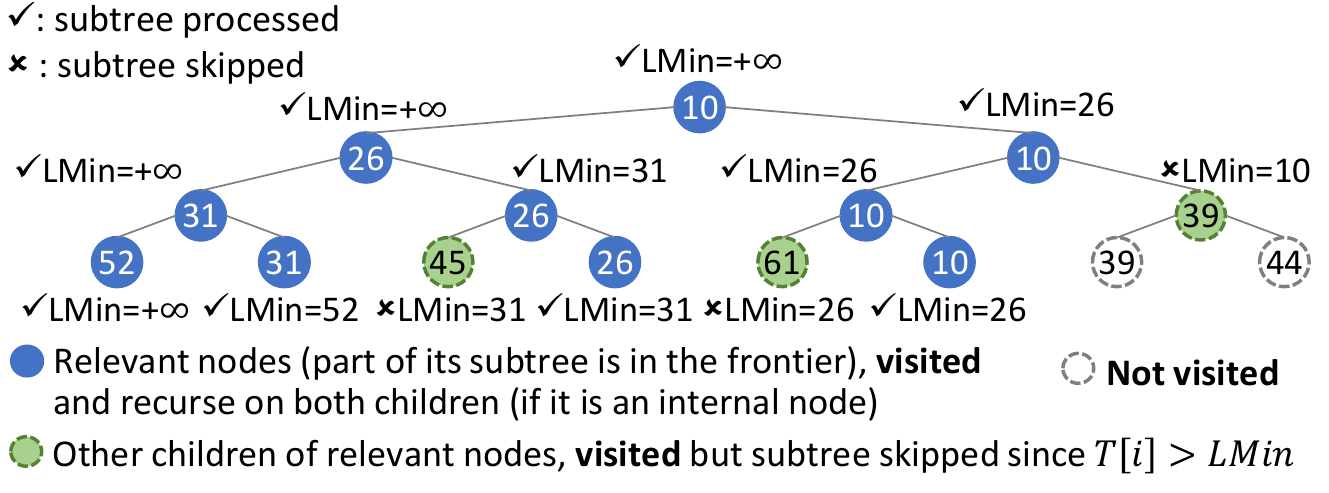}
	\caption{\small \textbf{The parallel tournament tree for \cref{algo:lis}.}
		The leaves store the input $A_{1..n}$. An internal node stores the minimum value in its subtree.
		The figure illustrates finding the first frontier for \cref{fig:lis}.
		The algorithm recursively traverses the tree from the root
		and maintains a $\leftmin$ value for each subtree as the smallest value before this subtree.
		If $\leftmin$ is smaller than the value at the subtree root, we skip the subtree.
		For example, the smallest value before the green node \nodecircle{39} is $\leftmin=10$.
		Therefore, no leaves in this subtree can be a \prefixmin{} object, so this subtree is skipped.
		\label{fig:tree}
	}
\end{figure}

Note that the work bound of \cref{thm:liscost} is parameterized on the LIS length $\lislength$. For small $\lislength$,
the work can be $o(n\log n)$. For example, if the input sequence is strictly decreasing, \cref{algo:lis} only needs $O(n)$ work because
the algorithm will find all objects in the first round in $O(n)$ work and finishes.


\input{wlis.tex}

%% file: lis-code.tex

\begin{algorithm}[t]
\fontsize{8pt}{8pt}\selectfont
\caption{The parallel (unweighted) LIS algorithm\label{algo:lis}}
\SetKwFor{parForEach}{parallel\_for\_each}{do}{endfor}
\KwIn{A sequence $A_{1..n}$}
\KwOut{All \DPvalue{s} (ranks) of $A_{1..n}$}
\SetKwFor{inParallel}{in parallel:}{}{}
    \vspace{0.5em}
\SetKwProg{MyStruct}{Struct}{ contains}{end}
\SetKwProg{RangeTree}{RangeTree$\langle$Point$\rangle$}{ with}{end}
\SetKwProg{NestedTree}{NestedTree<pair<int>{}>}{ with}{end}
\SetKwProg{AugFunc}{}{}{end}
\SetKwProg{MyFunc}{Function}{}{end}
\newcommand{\Int}{\KwSty{int}}
\newcommand{\Bool}{\KwSty{bool}}
\newcommand{\A}{\KwSty{A}}
\newcommand{\lisarray}{\mathit{rank}}
\DontPrintSemicolon

\Int{} $\lisarray[1..n]$ \tcp*[f]{$\lisarray[i]$: the LIS length ending at $A_i$.}\\
\Int{} $\twin[1..(2n-1)]$ \tcp*[f]{$\twin$: the (implicit) tournament tree. } \\
$r\gets 0$\tcp*[f]{$r$ is the current round. }\\
\codeskip
Initialize the tournament tree $\twin$

\MyFunc{\upshape\mf{LIS}(sequence $A_{1..n}$) }{
\While (\tcp*[f]{$\twin$ is not empty.}){$\twin[1]\ne +\infty$\label{line:whilebegin}} {
  $r\gets r+1$\\
  \processfrontier$()$\label{line:whileend} \tcp*[f]{process the $r$-th frontier}\\
}
\Return{$\lisarray[1..n]$}
}

\codeskip
\MyFunc{\upshape\mf{\processfrontier}() }{
  \prefixminfunc(1, $+\infty$)\tcp*[f]{Process all \prefixmin{} objects}\\
}

\codeskip
\tcp{Deal with subtree rooted at $\twin[i]$. Find objects $x$ s.t.: 1) $x\le$ any object before it, and 2) $x\le \leftmin$. Collect such objects in a binary tree.}
\MyFunc{\upshape\mf{PrefixMin}(\Int{} $i$, \Int{} $\leftmin$)}{
  \lIf {$\twin[i]> \leftmin$} {\Return{\texttt{NIL}}\label{line:skip}} 
  \If(\tcp*[f]{Found a leaf node in the frontier.}){$i\ge n$\label{line:leaf:begin}}{
  $\lisarray[i]\gets r$ \label{line:leaf:setdp}\tcp*{Set its rank as $r$.}
  $\twin[i]\gets +\infty$ \label{line:leaf:remove}\tcp*{Remove the object.}
  } \Else(\tcp*[f]{An internal node. Process two children in parallel. }\label{line:internal:begin}){
    \inParallel{\label{line:recurse}} {
    $L\gets$\prefixminfunc($2i$, $\mathit{\leftmin}$)\\
    $R\gets$\prefixminfunc($2i+1$, $\min(\leftmin,\twin[2i])$)
    }
    $\twin[i]\gets \min(\twin[2i],\twin[2i+1])$ \label{line:compmin}\\
    }
  }
\end{algorithm}

\hide{
\begin{algorithm}[t]
\fontsize{8pt}{9pt}\selectfont
\caption{The parallel (unweighted) LIS algorithm\label{algo:lis}}
\SetKwFor{parForEach}{parallel\_for\_each}{do}{endfor}
\KwIn{A sequence $A_{1..n}$ with comparison function $<$}
\KwOut{All \DPvalue{s} (ranks) of $A_{1..n}$}
    \vspace{0.5em}
\SetKwProg{MyStruct}{Struct}{ contains}{end}
\SetKwProg{RangeTree}{RangeTree$\langle$Point$\rangle$}{ with}{end}
\SetKwProg{NestedTree}{NestedTree<pair<int>{}>}{ with}{end}
\SetKwProg{AugFunc}{}{}{end}
\SetKwProg{MyFunc}{Function}{}{end}
\newcommand{\Int}{\KwSty{int}}
\newcommand{\Bool}{\KwSty{bool}}
\newcommand{\A}{\KwSty{A}}
\newcommand{\lisarray}{\mathit{rank}}
\DontPrintSemicolon

\Int{} $\lisarray[1..n]$ \tcp*[f]{$\lisarray[i]$: the LIS length ending at $A_i$.}\\
\Int{} $\twin[1..(2n-1)]$ \tcp*[f]{$\twin$: the (implicit) tournament tree. } \\
$r\gets 0$\tcp*[f]{$r$ is  the current round. }\\
\medskip
Initialize the tree $\twin$

\MyFunc{\upshape\mf{LIS}(sequence $A_{1..n}$) }{
\While (\tcp*[f]{$\twin$ is not empty.}){$\twin[1]\ne +\infty$\label{line:whilebegin}} {
  $r\gets r+1$\\
  \processfrontier$()$\label{line:whileend} 
}
\Return{$\lisarray[1..n]$}
}

\medskip
\MyFunc{\upshape\mf{\processfrontier}() }{
  \prefixminfunc(1, $+\infty$)\\
}

\medskip
\tcp{Deal with subtree rooted at $\twin[i]$. Find objects $x$ s.t.: 1) $x\le$ any object before it, and 2) $x\le \leftmin$.}
\MyFunc{\upshape\mf{PrefixMin}(\Int{} $i$, \Int{} $\leftmin$)}{
  \lIf {$\twin[i]> \leftmin$} {\Return\label{line:skip}} 
  \If(\tcp*[f]{Found a leaf node in the frontier.}){$i\ge n$\label{line:leaf:begin}}{
  $\mathdp[i]\gets r$ \label{line:leaf:setdp}\tcp*{Set its rank as $r$.}
  $\twin[i]\gets +\infty$ \label{line:leaf:remove}\tcp*{Remove the object.}
  } \Else(\tcp*[f]{An internal node. Process two children in parallel. }\label{line:internal:begin}){
    \prefixminfunc($2i$, $\mathit{\leftmin}$) $||$ \prefixminfunc($2i+1$, $\min(\leftmin,\twin[2i])$)\label{line:recurse}\\
    $\twin[i]\gets \min(\twin[2i],\twin[2i+1])$ \label{line:compmin}
  }
  }
\end{algorithm}
}

%% file: wlis.tex
\section{Weighted Longest Increasing Subsequence}\label{sec:wlis}

A nice property of the unweighted LIS problem is that the \dpvalue{} is the same as its rank.
In round $r$, we simply set the \dpvalue{s} of all objects in the frontier as $r$.
This is not true for the weighted LIS (WLIS) problem, and we need additional techniques to handle the weights.
Inspired by \swgs{}, our WLIS algorithm is built on an efficient data structure $\trange$ supporting
2D \defn{\dommax{}} queries: for a set of 2D points $(x_i,y_i)$ each with a \defn{score} $s_i$,
the \dommax{} query $(q_x,q_y)$ asks for the maximum score among all points in its lower-left corner $(-\infty,q_x)\times (-\infty,q_y)$.
\iffullversion{We illustrate a \dommax{} query in \cref{fig:dominantmax} in the appendix.}
\ifconference{A \dommax{} query is demonstrated in the full version of this paper.}
We will use such a data structure to efficiently compute the \dpvalue{s} of all objects.

\input{weighted-lis-code.tex}

We present our WLIS algorithm in \cref{algo:lis-weighted}.
We view each object as a 2D point $(A_i, i)$ with score $\mathdp[i]$,
and use a data structure $\trange$ that supports \dommax{} queries to maintain all such points.
Initially $\mathdp[i]=0$.
We call $A_i$ and $i$ as the x- and y-coordinate of the point, respectively.
Given the input sequence, we first call \cref{algo:lis} to compute the rank of each object and sort them by ranks to find each frontier $\ff_i$.
This can be done by
any parallel sorting with $O(n)$ work and $O(\log^2 n)$ span.
We then process all the frontiers in order.
When processing $\ff_i$, we compute the \dpvalue{s} for all $j\in \ff_i$ in parallel, using
$\mathdp[j]=\max_{j'<j,A_{j'}<A_j} \mathdp[j']$.
This can be done by the \dommax{} query on $\trange$ (\cref{line:dominant}), which
reports the highest score (\dpvalue{}) among all objects in the lower-left corner of the object $j$.
Finally, we update the newly-computed $\mathdp$ values to $\trange$ (\cref{line:update}) as their scores.


The efficiency of this algorithm then relies on the data structure to support \dommax{}.
We will propose two approaches
to achieve practical and theoretical efficiency, respectively.
The first one is similar to \swgs{} and uses range trees, which leads to $O(n\log^2 n)$ work and $\tilde{O}(k)$ span for the WLIS problem.
By plugging in an existing range-tree implementation~\cite{sun2018pam}, we obtain a simple parallel WLIS implementation that significantly outperforms the existing implementation from \swgs{}.
The details of the algorithm are in \cref{sec:wlis-range}, and the performance comparison is in \cref{sec:exp}.
We also propose a new data structure, called the \rangeveb{}, to enable a better work bound ($O(n \log n \log \log n)$ work) for WLIS.
Our idea is to re-design the inner tree in range trees as a \emph{parallel \veb{} tree}.
We elaborate our approach in \cref{sec:rangeveb,sec:veb}.

\subsection{Parallel WLIS based on Range Tree}\label{sec:wlis-range}

We can use a parallel \defn{range tree}~\cite{bentley1979data,sun2019parallel} to answer \dommax queries.
A range tree~\cite{bentley1979data} is a nested binary search tree (BST) where the \emph{outer tree} is an index of the $x$-coordinates of the points.
Each tree node maintains an \emph{inner tree} storing the same set of points in its subtree but keyed on the $y$-coordinates (see \cref{fig:range}).
We can let each inner tree node store the maximum score in its subtree, which enables efficient \dommax{} queries.
In particular, for the outer tree, we can search $(-\infty,q_x)$ on the $x$-coordinates.
This gives $O(\log n)$ relevant subtrees in this range (called the \defn{in-range} subtrees), and $O(\log n)$ relevant nodes connecting them (called the \defn{connecting} nodes).
In \cref{fig:range}, when $q_x=6.5$, the in-range inner trees are the inner trees of points $(2,6)$ and $(5,1)$, since their entire subtrees falls into
range $(-\infty, 6.5)$. The connecting nodes are $(4,5)$ and $(6,4)$, as their x-coordinates are in the range, but only part of their subtrees are in the range.
For each in-range subtree, we further search $(-\infty, q_y)$ in the inner trees to get the maximum score in this range,
and consider it as a candidate for the maximum score.
For each connecting node, we check if its $y$-coordinates are in range $(-\infty, q_y)$, and if so, consider it a candidate.
Finally, we return the maximum score among the selected candidates (both from the in-range subtrees and connecting nodes).
Using the range tree in ~\cite{sun2018pam,sun2019parallel,blelloch2020optimal}, we have the following result for WLIS.
\hide{
	\begin{theorem}(Parallel Range Trees~\cite{sun2018pam,sun2019parallel,blelloch2020optimal})
		\label{thm:rangetree}
		Given a set of $n$ points $(x_i,y_i)$ with score $s_i$ on a 2D plane, parallel range trees support the following operations with $O(n\log n)$ space:
		\begin{itemize}[leftmargin=*,topsep=2pt, partopsep=2pt,itemsep=1pt,parsep=1pt]
			\item constructing the tree in $O(n\log n)$ work and $O(\log^2 n)$ span,
			\item answering the \dommax{} query in $O(\log^2 n)$ work and span,
			\item updating the scores of a set of $m\le n$ points in $O(m\log^2 n)$ work and $O(\log^2 n)$ span.
		\end{itemize}
	\end{theorem}
}


\begin{theorem}
	Using a parallel range tree for the \dommax{} queries, \cref{algo:lis-weighted} computes the weighted LIS of an input sequence $A$ in $O(n\log^2 n)$ work and $O(\lislength\log^2 n)$ span, where $n$ is the length of the input sequence $A$, and $\lislength$ is the LIS length of $A$.
\end{theorem}


\subsection{WLIS Using the \rangeveb{} Tree}
\label{sec:rangeveb}

We can achieve better bounds for WLIS using parallel van Emde Boas (vEB) trees.
Unlike the solution based on parallel range trees, the \veb{}-tree-based solution is highly non-trivial.
Given the sophistication, we describe our solution in two parts.
This section shows how to solve parallel WLIS assuming we have a parallel vEB tree.
Later in \cref{sec:veb}, we will show how to parallelize vEB trees.

We first outline our data structure at a high level.
We refer to our data structure for the \dommax query as the \defn{\rangeveb{} tree}, which is
inspired by the classic range tree as mentioned in \cref{sec:wlis-range}.
The main difference is that the inner trees are replaced by \defn{\monoveb{} trees} (defined below). 
Recall that in \cref{algo:lis-weighted}, the \emph{RangeStruct} implements two functions \textsc{\Dommax} and \textsc{Update}.
We present the pseudocode of \rangeveb{} for these two functions in \cref{algo:rangeveb}, assuming we have parallel functions on \veb{} trees.

Similar to range trees, our \rangeveb{} tree is a two-level nested structure, where the outer tree is indexed by $x$-coordinates, and the inner trees are indexed by $y$-coordinates.
For an outer tree node $v$, we will use $S_v$ to denote the set of points in $v$'s subtree and $T_v$ as the inner tree of $v$.
Like a range tree, the inner tree $T_v$ also corresponds to the set of points $S_v$, but only the \emph{staircase} of $S_v$ (defined below).
Since the y-coordinates are the indexes of the input, which are integers within $n$, we can maintain this staircase in a \veb{} tree.
Recall that the inner tree stores the $y$-coordinates as the key and uses the \dpvalue{s} as the scores.
For two points $p_1=\langle x_1, y_1,\mathdp_1\rangle$ and $p_2=\langle x_2, y_2,\mathdp_2\rangle$, we say $p_1$ \defn{\dominates} $p_2$
if $y_1<y_2$ and $\mathdp_1\geq \mathdp_2$.
For a set of points $S$, the \defn{staircase}
of $S$ is the maximal subset $S'\subseteq S$ such that for any $p\in S'$, $p$ is not \dominated{} by any points in $S$.
In other words, for two input objects $A_i$ and $A_j$ in WLIS, we say $A_i$ \defn{\dominates{}} $A_j$ if $i$ comes before $j$ and has a larger or equal \dpvalue{}.
This also means that no objects will use the \dpvalue{} at $j$ since $A_i$ is strictly better than $A_j$.
Therefore, we ignore such $A_j$ in the inner trees, and refer to such a \veb{} tree maintaining the staircase of a dataset as a \defn{\monoveb{}} tree.
In a \monoveb{} tree, with increasing key ($y_i$), the score (\dpvalue{s}) must also be increasing.
\iffullversion{We show an illustration of the \emph{staircase} in \cref{fig:staircase} in the appendix.}
\ifconference{An illustration of the \emph{staircase} is presented in the full version of the paper.}



Due to monotonicity, the maximum \dpvalue{} in a \monoveb{} tree for all points with $y_i<q_y$ is exactly the score (\dpvalue{}) of $q_y$'s predecessor.
Combining this idea with the \dommax{} query in range trees, we have the \dommax{} function in \cref{algo:rangeveb}.
We will first search the range $(-\infty,q_x)$ in the outer tree for the $x$-coordinates and find all in-range subtrees and connecting nodes.
For each connecting node, we check if their $y$ coordinates are in the queried range, and if so, take their \dpvalue{s} into consideration.
For each in-range inner tree $t_i$, we call $\vebpred{}$ query on the \monoveb{} tree and obtain the score (\dpvalue{}) $\sigma_i$ of this predecessor (\cref{line:sigma}).
As mentioned, the value of $\sigma_i$ is the highest score from this inner tree among all points with an index smaller than $q_y$.
Finally, we take a max of all such results (all $\sigma_i$ and those from connecting nodes),
and the maximum among them is the result of the \dommax{} query (\cref{line:domreturn}).
As the $\vebpred$ function has cost $O(\log\log n)$, a single \dommax{} query costs $O(\log n \log \log n)$ on a \rangeveb{} tree.

\input{veb.tex}

Querying \dommax using a staircase is a known (sequential) algorithmic trick.
However, the challenge is how to update (in parallel) the newly computed \dpvalue{s} in each round (the \textsc{Update} function) to a \rangeveb{} tree.
We first show how to implement \updatefunc{} in \cref{algo:rangeveb} while assuming a parallel \veb{} tree.
We later explain how to parallelize a \veb{} tree in \cref{sec:veb}.

\myparagraph{Step 1. Collecting insertions for inner trees.}
Each point $p\in B$ may need to be added to $O(\log n)$ inner trees,
so we first obtain a list $L_i$ of points to be inserted for each inner tree $t_i$.
This can be done by first marking all points in $B$ in the outer tree $\trange$,
and (in parallel) merging them bottom-up, so that each relevant inner tree collects the relevant points in $B$.
When merging the lists, we keep them sorted by the y-coordinates, the same as the inner trees.

\myparagraph{Step 2. Refining the lists.} Because of the ``staircase'' property, we have to first refine each list $L_i$ to remove points that are not on the staircase.
A point in $L_i[j]$ should be removed if it is \dominated{} by its previous point $L_i[j-1]$,
or if any point in the \monoveb{} tree $t_i$ \dominates{} it.
The latter case can be verified by finding the predecessor $\pi$ of $L_i[j].y$,
and check if $\pi$ has a larger or equal \dpvalue{} than $L_i[j]$.
If so, $L_i[j]$ is \dominated{} by $\pi\in t_i$, so we ignore $L_i[j]$.
After this step, all points in $L[i]$ need to appear on the staircase in $t_i$.

\myparagraph{Step 3. Updating the inner trees.} Finally, for all involved subtrees, we will update the list $L_i$ to $t_i$ in parallel.
Note that some points in $L_i$ may \dominate{} (and thus replace) some existing points in $t_i$.
We will first use a function \domby{} to find all points (denoted as set $R$) in $t_i$ that are \dominated{} by any point in $L_i$.
\iffullversion{An illustration of \domby{} function is presented in \cref{fig:coveredby} in the appendix.}
\ifconference{An illustration of \domby{} function is presented in the full version of this paper.}
We will then use \veb{} batch-deletion to remove all points in $R$ from $t_i$.
Finally, we call \veb{} batch-insertion to insert all points in $L_i$ to $t_i$.

In \cref{sec:veb}, we present the algorithms \domby{}, \batchdelete{} and \batchinsert{} needed by \cref{algo:lis-weighted}, and prove \cref{thm:veb}.
Assuming \cref{thm:veb}, we give the proof of \cref{thm:mainweighted}.
\begin{proof}[Proof of \cref{thm:mainweighted}]
	We first analyze the work. We first show that the \textsc{\Dommax} algorithm in \cref{algo:rangeveb} takes $O(\log n\log \log n)$ work.
	In \cref{algo:rangeveb}, \cref{line:in-range} finds $O(\log n)$ connecting nodes and in-range inner trees, which takes $O(\log n)$ work.
	Then for all $O(\log n)$ in-range inner trees, we perform a \vebpred{} query in parallel, which costs $O(\log \log n)$.
	In total, this gives $O(\log n\log \log n)$ work for \textsc{\Dommax}.
	This means that the total work
	to compute the \dpvalue{s} in \cref{line:dominant} in the entire \cref{algo:lis-weighted} is $O(n\log n\log \log n)$.
	
	We now analyze the total cost of \textsc{Update}.
	In one invocation of \textsc{Update}, we first find all keys for each inner tree $t_i$ that appears in $B$.
	Using the bottom-up merge-based algorithm mentioned in \cref{sec:rangeveb}, each merge costs linear work.
	Similarly, refining a list $L_i$ costs linear work.
	Since each key in $B$ appears in $O(\log n)$ inner tree, the total work to find and refine all $L_i$ is $O(|B|\log n)$ for each batch, and is $O(n\log n)$ for the entire algorithm.
	
	For each subtree, the cost of running \domby{} is asymptotically bounded by \batchdelete{}.
	For \batchdelete{} and \batchinsert{}, note that the bounds in \cref{thm:veb} show that the amortized work to insert or delete a key is $O(\log \log n)$.
	In each inner tree, a key can be inserted at most once and deleted at most once, which gives $O(n\log n\log \log n)$ total work in the entire algorithm.
	
	Finally, the span of each round is $O(\log^2 n)$. In each round, we need to perform the three steps in \cref{sec:rangeveb}.
    The first step requires to find the list of relevant subtrees for each element in the insertion batch $B$.
    For each element $b\in B$, this is performed by first searching $b$ in the outer tree, and then merging them bottom-up,
    so that each node in the outer tree will collect all elements in $B$ that belong to its subtree. There are $O(\log n)$
    levels in the outer tree, and each merge requires $O(\log n)$ span, so this first step requires $O(\log^2 n)$ span.

    Step 2 will process all relevant lists in parallel (at most $n$ of them).
    For each list, it calls $\vebpred$ for each element in each list, and a filter algorithm at the end. The total span is bounded by $O(\log^2 n)$.

    Step 3 requires calling batch insertion and deletion to update all relevant inner trees, and all inner trees can be processed in parallel.
    Based on the analysis above, the span for each batch insertion and deletion is $O(\log n \log \log n)$, which is also bounded by $O(\log^2 n)$.

    Thus, the entire algorithm has span $O(\lislength\log^2 n)$.
    \iffullversion{Finally, as mentioned, we present the details of achieving the stated space bound in Appendix \ref{app:space} by relabeling all points
    in each inner tree.}
	\ifconference{Finally, as mentioned, we present the details of achieving the stated space bound in the full version of this paper by relabeling all points in each inner tree.}
\end{proof}

\myparagraph{Making \rangeveb{} Tree Space-efficient.}
A straightforward implementation of \rangeveb{} tree may require $O(n^2)$ space, as a plain \veb{} tree requires $O(U)$ space.
There are many ways to make \veb{} trees space-efficient ($O(n)$ space when storing $n$ keys);
\iffullversion{we discuss how they can be integrated in \rangeveb{} tree to guarantee $O(n\log n)$ total space in \cref{app:space}.}
\ifconference{we discuss how they can be integrated in \rangeveb{} tree to guarantee $O(n\log n)$ total space in the full version of this paper.}

\hide{
	The high-level idea is the same as \cref{algo:lis}.
	We also build a range tree (see \cref{sec:prelim}) as in \swgs{}, which views each object in the input as a point $(i,A_i)$.
	In round $r$, we also use the tournament tree $T$ to find all \prefixmin{} objects.
	Once such an object $T[i]$ is identified (Line~\ref{line:weighted:begin}--\ref{line:weighted:remove}),
	we compute its DP value by using the \dommax{} query on the range tree $\trange$.
	After that, we mark this object $T[i]$ as ``updated'' in the range tree (along all relevant paths) and remove it from the tournament tree.
	After traversing and processing the frontier, we traverse the range tree and follow the ``updated'' marks to update the newly-computed DP values for all objects in the current frontier.
	It is also easy to output a specific increasing subsequence with the highest weight.
	The best decision at any object $i$ can be returned by the range tree along with the maximum DP value in its lower-left corner.
	Then, using the best decisions of each object, we can iteratively output the optimal subsequence.

	To achieve better work bound for the WLIS problem, we propose a new data structure called \rangeveb{} tree to support the \dommax{} query more efficiently.
	Our idea is to replace the inner trees in range trees to be \veb{} trees, but only maintain \emph{staircase}~\alert{\cite{}} of the elements based on their y-coordinates and DP values.
	In \cref{algo:rangeveb}, we present the pseudocode of \rangeveb{} for the functions \textsc{RangeMax} and \textsc{Update} needed in \cref{algo:lis-weighted}.
	We start with defining some concepts. For two pairs $p_1=\langle a_1,b_1\rangle$ and $p_2=\langle a_2,b_2\rangle$, we say $p_1$ \defn{\dominates} $p_2$
	if $a_1<a_2$ and $b_1\geq b_2$. For a sequence of pairs $S$, the \defn{staircase} of $S$ is the maximal subset $S'\subseteq S$ such that for any $p\in S'$, $p$ is not \dominated{} by any points in $S$. 
	In this paper, we propose parallel batch insertion and deletion algorithms for \emph{general} \veb{} trees and use them to support \updatefunc{}
	on \monoveb{} trees.
	We first describe how to implement \updatefunc{} (shown in \cref{algo:rangeveb}) to update the \dpvalue{s} of a list of points $B=\{\langle x_i,y_i,\mathdp_i\rangle\}$
	using batch insertions and deletions of \veb{} trees.
	
	In this paper, we propose parallel batch insertion and deletion algorithms for \emph{general} \veb{} trees and use them to support \updatefunc{}
	on \monoveb{} trees.
	We first describe how to implement \updatefunc{} (shown in \cref{algo:rangeveb}) to update the \dpvalue{s} of a list of points $B=\{\langle x_i,y_i,\mathdp_i\rangle\}$
	using batch insertions and deletions of \veb{} trees.

}

%% file: weighted-lis-code.tex
\newcommand{\dparray}{\mathit{dp}}
\newcommand{\cur}{\mathit{count}}
\newcommand{\mathfrontier}{\mathit{frontier}}
\hide{
\begin{algorithm}[t]
\fontsize{9pt}{10pt}\selectfont
\caption{\small The parallel weighted LIS algorithm\label{algo:lis-weighted}}
\SetKwFor{parForEach}{parallel\_for\_each}{do}{endfor}
\KwIn{A sequence $A_{1..n}$ with weit comparison function $<$. Object $i$ has weight $w[i]$}
\KwOut{The DP values $\dparray[1..n]$ for each object $A_i$.}
    \vspace{0.5em}
\SetKwProg{MyStruct}{Struct}{}{end}
\SetKwProg{RangeTree}{RangeTree$\langle$Point$\rangle$}{ contains}{end}
\SetKwProg{MyFunc}{Function}{}{end}
\newcommand{\Int}{\KwSty{int}}
\newcommand{\A}{\KwSty{A}}
\DontPrintSemicolon
\MyStruct{Point}{
  \Int{} $x,y$ \tcp*{x = index, y = $A_x$.}
  \Int{} $\mathdp$ \tcp*{The DP value of $A_x$.}
}
Point $p[1..n]$\\
\Int{} $\dparray[1..n]$ \tcp*[f]{$\dparray[i]$: the DP value of $A_i$.}

\MyStruct{RangeTree$\langle$Point$\rangle$}{
  Contains 2D Points $(x_i,y_i)$ with weights \\
  Supports RangeMax$(p,q)$ query: returns the maximum weight (the $\mathdp{}$ value) among all points $(x_i,y_i)$ where $x_i<p$ and $y_i<q$
}

RangeTree $\trange$\\

\smallskip

%

\Int{} $T[1..(2n-1)]$ \tcp*[f]{$T$: the same as Algorithm \ref{algo:lis}.}\\

\smallskip

\lparForEach{$A_i\in A$}{
  $p[i] = \langle i, A_i,0\rangle$
}
Initialize the tournament tree $T$ \\
Construct $\trange$ from $p[\cdot]$\\
\smallskip
\While (\tcp*[f]{$T$ is not empty.}){$T[1]\ne +\infty$} {
  \processfrontier()  \\
  Update all DP values by $\mathdp[\cdot]$ in the range tree $\trange$ marked as ``updated'', then clear the mark\\
}

\medskip

\MyFunc{\upshape\mf{\processfrontier}() }{
  \prefixminfunc(1, $+\infty$)\\
}

\medskip
\tcp{Find and remove objects $x$ in subtree rooted at $T[i]$, s.t.: 1) $x\le$ all object before it, and 2) $x\le\leftmin$.
Then compute the DP values of $x$ from the range tree.
}
\MyFunc{\upshape\mf{PrefixMin}(\Int{} $i$, \Int{} $\leftmin$) }{
  \lIf {$T[i] > \leftmin$} {\Return} 
  \If(\tcp*[f]{Found a leaf node in the frontier.}){$i\ge n$\label{line:weighted:begin}}{
  \tcp{Compute the DP value by a 2D range-max.}
  $\mathdp[i]\gets \trange$.2DMax$(i-n+1, A_{i-n+1})+w[i-n+1]$ \\
  Mark object $i-n+1$ as ``updated'' in $\trange$\\
  $T[i]\gets +\infty$ \label{line:weighted:remove}\tcp*{Remove the object from $T$.}
  } \Else(\tcp*[f]{Deal with an internal node. }){ 
    \textsc{PrefixMin}($2i$, $\leftmin$) $||$ \textsc{PrefixMin}($2i+1$, $\min(\leftmin,T[2i])$)\tcp*[f]{Run in parallel.}\\
    $T[i]\gets \min(T[2i],T[2i+1])$  
  }

}

\end{algorithm}
}

\begin{algorithm}[t]
\fontsize{8pt}{8.5pt}\selectfont
\caption{The parallel weighted LIS algorithm\label{algo:lis-weighted}}
\SetKwFor{parForEach}{parallel-foreach}{do}{endfor}
\KwIn{A sequence $A_{1..n}$. Object $A_i$ has weight $w_i$}
\KwOut{The DP values $\dparray[1..n]$ for each object $A_i$.}
    \vspace{0.5em}
\SetKwProg{MyStruct}{Struct}{}{end}
\SetKwProg{MyFunc}{Function}{}{end}
\newcommand{\Int}{\KwSty{int}}
\newcommand{\A}{\KwSty{A}}
\DontPrintSemicolon
\MyStruct{Point}{
  \Int{} $x,y$ \tcp*{y = index, x = $A_y$.}
  \Int{} $\mathdp$ \tcp*{The DP value of $A_y$, used as the score of the point.}
}
Point $p[1..n]$\\
\Int{} $\dparray[1..n]$ \tcp*[f]{$\dparray[i]$: the DP value of $A_i$.}

\MyStruct{RangeStruct$\langle$Point$\rangle$}{
  Stores points $\langle x_i,y_i,\mathdp_i\rangle$ with coordinate $(x_i,y_i)$ and score~$\mathdp_i$\\
  Supports \textsc{\Dommax}$(p,q)$: return the maximum score (the $\mathdp[\cdot]$ value) among all points $(x_i,y_i)$ where $x_i<p$ and $y_i<q$\\
  Supports \textsc{Update}$(B)$, where $B=\{\langle x_i, y_i, \mathdp_i \rangle\}$ is a batch of points: update the score of each point $(x_i,y_i)$ to $\mathdp_i$\\
}

\textit{RangeStruct} $\trange$ \tcp*[f]{Any data structure that supports \textsc{\Dommax}}\\

\smallskip

%


Run \cref{algo:lis}.  Sort the $\mathit{rank}$ array and get all the $k$ frontiers $\ranklist_{1..k}$. $\ranklist_i$ contains the indexes of all objects with rank $i$.\label{line:unweighted}

\smallskip

\lparForEach{$A_i\in A$}{
  $p[i] = \langle A_i, i,0\rangle$
}
Construct $\trange$ from $p[\cdot]$\\
\smallskip
\For{$i\gets 1$ to $k$\label{line:for-batch}} {
  \parForEach(\tcp*[f]{$A_j$ is an object with rank $i$}){$j \in \ranklist_i$} {
    $\mathdp[j]\gets \trange.$\textsc{\Dommax}$(A_j, j)+w_j$\label{line:dominant}\\
  }
  $B\gets \{\langle A_j, j, \mathdp[j]\rangle : j\in \ranklist_i\}$\\
  $\trange.$\textsc{Update}$(B)$\label{line:update}
}

\Return{$\mathdp[\cdot]$}
\end{algorithm}

%% file: veb.tex
\begin{algorithm}[t]
\fontsize{8pt}{9pt}\selectfont
\caption{\small The parallel RangeStruct using \rangeveb{} trees\label{algo:rangeveb}}
\SetKwFor{parForEach}{parallel-foreach}{do}{endfor}
\SetKwProg{MyStruct}{Struct}{}{end}
\SetKwProg{MyFunc}{Function}{}{end}
\newcommand{\Int}{\KwSty{int}}
\newcommand{\A}{\KwSty{A}}
\SetKw{MIN}{min}
\SetKw{MAX}{max}
\DontPrintSemicolon
Structures \textit{Point} and \textit{RangeStruct} are defined in \cref{algo:lis-weighted}\\

\MyFunc{\upshape\textsc{\Dommax$(q_x,q_y)$}}{
  In the \rangeveb, find the range of $(-\infty, q_x)$, and let $S_{\mathit{node}}$ be the set of connecting nodes and $S_{\mathit{tree}}$ be the set of in-range inner (\monoveb{}) trees \label{line:in-range}\\
  \tcp{For each in-range inner tree, find the max score up to coordinate $q_y$}
  \parForEach{$t_i\in S_{\mathit{tree}}$}{
    $\langle \cdot, \cdot, \sigma_i \rangle \gets \vebpred(t_i,q_y)$\tcp*[f]{$\sigma_i$ is the score of $q_y$'s predecessor}\label{line:sigma}
  }
  \tcp{For connecting nodes, check if the y-coordinates are smaller than $q$ and get the maximum score for such points}
  \ForEach {$\langle x,y,\mathdp \rangle \in S_{\mathit{node}}$ s.t. $y<q_y$}{
    $\sigma' = \MAX(\sigma', \mathdp)$\\
  }
  \Return{$\MAX(\sigma', \MAX_i \{\sigma_i\})$}\label{line:domreturn}
}
\smallskip

\MyFunc(\tcp*[f]{$B$ is a list of points $\{\langle x_i,y_i, \mathdp_i \rangle\}$}){\upshape\textsc{Update}$(B)$} {
  Update $\langle x_i,y_i, \mathdp_i \rangle$ in the outer \rangeveb tree\\
  For each relevant inner (\monoveb{}) tree $t_i$, gather a list of points $L_i\subseteq B$ to be added to $t_i$\\
  Points in $L_i$ are sorted by the y-coordinates\\
  Let $S_{\mathit{tree}}$ be the set of inner trees $t_i$ that need new insertions\\
  \tcp{Refine $L_i$: Remove $L_i[j]$ if any other point \dominates{} it}
  For each list $L_i$, remove $L_i[j]$ if\\
  ~~ $\bullet~~~~\exists~l<j$, s.t. $L_i[j].\mathdp < L_i[l].\mathdp$, or\\
  ~~ $\bullet~~~~\pi.\mathdp{}\ge L_i[j].\mathdp{}$, where $\pi = \vebpred(t_i, L_i[j].y)$ is $L_i[j]$'s predecessor in the corresponding \monoveb{} tree $t_i$ \\
  \parForEach{$t_i\in S_{\mathit{tree}}$} {
    \tcp{Find elements in $t_i$ that are \dominated{} by points in $L_i$}
    $R \gets t_i.$\textsc{\domby}$(L_i)$\\
    $t_i.$\textsc{BatchDelete}$(R)$\tcp*[f]{Delete points in $R$}\\
    $t_i$.\textsc{BatchInsert}$(L_i)$\tcp*[f]{Insert points in $L_i$}\\
  }
}
\end{algorithm} 

%% file: parallel-veb.tex
\section{Parallel van Emde Boas Trees}\label{sec:veb}

\input{notationtab.tex}

\begin{figure*}[t]
	\vspace{-1.5em}
	\begin{minipage}[c]{0.28\textwidth}
		\includegraphics[width=0.9\textwidth]{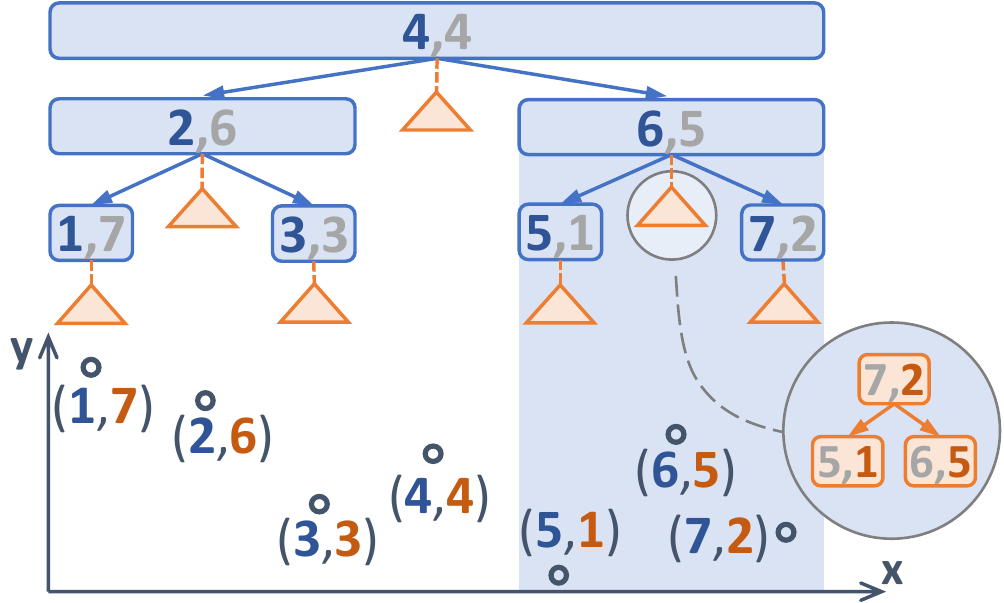}
		\caption{\small \textbf{An illustration of a 2D range tree}.  Outer tree is indexed by $x$ (blue) and inner trees are indexed by $y$ (red).
			\label{fig:range}}
	\end{minipage}\hfill
	\begin{minipage}[c]{0.7\textwidth}
		\includegraphics[width=\textwidth]{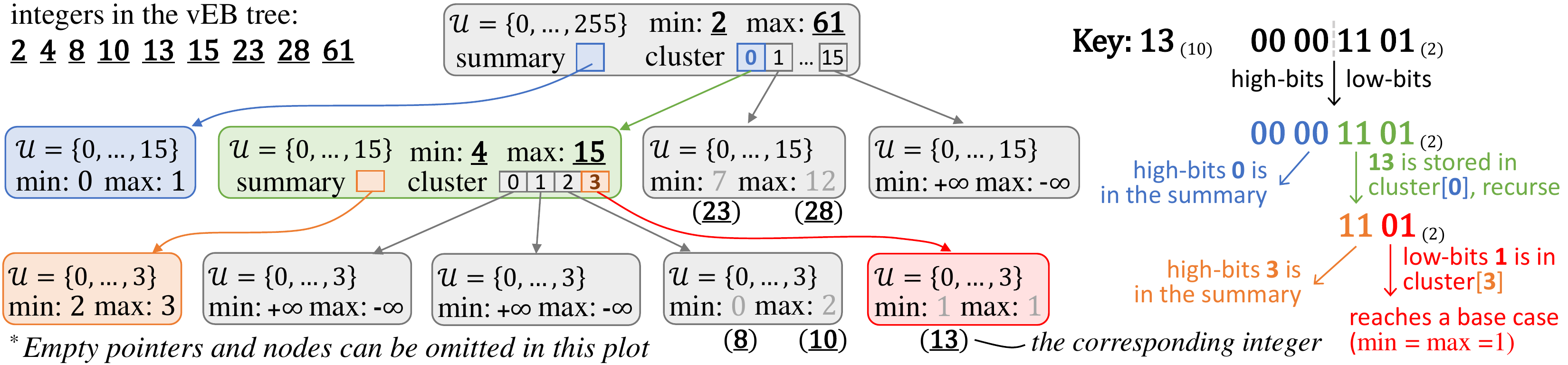}
		\caption{\small \textbf{An example \veb tree with $\boldsymbol{U=256}$ and a demonstration on how 13 is stored}.
			{The \veb tree contains the set of keys $\{2,4,8,10,13,15,23,28,61\}$.\label{fig:vEBdemo}}}
	\end{minipage}
	\vspace{-1.8em}
\end{figure*}

The van Emde Boas (vEB) tree~\cite{van1977preserving} is a famous data structure that implements the ADTs of priority queues and ordered sets and maps for integer keys and is introduced in textbooks (e.g.,~\cite{CLRS}).
For integer keys in range 0 to $U$, single-point updates and queries cost $O(\log \log U)$, better than the $O(\log n)$ cost for BSTs or binary heaps.
We review \veb{} trees in \cref{sec:veb-review}.

However, unlike BSTs or binary heaps that have many parallel versions ~\cite{Blelloch1998,blelloch2016just,blelloch2022joinable,sun2018pam,blelloch2020optimal,wang2020parallel,dong2021efficient,lim2019optimal,williams2021engineering,akhremtsev2016fast,zhou2019practical}, we are unaware of any parallel vEB trees.
Even the sequential \veb tree is complicated (compared to most BSTs and heaps), and the invariants are maintained sophisticatedly to guarantee doubly-logarithmic cost.
Such complication adds to the difficulty of parallelizing updates (insertions and deletions) on \veb trees.
Meanwhile, for queries, we note that \veb trees do not directly support range-related queries---when using \veb trees for ordered sets and maps, many applications heavily rely on repeatedly calling successors and/or predecessors, which is inherently sequential.
Hence, we need to carefully redesign the \veb tree to achieve parallelism.
In this section, we first review the sequential \veb{} tree and then present our parallel \veb{} tree to support the functions needed in \cref{algo:rangeveb}.

\subsection{Review of the Sequential \veb{} Tree}\label{sec:veb-review}

A van Emde Boas (\veb) tree~\cite{van1977preserving} is a search tree structure with keys from a universe $\univ$, which are integers from 0 to $U-1$.
We usually assume the keys are $w$-bit integers (i.e., $U=2^w$).
A classic \veb tree supports insertion, deletion, lookup, reporting the min/max key in the tree,
reporting the predecessor (\vebpred{}) and successor (\vebsucc{}) of a key, all in $O(\log \log U)$ work.
Other queries can be implemented using these functions.
For instance, reporting all keys in a range can be implemented by repeatedly calling \vebsucc{} in $O(m\log \log U)$ work,
where $m$ is the output size. 


A \veb{} tree stores a key using its binary bits as its index.
We use $\vt$ to denote a \veb{} tree, as well as the set of keys in this \veb{} tree.
We present the notation for \veb{} trees in \cref{tab:notation}, and show an illustration of \veb{} trees in \cref{fig:vEBdemo}.
We use 13 as an example to show how a key is decomposed and stored in the tree nodes.
A \veb{} tree is a quadruple $(\summary,\cluster[\cdot],\min,\max)$.
$\vt.\mmin$ and $\vt.\mmax$ store the minimum and maximum keys in the tree.
When $\vt$ is empty, we set $\vt.\mmin=+\infty$ and $\vt.\mmax=-\infty$.
For the rest of the keys (other than $\mmin/\mmax$),
their \defn{high-bits} (the first $\lceil w/2 \rceil$ bits) are maintained recursively in a \veb{} tree, noted as ${\vt.\summary}$.
In \cref{fig:vEBdemo}, the high-bits are the first 4 bits, and there are two different unique high-bits (0 and 1). They
are maintained recursively in a \veb (sub-)tree ${\vt.\summary}$ (the blue box).
For each unique high-bit, the relevant \defn{low-bits} (the last $\lfloor w/2 \rfloor$ bits) are also organized as a \veb{} (sub-)tree recursively.
In particular, the low-bits that belong to high-bit $h$ are stored in a \veb{} tree ${\vt.\cluster[h]}$.
In \cref{fig:vEBdemo}, five keys in $\vt$ have high-bit 0 (4, 8, 10, 13, and 15). 
They are maintained in a \veb (sub-)tree as $\vt.\cluster[0]$ (the green box and everything below).
Each subtree ($\summary$ and all $\cluster[\cdot]$) has universe size $O(\sqrt{U})$ (about $w/2$ bits).
This guarantees traversal from the root to every leaf in $O(\log \log U)$ hops.
Note that the $\mmin/\mmax$ values of a \veb{} tree are not stored again in the summary or clusters.
For example, in \cref{fig:vEBdemo}, at the root, $\vt.\mmin=2$, and thus 2 is not stored again in $\vt.\cluster[0]$.
Such design is crucial to guarantee doubly logarithmic work for \insertop, \deleteop, \vebpred, and \vebsucc{}.

Note that although we use ``low/high-bits'' in the descriptions, algorithms on \veb{} trees can use simple RAM operations to extract the corresponding bits,
without any bit manipulation, as long as the universe size for each subtree is known.
Due to the page limit, we refer the audience to the textbook~\cite{CLRS} for more details about sequential \veb{} tree algorithms.

\hide{
	A \veb{} tree stores a key using its binary bits as its index.
	\cref{fig:vEBdemo} shows an example of a \veb{} tree with keys in a universe of 8-bit integers.
	A \veb{} tree is a quadruple $(\summary,\cluster[\cdot],\min,\max)$, explained in more details below.
	We use $\vt$ to denote a \veb{} tree (the quadruple), as well as the set of keys in this \veb{} tree.
	A \veb{} $\vt$ tree organizes the \defn{high-bits} (the first $\lceil w/2 \rceil$ bits) of the keys recursively in a \veb{} tree, noted as ${\vt.\summary}$.
	In the \veb{} tree in \cref{fig:vEBdemo}, the high-bits are the first 4 bits, and there are two different unique high-bits (0 and 1). They
	are maintained recursively in a \veb (sub-)tree ${\vt.\summary}$ (the blue box).
	For each unique high-bit, the relevant \defn{low-bits} (the last $\lfloor w/2 \rfloor$ bits) are also organized as a \veb{} (sub-)tree recursively.
	In particular, the low-bits that belong to high-bit $h$ will be stored in a \veb{} tree ${\vt.\cluster[h]}$.
	For example, in \cref{fig:vEBdemo}, six keys in $\vt$ have high-bit 0 (2, 4, 8, 10, 13, and 15). \zheqi{shall we include 2 (min)?}
	They are maintained in a \veb (sub-)tree as $\vt.\cluster[0]$ (the green box and everything below).
	Each subtree ($\summary$ and all $\cluster[\cdot]$) has universe size $O(\sqrt{U})$ (about $w/2$ bits).
	This guarantees traversal from the root to every leaf in $O(\log \log U)$ hops.
	Notably, a \veb tree directly stores the minimum and maximum keys (${\vt.\mmin}$ and ${\vt.\mmax}$) in it, and these two keys are not considered in the summary or clusters, specially, an empty \veb{} tree would assign $\vt.\mmin$ to $+\infty$ and $\vt.\mmax$ to $-\infty$.
	Such augmentation is crucial to guarantee doubly logarithmic work for \insertop, \deleteop, \vebpred, and \vebsucc{}.
	When $\vt$ is empty, we set $\vt.\mmin=+\infty$ and $\vt.\mmax=-\infty$.
	We present the notation we use for \veb{} tree in \cref{tab:notation}, and show an illustration of \veb{} trees in \cref{fig:vEBdemo}.
	We use 13 as an example to show how a key is decomposed into $\log \log U$ levels and eventually stored as a $\mmin$ and/or $\mmax$ in a tree node.
	Due to the page limit, we refer the audience to the textbook~\cite{CLRS} for more details about sequential \veb{} tree algorithms.
}



\subsection{Our New Results}\label{sec:veb-results}
We summarize our results on parallel \veb{} tree in \cref{thm:veb}.
Both batch insertion/deletion and range reporting are work-efficient---the work is the same as performing them on a sequential \veb{} tree.
In \cref{algo:lis-weighted}, the key range $U=n$.
Using the \rangequery{} query, we can implement \domby{} in \cref{algo:rangeveb} in $O(m'\log\log n)$ work and polylogarithmic span, where $m'$ is the number
of objects returned.

Similar to the sequential \veb tree, batch-insertion is relatively straightforward among the parallel operations.
We present the algorithm and analysis in \cref{sec:batch-insert}.
Batch-deletion is more challenging, as once ${\vt.\mmin}$ or ${\vt.\mmax}$ is deleted, we need to replace it with a proper key $k'$ stored in the subtree of $\vt$.
However, when finding the replacement $k'$, we need to avoid the values in the deletion batch $B$, and take extra care to handle the case when $k'$ is the $\mmin/\mmax$ of a cluster.
We propose a novel technique: \defn{Survivor Mapping} (see \cref{def:survivormapping}) to resolve this challenge.
The batch-deletion algorithm is illustrated in \cref{sec:batch-delete},\ifconference{and analysis in the full version of this paper.}\iffullversion{and analysis in \cref{app:deleteproof}.}For range queries, we need to avoid the iterative solution (repeatedly calling \vebsucc{}) since it is inherently sequential.
Our high-level idea is to divide-and-conquer in parallel, but
uses delicate amortization techniques to bound the extra work. 
Due to the page limit, we summarize the high-level idea of \rangequery{} and \domby{} in \cref{sec:rangequery} and provide the details in \cref{app:range,app:dominatedby}.

\input{batchinsert.tex}

\input{batchdelete.tex}

\subsubsection{Range Query.}\label{sec:rangequery}
Sequentially, the range query on \veb trees is supported by repeatedly calling \vebsucc ~from the start of the range until the end of the range is reached.
However, such a solution is inherently sequential.
Although we can find the start and end of the range directly in the tree, reporting all keys in the range to an array in parallel is difficult---\veb{} trees cannot maintain subtree sizes efficiently, so we cannot decide the size of the output array to be assigned to each subtree.
In this paper, we design novel parallel algorithms for range queries on \veb{} trees and use them to implement \domby{} in \cref{algo:rangeveb}.

To achieve parallelism, our high-level idea is to use divide-and-conquer to split the range in the middle and search the two subranges in parallel.
Even though the partition can be unbalanced, we can still bound the recursion depth while amortizing the work for partitioning to the operations in the sequential execution.
We collect results first in a binary tree and then flatten the tree into a consecutive array.
Putting all pieces together, our range query has optimal work (same as a sequential algorithm) and polylogarithmic span (see \cref{thm:veb:rangequery}). 
On top of that, we show how to implement \domby{}, which also requires amortization techniques. The details are provided in \cref{app:range,app:dominatedby}, which we believe are very interesting algorithmically.



\hide{
	\begin{table}
		\small
		
		\begin{tabular}{ll}
			\hline
			\insertop{}$(x)$,\deleteop{}$(x)$,\minop{}$()$, \maxop{}$()$ & $O(\log\log U)$ \\
			\memberop{}$(x)$, \predop{}$(x)$, \succop{}$(x)$ & (same as sequential) \\
			\hline
			\batchinsert{}$(B)$,\batchdelete{}$(B)$ & $O(|B|\log \log U)$\\
			\hline
		\end{tabular}
		\caption{\small Bounds for parallel \veb{} tree. $B\subseteq \mathcal{U}$ is a sorted batch of elements in $[1..U]$. }\label{tab:vebbounds}
	\end{table}
}





%% file: notationtab.tex
\hide{
\begin{table}
  \small
  \hide{\centering
  \begin{tabular}{ll}
    \multicolumn{2}{l}{\textbf{Longest Increasing Subsequence (LIS) or Weighted LIS:}}\\
    \hline
    $A_{1..n}$ & The input sequence \\
    $k$ & The LIS length of $A$ \\
    $\rank(x)$ & The LIS length ending at $x \in A$ \\
    $\mathdp[i]$ & The \dpvalue{} of $A[i]$. It is the LIS length ending at $A[i]$ for LIS, \\
    & It is the largest weighted sum of increasing subsequences ending at \\
    &$A[i]$ for WLIS.\\
    $\twin$ & The winning tree (used in \cref{algo:lis})\\
    $\trange$ & The range search structure (used in \cref{algo:lis-weighted})\\
    \hline
    \multicolumn{2}{l}{\textbf{\vebfull{} Tree:}}\\
    $\tveb$ & A \veb{} (sub-)tree\\
    $\tveb.\mmin$ & The min value in \veb{} tree $\tveb$ ($\tveb.\mmax$ is defined similarly)\\
    $\vebpred(\tveb, x)$ & Find the predecessor of $x$ in \veb{} tree $\tveb$\\
    $\vebsucc(\tveb, x)$ & Find the successor of $x$ in \veb{} tree $\tveb$\\

  \end{tabular}
  }
\flushleft \textbf{\underline{LIS and Weighted LIS Algorithms:}}
  \begin{description}[topsep=0in, labelwidth=.4in,leftmargin=.5in,itemindent=0in]
    \item[$\boldsymbol{A_{1..n}}$]: Input sequence
    \item [$\boldsymbol{k}$]:  The LIS length of $A$
    \item[$\boldsymbol{\rank(x)}$]: The LIS length ending at $x \in A$
    \item[$\boldsymbol{\mathit{dp}}\boldsymbol{{[i]}}$]: The \dpvalue{} of $A_i$. It is the LIS length ending at $A_i$ for LIS. It is the largest weighted sum of an increasing subsequence ending at $A_i$ for WLIS.
    \item[$\boldsymbol{\twin}$]: The winning tree (used in \cref{algo:lis})
    \item[$\boldsymbol{\trange}$]: The range search structure (used in \cref{algo:lis-weighted})
    \item[$\boldsymbol{B}$]: A batch of elements (points)
  \end{description}
\flushleft \textbf{\underline{\vebfull{} Tree Algorithms:}}
  \begin{description}[topsep=0in, labelwidth=.65in, leftmargin=.75in,itemindent=0in]
    \item[$\boldsymbol{x}$]: A $w$-bit integer $x$ from universe $\univ$, where $\univ=[2^w]$ and $|\univ|=U$.
    \item[$\boldsymbol{\high(x)}$]: high-bit of $x$, equals to $\lfloor x/2^{\lceil w/2\rceil}\rfloor$
    \item[$\boldsymbol{\low(x)}$]: low-bit of $x$, equals to $(x\mod 2^{\lceil w/2\rceil})$
    \item[$\boldsymbol{\idx(h,l)}$]: The integer by concatenating high-bit $h$ and low-bit $l$, equals to $(h\cdot 2^{\lceil w/2\rceil}+l)$
    \item[$\boldsymbol{\tveb}$]: A \veb{} (sub-)tree
    \item [$\boldsymbol{\tveb.\mmin}$ ($\boldsymbol{\tveb.\mmax}$)]:  The min (max) value in \veb{} tree $\tveb$
    \item[$\boldsymbol{\vebpred(\tveb, x)}$]: Find the predecessor of $x$ in \veb{} tree $\tveb$
    \item[$\boldsymbol{\vebsucc(\tveb, x)}$]: Find the successor of $x$ in \veb{} tree $\tveb$
    \item[$\boldsymbol{\tveb.\summary}$]: The set of high-bits in \veb{} tree $\tveb$
    \item[$\boldsymbol{\tveb.\cluster[h]}$]: The subtree of $\tveb$ with high-bit $h$
    \item[$(^*)\boldsymbol{\mathcal{P}_{\tveb{},B}(x)}$]: The survival predecessor of $x\in B$ in \veb{} tree $\tveb$ (used in \cref{vebdelete}). $\mathcal{P}(x)=\max\{ y: y\in \tveb \setminus B, y<x\}$.
    \item[$(^*)\boldsymbol{\mathcal{S}_{\tveb{},B}(x)}$]: The survival successor of $x\in B$ in \veb{} tree $\tveb$ (used in \cref{vebdelete}). $\mathcal{S}(x)=\min\{ y: y\in \tveb \setminus B, y>x\}$.
  \end{description}
  \caption{Notation. $(^*)$: We drop the subscript with clear context. }\label{tab:notation}
\end{table}
}

\begin{table}
  \small
  \begin{description}[topsep=0in, labelwidth=.6in, leftmargin=.7in,itemindent=0in,font=\normalfont]
    \item[$\aboldsymbol{x}$]: A $w$-bit integer $x$ from universe $\univ$, where $\univ=[0,2^w)$
    \item[$\aboldsymbol{\high(x)}$]: high-bit of $x$, equals to $\lfloor x/2^{\lceil w/2\rceil}\rfloor$
    \item[$\aboldsymbol{\low(x)}$]: low-bit of $x$, equals to $(x\mod 2^{\lceil w/2\rceil})$
    \item[$\aboldsymbol{\idx(h,l)}$]: The integer by concatenating high-bit $h$ and low-bit $l$
    \item[$\aboldsymbol{\tveb}$]: A \veb{} (sub-)tree / the set of keys in this \veb{} tree
    \item[$\aboldsymbol{\tveb.\mmin}$ ($\aboldsymbol{\tveb.\mmax}$)]:  The min (max) value in \veb{} tree $\tveb$
    \item[$\aboldsymbol{\vebpred(\tveb, x)}$]: Find the predecessor of $x$ in \veb{} tree $\tveb$
    \item[$\aboldsymbol{\vebsucc(\tveb, x)}$]: Find the successor of $x$ in \veb{} tree $\tveb$
    \item[$\aboldsymbol{\tveb.\summary}$]: The set of high-bits in \veb{} tree $\tveb$
    \item[$\aboldsymbol{\tveb.\cluster[h]}$]: The subtree of $\tveb$ with high-bit $h$
    \item[$(^*)\aboldsymbol{\mathcal{P}_{\tveb{},B}(x)}$]: The survival predecessor of $x\in B$ in \veb{} tree $\tveb$ (used in \cref{vebdelete}). $\mathcal{P}(x)=\max\{ y: y\in \tveb \setminus B, y<x\}$.
    \item[$(^*)\aboldsymbol{\mathcal{S}_{\tveb{},B}(x)}$]: The survival successor of $x\in B$ in \veb{} tree $\tveb$ (used in \cref{vebdelete}). $\mathcal{S}(x)=\min\{ y: y\in \tveb \setminus B, y>x\}$.
  \end{description}
  \caption{\small \textbf{Notation for \veb trees.} $(^*)$: We drop the subscript with clear context. }\label{tab:notation}
  \vspace{-.5em}
\end{table} 

%% file: batchinsert.tex
\vspace{-.2em}
\subsubsection{Batch Insertion.}\label{sec:batch-insert} We show our batch-insertion algorithm in \cref{vebinsert}, which inserts a sorted batch $B\subseteq \univ$ into $\vt$ in parallel.
Here we assume the keys in $B$ are not in $\vt$; otherwise, we can simply look up the keys in $\vt$ and filter out those in $\vt$ already.
To achieve parallelism, we need to appropriately handle the high-bits and low-bits, both in parallel, as well as taking extra care to maintain the $\mmin/\mmax$ values.


We first set the $\mmin/\mmax$ values at $\vt$ (\cref{backupMinMax}--\ref{updatedInsert}).
If $B.\mmin<\vt.\mmin$, we update $\vt.\mmin$ by swapping it with $B.\mmin$ (\cref{nonEmptyMin}); similarly we update $\vt.\mmax$ (\cref{nonEmptyMax}).
Since we need the batch $B$ sorted when adding $\vt.\mmin$ and/or $\vt.\mmax$ back to $B$ (\cref{updatedInsert}), we need to insert them to the correct position, causing $O(m)$ work.
If $B$ is not empty, we will insert the keys in $B$ to $\vt.\summary$ and $\vt.\cluster$. 
We first find the new high-bits (not yet in $\vt.\summary$) from keys in $B$, and denote them as 
$H$ (\cref{extractNewhigh})
This step can be done by a parallel filter.
For each new high-bit $h\in H$, we select the smallest key with high-bit $h$ and put them in an array $B'$ (\cref{line:bprime}).
The new subtrees in $\vt.\cluster[h]$ are initialized by inserting the smallest low-bit $\{low(x)| x\in B', \high(x)=h \}$ in parallel (\cref{emptyInsertForLoop}--\ref{line:emptyMinMax2}), after which all subtrees $\vt.\cluster[h],h\in H$ are certified to be non-empty. The remaining new low-bits $L[h]$ are gathered by the corresponding high-bits $h\in H$ (\cref{unify}). Finally, new high-bits $H$ and each new low-bits $L[h],h\in H$ are inserted into the tree in parallel recursively (\cref{line:parallelInsert}--\cref{ins-low}).



\hide{
	If $\vt$ is empty, as \cref{begin-minmax} tests, we set the $\vt.\mmin$ and $\vt.\mmax$ to batch's minimum and maximum and exclude these elements from batch; otherwise, \cref{vnotempty} reveals that $\vt$ is non-empty and $\vt.\mmin$ needs to be updated if $B.\mmin$ is smaller than $\vt.\mmin$, we then exchange these two elements, so as when $B.\mmax$ greater than $\vt.\mmax$.
	Note that in order to keep batch elements sorted, adding substituted min/max in \cref{updatedInsert} requires additional $O(m)$ work.
	To insert the remaining elements in $B$ when we reach \cref{beginInsert}, we first pack all elements' high-bits in $H$ and use a filter function to sieve those unique high-bits patterns by testing whether an element has the same high-bits as its left neighbor within $O(m)$ work and $O(\log m)$ span, after which we encapsulate all low-bits for elements with high-bits $h\in H$ within the array cell $L[h]$, as illustrated in \cref{unify}.
	\cref{ins-high} then inserts those high-bits $H'\subseteq H$ that are not contained in current tree summary recursively.
	For every high-bits pattern $h\in H$, the low-bits in corresponding cell $L[h]$ are inserted in parallel, as \cref{spawn-high} illustrates.
	Specially, if $\vt.\cluster[h]$ is empty, as \cref{test-alone} tests, we insert the minimum low-bits in $L[h]$ into that cluster separately using constant work in \cref{ins-alone} and remove it from the low-bits cell $L[h]$ in \cref{remove-min}. When we reach \cref{non-empty-low}, every cluster is guaranteed to be non-empty, \cref{ins-low} then inserts the remaining low-bits in each cell $L[h]$ recursively.
}

\input{batchInsertVEB}

The correctness of the algorithm can be shown by checking that all $\mmin$/$\mmax$ values for each node are set up correctly. Next, we analyze the cost bounds of \cref{vebinsert} in \cref{insertTheorem}.


\begin{theorem}\label{insertTheorem}
	Inserting a batch of sorted keys into a \veb{} tree can be finished in $O(m\log\log U)$ work and 
	$O(\log U)$ span, where $m$ is batch size and $U=|\univ|$ is the universe size.
\end{theorem}

\begin{proof}
	
	Let $W(u,m)$ and $S(u,m)$ be the work and span of \batchinsert{} on a batch of size $m$ and \veb{} tree with universe size $u$.
	In each invocation of \batchinsert{}, we need to restore the $\mmin/\mmax$ values,
	find the high-bits in $H$, initialize the clusters for the new high-bits, and gather the low-bits for each cluster.
	
	
	All these operations cost $O(m)$ work and $O(\log m)=O(\log u)$ span.
	Then the algorithm makes at most $\sqrt{u}+1$ recursive calls, each dealing with a universe size $\sqrt{u}$.
	Hence, we have the following recurrence for work and span:
	\vspace{-.5em}
	\begin{align}
		\textstyle  W(u,m)&=\textstyle\sum_{i=0}^{\sqrt{u}}W(\sqrt{u},m_i)+O(m)\label{eqn:vebwork}\\
		S(u,\cdot)&=S(\sqrt{u},\cdot)+O(\log u)\label{eqn:vebspan}
	\end{align}
	Note that each key in $B$ falls into at most one of the recursions, and thus $\sum_{i=0}^{\sqrt{u}}m_i=m$. 
    Since $m$ is the number of distinct values to be inserted into the subtree with universe size $u$, we also
    have $m\le u$ in all recursive calls. 
    By solving the recursions above, we can get the claimed bound in the theorem.\ifconference{We also provide the proof to solve them in the full version of this paper.}\iffullversion{We solve them in \cref{app:solverecurrence}.}Note that we assume an even total bits for $u$.
	If not, the number of subproblems and their size become $\sqrt{u/2}+1$ and $\sqrt{2u}$, respectively.
	One can check that the bounds still hold, the same as the sequential analysis.
\end{proof}

\hide{
	\begin{lemma}\label{insertLemma}
		Every batch element makes at most one recursive insertion on a non-empty substructure using \cref{vebinsert}.
	\end{lemma}
	\begin{proof}
		We first recognize that the insertion of $B.\mmin$ should take no recursion if it is smaller than $\vt.\mmin$ and similar observation holds for $B.\mmax$ as well, as indicated in \cref{emptyMinMax}, \cref{nonEmptyMin} and \cref{nonEmptyMax}.
		
		The batch elements when reaching \cref{beginInsert} can be divided into three categories:
		\begin{enumerate}
			\item $\{index(h,l)~|~ h \in H\setminus H', l\in L_h\}$, the elements whose high-bits is already presented in the summary of the current tree. Only their low-bits will incur a recursive insertion, as in \cref{ins-low}.
			\item $\{index(h,)~|~ h\in H', l=\min\{L_h\}\}$, the smallest elements whose high-bits are not contained in $\vt.\summary$. There is one recursive insertion for their high-bits on \cref{ins-high}, and the low-bits insertion on \cref{ins-alone} costs constant work.
			\item $\{index(h,l)~|~ h\in H', l\neq \min\{L_h\}\}$, the other elements whose high-bits is not contained in $\vt.\summary$. Only the low-bits needs a recursive insertion since their high-bits has already been inserted in case (2).
		\end{enumerate}
		Elements above form a partition of the batch in \cref{beginInsert}; therefore we finish the proof.
	\end{proof}
	
	\cref{insertLemma} certifies the insertion routine of each element passes at most $(\log\log U)$ nodes, since every operation in \cref{beginInsert}, \cref{unify}, \cref{extractNewhigh} and \cref{spawn-high} cost $O(m)$ work as discussed, meanwhile visiting a vEB tree node needs  constant work, hence we could derive the following recurrence \ref{vebWorkRecur} to bound the work of \cref{vebinsert}:
	\begin{align}
		w(u,m)=\begin{cases}
			C\cdot m&u=2\\
			\sum_{i=1}^{\sqrt{u}}w(\sqrt{u},m_i)+C\cdot m &u>2
		\end{cases}\label{vebWorkRecur}
	\end{align}
	where $\sum_{i=1}^{\sqrt{u}}m_i=m\leq u$, and $u\leq U$ is the universe size represented by the current (sub)-tree, initially $u=U$. The constant $C$ is non-negative.
	
	\begin{lemma}
		Polynomial ($C\cdot m\log\log U$) is a feasible solution to recurrence \ref{vebwork}.
	\end{lemma}
	\begin{proof}
		Prove by induction on $u$.
		
		When $m=1$, the recurrence reduces to the one for single element insertion, which is certified by ~\cite{CLRS}, now consider only $m>1$.
		
		Base case is validated by picking $C=0$. Now assume the inference sounds for $\sqrt{u}\leq k$, where $k\leq U$, then it holds for $\sqrt{u}\leq k^2$ inductively, since:
		\begin{align*}
			w(u,m)&=C\cdot \sum_{i=1}^{\sqrt{u}}m_i\log\log\sqrt u + C\cdot m\\
			&= C\cdot m (\log\log \sqrt u +1)\\
			&= C\cdot m \log\log u
		\end{align*}
		picking $u=U$ and the proof follows.
	\end{proof}
	
	Since the vEB tree has depth $\Theta(\log\log U)$ and in each recursion level, performing operations on \cref{beginInsert}, \cref{unify}, \cref{extractNewhigh} and \cref{spawn-high} incur additional $O(\log m)$ span; therefore the span for \cref{vebinsert} is $O(\log m\log\log U)$ as desired.
	
	Summarizing above, we state the following \cref{insertTheorem}.
}

%% file: batchInsertVEB.tex
\setlength{\algomargin}{0em}
\begin{algorithm}[t]
	\fontsize{8pt}{8.5pt}\selectfont
	\caption{Batch Insertion Algorithm for vEB tree\label{vebinsert}}
	\SetKwFor{parForEach}{parallel-foreach}{do}{endfor}
	\KwIn{Batch of elements $B$ in sorted order, vEB tree $\vt$. $B\cap \vt=\emptyset$}
	\KwOut{A veb Tree $\vt$ with all keys $x\in B$ inserted}
	\SetKw{MIN}{min}
	\SetKw{MAX}{max}
	\SetKwProg{MyFunc}{Function}{}{end}
	 \SetKwInOut{Note}{Note}
	 \SetKwFor{inParallel}{in parallel:}{}{}
    \DontPrintSemicolon
    \vspace{.2em}
	\MyFunc{\upshape{\textsc{BatchInsert$(\vt,B)$}}}{
		$S\leftarrow \{\vt.\mmin\} \cup \{\vt.\mmax\}$\tcp*[f]{Backup min and max}\label{backupMinMax}\\
		$\vt.\mmin\leftarrow \MIN\{\vt.\mmin,B.\mmin\}$\label{nonEmptyMin}\\
		$\vt.\mmax\leftarrow \MAX\{\vt.\mmax,B.\mmax\}$\label{nonEmptyMax}\\
		$B\leftarrow B\cup S\setminus \{\vt.\mmin\} \setminus \{\vt.\mmax\}$\label{updatedInsert}\\
		
		\If(\tcp*[f]{Deal with high-bits and low-bits of $B$}){$B\neq \emptyset$}{
            \tcp{$H$ are the new high-bits}
            $H\gets \{\high(x)\,|\,x\in B, \vt.\cluster[\high(x)]\text{ is empty}\}$\label{extractNewhigh}\\
            \tcp{For each new high-bit $h$, find the smallest key in $B$ to form $B'$}
            $B'\gets \{x_{h} \,|\, \forall h \in  H, \text{where }x_{h}=\min_{y\in B,\high(y)=h} y \}$\label{line:bprime}\\
            \parForEach(\tcp*[f]{Initialize each new high-bit}){$x\in B'$\label{emptyInsertForLoop}}{
            	$\vt.\cluster[\high(x)].\mmin\gets \low(x)$\label{line:emptyMinMax1}\\
            	$\vt.\cluster[\high(x)].\mmax\gets \low(x)$\label{line:emptyMinMax2}
            }
			\tcp{Find remaining new low-bits}
	        $L[h]\gets\{\low(x)\,|\,\forall x\in B\setminus B', \high(x)=h\in H\}$\label{unify}\\	
            \inParallel{\label{line:parallelInsert}}{
            \batchinsert{}$(\vt.\summary,H)$\label{ins-high}\tcp*[f]{Insert $H$ to summary}\\
            \parForEach(\tcp*[f]{Insert to each cluster}){$h\in H$\label{spawn-high}}{
            	\batchinsert{}$(\vt.\cluster[h],L[h])$ \label{ins-low}}
            }

		}
	}
\end{algorithm}

\hide{
\begin{algorithm}[t]
	\fontsize{8pt}{8.5pt}\selectfont
	\caption{Batch Insertion Algorithm for vEB tree\label{vebinsert}}
	\SetKwFor{parForEach}{parallel-foreach}{do}{endfor}
	\KwIn{Batch of elements $B$ in sorted order, vEB tree $\vt$. $B\cap \vt=\emptyset$}
	\KwOut{A veb Tree $\vt$ with all keys $x\in B$ inserted}
	\SetKw{MIN}{min}
	\SetKw{MAX}{max}
	\SetKwProg{MyFunc}{Function}{}{end}
	\SetKwInOut{Note}{Note}
	\DontPrintSemicolon
	\vspace{.2em}
	\MyFunc{\upshape{\textsc{BatchInsert$(\vt,B)$}}}{
		$S\leftarrow \{\vt.\mmin\} \cup \{\vt.\mmax\}$\tcp*[f]{Backup min and max}\label{backupMinMax}\\
		$\vt.\mmin\leftarrow \MIN\{\vt.\mmin,B.\mmin\}$\label{nonEmptyMin}\\
		$\vt.\mmax\leftarrow \MAX\{\vt.\mmax,B.\mmax\}$\label{nonEmptyMax}\\
		$B\leftarrow B\cup S\setminus \{\vt.\mmin\} \setminus \{\vt.\mmax\}$\label{updatedInsert}\\
		
		\If(\tcp*[f]{Deal with high-bits and low-bits of $B$}){$B\neq \emptyset$}{
			\tcp{$H'$ are the new high-bits}
			$H'\gets \{\high(x)\,|\,x\in B, \vt.\cluster[\high(x)]\text{ is empty}\}$\label{extractNewhigh}\\
			\tcp{For each new high-bit $h'$, find the smallest key in $B$ to form $B'$}
			$B'\gets \{x_{h'} \,|\, \forall h' \in  H', \text{where }x_{h'}=\min_{y\in B,\high(y)=h'} y \}$\label{line:bprime}\\
			\batchinsert{}$(\vt.\summary,H')$\label{ins-high}\tcp*[f]{Insert $H'$ to summary}\\
			\parForEach(\tcp*[f]{Initialize each new high-bit}){$x\in B'$\label{emptyInsertForLoop}}{
				$\vt.\cluster[\high(x)].\mmin\gets \low(x)$\label{line:emptyMinMax1}\\
				$\vt.\cluster[\high(x)].\mmax\gets \low(x)$\label{line:emptyMinMax2}
			}
			$H\gets\{\high(x)\,|\,\forall x\in B\setminus B'\}$\tcp*[f]{exclude keys in $B'$}\label{line:getH}\\
			$L[h]\gets\{\low(x)\,|\,\forall x\in B\setminus B', \high(x)=h\in H\}$\label{unify}\\
			\parForEach(\tcp*[f]{Insert to each cluster}){$h\in H$\label{spawn-high}}{
				\batchinsert{}$(\vt.\cluster[h],L[h])$ \label{ins-low}}
		}
	}
\end{algorithm}
}

\hide{
\begin{algorithm}[h]
	\fontsize{8pt}{9pt}\selectfont
	\caption{Batch Insertion Algorithm for vEB tree\label{vebinsert}}
	\SetKwFor{parForEach}{parallel\_for\_each}{do}{endfor}
	\KwIn{Batch of elements $B$ in sorted order, vEB tree $\vt$}
	\KwOut{A veb Tree $\vt$ with all elements $x\in B$ inserted}
	\SetKw{MIN}{min}
	\SetKw{MAX}{max}
	\SetKwProg{MyFunc}{Function}{}{end}
	 \SetKwInOut{Note}{Note}
	 \Note{Assume all elements $x\in B$ is not already an element in the set represented by vEB tree $\vt$}
	
	\vspace{0.5em}
	\MyFunc{\upshape{\textsc{BatchInsert$(\vt,B)$}}}{
		\tcp{Maintaining minimum and maximum of current subtree}
		\eIf{$\vt$ is empty}{\label{begin-minmax}
			$\vt.\mmin\leftarrow B.\mmin, \vt.\mmax\leftarrow B.\mmax$\label{emptyMinMax}\\
			$B\leftarrow B\setminus(\{B.\mmin\}\cup \{B.\mmax\})$\\
		}{\label{vnotempty}
			$S\leftarrow \{\vt.\mmin\} \cup \{\vt.\mmax\}$\tcp*[f]{Backup min and max}\\
			$\vt.\mmin\leftarrow \MIN\{\vt.\mmin,B.\mmin\}$\label{nonEmptyMin}\\
			$\vt.\mmax\leftarrow \MAX\{\vt.\mmax,B.\mmax\}$\label{nonEmptyMax}\\
			$B\leftarrow B\cup S\setminus (\{\vt.\mmin\} \cup \{\vt.\mmax\})$\label{updatedInsert}\\
		}\label{end-minmax}
		\tcp{Parallel insert high-bits and low-bits of candidate elements}
		\If{$B\neq \emptyset$}{
			$H\leftarrow\{\high(x)|\forall x\in B\}$\label{beginInsert}\\
			$L[h]\leftarrow\{\low(x)|high(x)=h\in H,\forall x\in B\}$\label{unify}\\
			$H'\leftarrow \{h\in H|h\notin \vt.\summary\}$\tcp*[f]{Insert new high-bits}\label{extractNewhigh}\\
			\textsc{BatchInsert($\vt.\summary,H'$)}\label{ins-high}\\
			\parForEach{$h\in H$}{\label{spawn-high}
				\If(\tcp*[f]{insert one low-bits into a vacant cluster }){$h\in H'$\label{test-alone}}{
				\textsc{BatchInsert($\vt.\cluster[h],\MIN\{L[h]\}$)}\\\label{ins-alone}
				$L[h]\gets L[h]\setminus \MIN \{L[h]\}$	\label{remove-min}
			}
				\If(\tcp*[f]{skip empty low-bits cell}){$L[h]\neq\emptyset$\label{non-empty-low}}{\textsc{BatchInsert($\vt.\cluster[h],L[h]$)}}
				 \label{ins-low}
			}
		}
	}
\end{algorithm}
} 

%% file: batchdelete.tex
\subsubsection{Batch Deletion.}\label{sec:batch-delete}
The function \batchdelete{}$(\vt,B)$ deletes a batch of sorted keys $B\subseteq \univ$ from a \veb{} tree $\vt$.
Let $m=|B|$ be the batch size.
For simplicity, we assume $B\subseteq \vt$. If not, we can first look up all keys in $B$ and filter out those that are not in $\vt$ in $O(m\log\log U)$ work and $O(\log m+\log\log U)$ span.
We show our algorithm in \cref{vebdelete}.
The main challenge to performing $m$ deletions in parallel is to properly set the $\mmin$ and $\mmax$ values for each subtree~$t$.
When the $\mmin/\mmax$ value of a subtree $t$ is in $B$, we need to replace it with another key in its subtree that 1) does not appear in $B$, and 2) needs to be further deleted from the corresponding $\cluster[\cdot]$ (recall that the $\mmin/\mmax$ values of a subtree should not be stored in its children).
To resolve this challenge, we keep the \defn{survival predecessor} and \defn{survival successor} for all $x\in B$ wrt.\ a \veb{} tree, defined as follows.

\input{batchDeleteVEB}
\begin{definition}[Survivor Mapping]\label{def:survivormapping}
	\hide{
		Given a vEB tree $\vt$ and a batch $B\subseteq \vt$, define the mapping \defn{survival predecessor} $\mathcal{P}:B\mapsto \vt\setminus B,\forall x\in V$, such that $\mathcal{P}(x)$ is the predecessor of $x$ in set $\vt\setminus B$, i.e., the maximum element smaller than $x$ in set $\setveb$ after the removal of $B$.
		If no such key exists, $\mathcal P(x):=-\infty$.
		Similarly, define \defn{survival successor} $\mathcal{S}:B\mapsto \setveb\setminus B$ that maps every element $x\in B$ to its successor in $\setveb\setminus B$ and $S(x):=+\infty$ if
		no such key exists.
	}
	Given a \veb{} tree $\vt$ and a batch $B\subseteq \vt$, the \defn{survival predecessor} $\mathcal{P}(x)$ for $x\in B$ is the maximum key in $\vt \setminus B$ that is smaller than $x$.
	If no such key exists, $\mathcal P(x):=-\infty$.
	Similarly, the \defn{survival successor} $\mathcal{S}(x)$ for $x\in B$ is the minimum key in $\vt \setminus B$ that is larger than $x$, and is $+\infty$ if no such key exists.
	$\langle\mathcal{P},\mathcal{S}\rangle$ are called the \defn{survival mappings}.
\end{definition}

$\mathcal{P}(\cdot)$ and $\mathcal{S}(\cdot)$ are used to efficiently identify the new keys to replace a deleted key.
For instance, 
if $\vt.\mmax \in B$ (then it must be $B.\mmax$), we can update the value of $\vt.\mmax$ to $\mathcal{P}(B.\mmax)$ directly.

\cref{vebdelete} first initializes the survival mappings (\cref{preprocessing}) as follows.
For each $x\in B$, we set (in parallel) $\mathcal{P}(x)$ as its predecessor in $\vt$ if this predecessor is not in $B$, and set $\mathcal{P}(x)=-\infty$ otherwise.
Then we compute prefix-max of $\mathcal{P}$, and replace the $-\infty$ values by the proper survival predecessor of $x$ in $\vt{}$.

The initial values of $\mathcal{S}$ can be computed similarly.
We then use the \batchdeletehelper{} function to delete batch $B$ from a \veb{} (sub-)tree $\vt$ using the survival mappings, starting from the root.
We use $m$ as the batch size of \emph{the current recursive call}, and $u$ as the universe size of \emph{the current \veb{} subtree}.
The algorithm works in two steps: we first set the $\mmin/\mmax$ values of the tree $\vt$ properly, and then recursively
deal with the $\summary$ and $\cluster$ of $\vt$.

\myparagraph{Restoring $\boldsymbol{\mmin/\mmax}$ values.}
We first discuss how to update $\vt.\mmin$ and $\vt.\mmax$ if they are deleted, in \cref{line:delete:backup}--\ref{line:remove} of \cref{vebdelete}.
We first duplicate $\vt.\mmin$ and $\vt.\mmax$ as $v_{\mmin}$ and $v_{\mmax}$ (\cref{line:delete:backup}), and then check whether $v_{\mmin} \in B$ (\cref{line:testDelMin}).
If so, we replace it with its survival successor (denoted as $y$ on \cref{line:updateminstart}). 
If $y$ is in the clusters ($y\ne \vt.\mmax$), $y$ will be extracted from the corresponding cluster and become $\vt.\mmin$.
To do so, we first delete $y$ sequentially (\cref{line:seqdelete}), and the cost is $O(\log\log u)$.
Then we redirect the survival mapping for keys in $B$ using function \textsc{SurvivorRedirect} since their images may have changed (\cref{redirectMin})---if any of them have survival predecessor/successor as $y$, they should be redirected to some other key in $\vt$ (\cref{redMin}--\ref{redMax}). In particular, if $\mathcal{P}(x)$ is $y$, it should be redirected to $y$'s survival predecessor (\cref{redDelP}).
Similarly, if $\mathcal{S}(x)$ is $y$, it should be redirected to $y$'s survival successor (\cref{redDelS}).
Regarding the cost of \textsc{SurvivorRedirect}, \cref{line:redirect:find} (finding $y$'s predecessor and successor) costs $O(\log \log u)$, but we can charge this cost to the previous sequential deletion on \cref{line:seqdelete}.
The rest of this part costs $O(m)$ work and $O(\log m)$ span.
After that, we set the new $\vt.\mmin$ value as $y$.
The symmetric case applies to when $\vt.\mmax\in B$ (\cref{line:delete:max}).
We then exclude $v_{\min}$ and $v_{\max}$ from $B$ on \cref{line:remove}, since we have handled them properly.
Finally, on \cref{singleEleRemain}, we deal with the particular case where only one key remains after deletion, in which case we have to store it twice in both $\vt.\mmin$ and $\vt.\mmax$.

\myparagraph{Recursively dealing with the low/high bits.}
After we update $\vt.\mmin$ and $\vt.\mmax$ (as shown above), we will recursively update $\vt.\summary$ (high-bits) and $\vt.\cluster{}$ (low-bits), which
requires the algorithm to construct the survival mappings for $\summary{}$ and each $\cluster$.
We first consider the low-bits for $\cluster[h]$ and construct the survival mappings as $\mathcal{P}_h$ and $\mathcal{S}_h$ (\cref{survivorlow}).
Given a key $x\in B$, where $\high(x)=h$, the survival predecessor for its low-bits $\mathcal P_h(\low(x))$ is the low-bits of its survival predecessor $\low(\mathcal P(x))$ if $x$ and $\mathcal P(x)$ have same high-bits $h$ (\cref{surLowPreSameHigh}). Otherwise $\low(x)$ would become the smallest key in $\vt.\cluster[h]$ after removing $B$, therefore we map $\mathcal P(\low(x))$ to $-\infty$ (\cref{surLowPreNotSameHigh}).
We can construct $\mathcal S(\low(x))$ similarly (\cref{surLowSucSameHigh}--\ref{surLowSucNotSameHigh}).
Note that we have to exclude $\vt.\mmin$ and $\vt.\mmax$ since they do not appear in the clusters.
Then we can recursively call \batchdeletehelper{} on the $\cluster[h]$ using the survival mappings $\mathcal{P}_h$ and $\mathcal{S}_h$ (\cref{line:recursionlow}).
In total, constructing all survival mappings for low-bits costs $O(m)$ work and $O(\log m)$ span.



We then construct the survival mapping $\mathcal{P}'$ and $\mathcal{S}'$ for high-bits (\cref{survivorhigh}).
Recall that the clusters of $\vt$ contain all keys in $\vt \setminus \{\vt.\mmin\} \setminus \{\vt.\mmax\}$. Therefore if $\mathcal P(x) = \vt.\mmin$, then $\mathcal P(\high(x))$ should be mapped to $-\infty$.
Otherwise let $y\in B$ be the maximum key except $\vt.\mmin$ and $\vt.\mmax$ with the same high-bits as $x$,
then we have $\mathcal P(\high(x))=\high(\mathcal P(y))$ by definition (\cref{surHighPre}). 
We can construct $\mathcal{S}(h)$ similarly (\cref{surHighSuc}).
The total cost of finding survival mappings for high-bits is also $O(m)$ work and $O(\log m)$ span.

We now analyze the cost of \cref{vebdelete} in \cref{deletetheorem}.

\begin{theorem}\label{deletetheorem}
	Given a \veb{} tree $\vt$, deleting a sorted batch $B\subseteq \vt$ costs $O(m\log\log U)$ work and $O(\log U\log\log U)$ span, where $m=|B|$ is the batch size and $U=|\univ|$ is the universe size.
\end{theorem}

\ifconference{Due to the page limit, we show the (informal) high-level ideas here and prove it in the full version of this paper.}
\iffullversion{Due to the page limit, we show the (informal) high-level ideas here and prove it in \cref{app:deleteproof}.}
The span recurrence of the batch-deletion algorithm is similar to batch-insertion, which indicates the same span bound.
The work-bound proof is more involved.
The challenge lies in that a key in $B$ can be involved in both recursive calls on \cref{line:recursionhigh} and \cref{line:recursionlow},
which seemingly costs work-inefficiency.
However, for each high-bit $h$ to be deleted on the recursive call on \cref{line:recursionhigh},
it indicates that the corresponding $\cluster[h]$ will become empty after the deletion of low-bits.
Therefore, the smallest low-bit among them must be exceptional and will be handled
by the base cases on Lines \ref{line:updateminstart}--\ref{line:updateminend}.
Therefore, for each key in $B$, only one of the recursive calls will be ``meaningful''.
If the audience is familiar with (sequential) \veb{} trees, this is very similar to the sequential analysis---for the two recursive calls
on the low- and high-bits, only one of them will be invoked in any single-point insertion/deletion, and the $O(\log \log U)$ bound thus holds.
In the parallel version, we need to further analyze the cost on Lines \ref{line:updateminstart}--\ref{line:updateminend} to restore the $\mmin/\mmax$ values.
\ifconference{We show a formal proof for \cref{deletetheorem} in the full version of this paper.}
\iffullversion{We show a formal proof for \cref{deletetheorem} in \cref{app:deleteproof}.}

\hide{
	\begin{lemma}\label{deletionLemma}
		Delete one element in batch requires at most one recursion call on a non-empty substructure using \cref{vebdelete}.
	\end{lemma}
	\begin{proof}
		We first observe that deleting the $\vt.\mmin$ and $\vt.\mmax$ costs $O(1)$ work, even if we lift a substitute element from sub-tree of $\vt$, one recursion call is sufficient to delete it sequentially ~\cite{CLRS}.
		
		\cref{removelow} can be interpreted as two steps: we first delete batch low-bits in $L[h]$ except the smallest one, $\min\{L[h]\}$, using one recursion call, then we delete the $\min\{L[h]\}$ individually. Meanwhile, notice that to delete a high-bits pattern $h$ from summary in \cref{removehigh} using one recursion call, $\min\{L[h]\}$ must be the last element that remains in  $\vt.\cluster[h]$, but this would result in the deletion for itself costing only $O(1)$ work by \cref{deleteone}.
		
		Proof finish.
	\end{proof}
	
	By \cref{deletionLemma} we know every element presents in at most $O(\log\log U)$ sub-problems, and the cost to visit a node in the vEB tree is $\Theta(1)$, this also bounds the work for the sequential deletion of a lifted element $y$ in \cref{sequentiallyDelete}, so that the work of finding the predecessor and successor of $y$ in \cref{findPredSucc} can be amortized to the sequential deletion of $y$. Meanwhile, the work in Lines \ref{redirectMin}, \ref{redirectMax}, \ref{findhigh}, \ref{findlow}, \ref{survivorlow} and \ref{survivorhigh} cost $O(m)$ work as discussed, therefore recurrence \ref{vebWorkRecur} is still valid to bound the work of \cref{vebdelete}. To analyze the span, notice all operations in Lines \ref{redirectMin}, \ref{redirectMax}, \ref{findhigh}, \ref{findlow}, \ref{survivorlow} and \ref{survivorhigh} requires $O(\log m)$ span in each recursion level and vEB tree has depth $O(\log\log U)$. Combined with the $O(m\log\log U)$ work and $O(\log m)$ span imposed by preprocessing in \cref{preprocessing}, we summarize the above analysis by stating the following \cref{deletetheorem}.
}

%% file: batchDeleteVEB.tex
\begin{algorithm}[t]
	\fontsize{8pt}{8.5pt}\selectfont
	\caption{Batch Deletion Algorithm for vEB tree\label{vebdelete}}
	\SetKwFor{parForEach}{parallel-foreach}{do}{endfor}
	\KwIn{A \veb{} tree $\vt$ and a batch of keys $B\subseteq \vt$ in sorted order}
	\KwOut{Update $\vt$ by deleting all keys $x\in B$}
	\SetKwInOut{Note}{Note}
	\SetKw{MIN}{min}
	\SetKw{MAX}{max}
	\SetKw{OR}{or}
	\SetKw{AND}{and}
	\SetKw{IF}{if}
	\SetKwProg{MyFunc}{Function}{}{end}
	\codeskip{}
	\DontPrintSemicolon
	
	\MyFunc{\upshape\batchdelete{}($\vt,B$)}{
		Initialize survival mappings $\mathcal P$ and $\mathcal S$ with respect to $B$ and $\vt$\label{preprocessing}\\
		\lIf{$B\ne \emptyset$}{\batchdeletehelper{}($\vt,B,\mathcal{P},\mathcal{S}$)}
	}
	
	\codeskip{}
	
	\MyFunc{\upshape\batchdeletehelper{}($\vt,B,\mathcal{P},\mathcal{S}$)}{
		\tcp{Maintaining min/max of current tree}
        $\langle v_{\mmin}, v_{\mmax} \rangle \gets \langle \vt.\mmin,\vt.\mmax \rangle$\label{line:delete:backup}\\
        \If(\tcp*[f]{if $\vt.\mmin\in B$, it must be $B.\mmin$}\label{line:testDelMin}){$v_{\mmin}=B.\mmin$}{
          $y\gets \mathcal{S}(B.\mmin)$\label{line:updateminstart}\\
		  \If(\tcp*[f]{if $y$ is in the clusters}){$y\ne\vt.\mmax$ \AND $y\ne +\infty$}{
            Delete $y$ from $\vt$ sequentially\label{line:seqdelete}\\
			$\langle\mathcal P,\mathcal S\rangle\gets$ \textsc{SurvivorRedirect($\vt,B,y,\mathcal P,\mathcal S$)}\label{redirectMin}\\
		  }
		  $\vt.\mmin\gets y$\label{line:updateminend}
        }
        \lIf{$v_{\mmax}=B.\mmax$} {...\tcp*[f]{Mostly symmetric to Lines \ref{line:updateminstart}--\ref{line:updateminend}}\label{line:delete:max}}
        $B\gets B\setminus \{v_{\mmin}\} \setminus \{v_{\mmax}\}$\label{line:remove}\\
        \lIf{$\vt.\mmax=-\infty$ \AND $\vt.\mmin\neq +\infty$}{$\vt.\mmax\gets \vt.\mmin$}\label{singleEleRemain}
		\If(\tcp*[f]{Recursively deal with the batch}){$B\neq \emptyset$}
		{\label{beginDelete}
			$H\leftarrow\{\high(x)|\forall x\in B\}$\\\label{findhigh} $L[h]\leftarrow\{\low(x)|\high(x)=h\in H,\forall x\in B\}$\\\label{findlow}
			\parForEach{$h\in H$}{
				$\langle\mathcal P_h, \mathcal S_h\rangle\leftarrow\textsc{SurvivorLow}(h,L[h],\mathcal P,\mathcal{S})$\\\label{survivorlow}
				\batchdeletehelper($\vt.\cluster[h],L[h],\mathcal P_h,\mathcal S_h$)\label{line:recursionlow}\\
			}
			$H'\leftarrow \{h\in H\,|\,\vt.\cluster[h] \text{ is empty}\}$\label{extractDeletedHigh}\\
			$\langle\mathcal P',\mathcal S'\rangle\leftarrow \textsc{SurvivorHigh}(H',L,\mathcal P,\mathcal S)$\label{survivorhigh}\\
			\batchdeletehelper$(\vt.\summary,H',\mathcal P_{H'}, \mathcal S_{H'})$\label{line:recursionhigh}\\
		}
	}
	\vspace{0.5em}
	\tcp{Redirect the survival mapping $\mathcal P$ and $\mathcal S$ concerning elements in batch $B$ after sequential deletion of $y$ from vEB tree $\vt$}
	\MyFunc{\upshape\textsc{SurvivorRedirect$(\vt,B,y,\mathcal P, \mathcal S)$}\label{line:redirect}}{
		$\langle p, s\rangle \leftarrow \langle \predessor(\vt,y), \successor(\vt,y)\rangle$\label{line:redirect:find}\;
		\lIf{$p\in B$}{$p\gets \mathcal P(p)$}\label{redDelP}
		\lIf{$s\in B$}{$s\gets \mathcal S(s)$}\label{redDelS}
		\parForEach{$x\in B$}{
			\lIf{$\mathcal P(x) = y$}{$\mathcal P(x)\gets	p$}\label{redMin}
			\lIf{$\mathcal S(x) = y$}{$\mathcal S(x)\gets	s$}\label{redMax}
		}
		\Return{$\langle\mathcal P,\mathcal S\rangle$}\\
	}
	
	\codeskip{}
	\tcp{Build survival predecessor $\mathcal P_h$ and successor $\mathcal S_h$ for elements in $L[h]$}
	\MyFunc{\upshape\textsc{SurvivorLow($h, L, \mathcal P, \mathcal S$)}}{
		$\mathcal P_h\leftarrow\emptyset,\mathcal S_h\leftarrow\emptyset$\\
		\parForEach{$l\in L[h]$}{
			$\langle p, s\rangle \gets \langle \mathcal P(\idx(h,l)), \mathcal S(\idx(h,l))\rangle$\\
			\lIf{$\high(p)=h$ \AND $p\neq \vt.\mmin$}{$\mathcal P_h(l)\leftarrow \low(p)$}\label{surLowPreSameHigh}
			\lElse{$\mathcal P_h(l)\leftarrow -\infty$}\label{surLowPreNotSameHigh}
			\lIf{$\high(s)=h$ \AND $s\neq \vt.\mmax$}{$\mathcal S_h(l)\leftarrow \low(s)$}\label{surLowSucSameHigh}
			\lElse{$\mathcal S_h(l)\leftarrow +\infty$}\label{surLowSucNotSameHigh}
		}
		\Return{$\langle\mathcal P_h,\mathcal S_h\rangle$}
	}
	\codeskip{}
	\tcp{Build survival predecessor $\mathcal P'$ and successor $\mathcal S'$ for elements in $H$}
	\MyFunc{\upshape\textsc{SurvivorHigh}$(H,L,\mathcal P,\mathcal S)$}{
		$\mathcal P'\leftarrow \emptyset, \mathcal S'\leftarrow \emptyset$\\
		\parForEach{$h\in H$}{
			$\langle p,s\rangle \gets \langle\mathcal P(\idx(h,\MIN \{L[h]\}), \mathcal S(\idx(h,\MAX \{L[h]\}))\rangle$\\
			\leIf{$p\neq \vt.\mmin$}{$\mathcal P'(h)\leftarrow \high(p)$}{$\mathcal P'(h)\leftarrow-\infty$}\label{surHighPre}
			\leIf{$s\neq \vt.\mmax$}{$\mathcal S'(h)\leftarrow \high(s)$}{$\mathcal S'(h)\leftarrow+\infty$}\label{surHighSuc}
		}
		\Return{$\langle\mathcal P', \mathcal S'\rangle$}
	}
	
\end{algorithm} 

%% file: exp.tex
\newcommand{\seqbs}{\textsf{Seq-BS}}
\newcommand{\seqavl}{\textsf{Seq-AVL}}
\newcommand{\swgsimp}{\textsf{\swgs}}
\newcommand{\randompattern}{range}
\newcommand{\linepattern}{line}

\section{Experiments}
\label{sec:exp}
In addition to the new theoretical bounds,
we also show the practicality of the proposed algorithms by implementing our LIS (\cref{algo:lis}) and WLIS algorithms (\cref{algo:lis-weighted} using range trees).
Our code is light-weight. 
We use the experimental results to show how theoretical efficiency enables better performance in practice over the existing results.
We plan to release our code.

\myparagraph{Experimental Setup.}
We run all experiments on a 96-core (192-hyperthread) machine equipped with four-way Intel Xeon Gold 6252 CPUs and 1.5 TiB of main memory. Our implementation is in \texttt{C++} with ParlayLib~\cite{blelloch2020parlaylib}.
All reported numbers are the averages of the last three runs among four repeated tests.

\input{fig-exp}
\myparagraph{Input Generator.}
We run experiments of input size $n=10^8$ and $n=10^9$ with varying ranks (LIS length $\lislength$).
We use two generators and refer to the results as the \defn{\randompattern{}} pattern and the \defn{\linepattern{}} pattern, respectively.
The \defn{\randompattern{}} pattern is a sequence consisting of integers randomly chosen from a range $[1,k']$.
The values of $k'$ upper bounds the LIS length.
When $k$ is large, and the largest possible rank of a sequence of size $n$ is expected to be $2\sqrt{n}$~\cite{johansson1998longest}. 
To generate inputs with larger ranks, we use a \defn{\linepattern{}} pattern generator that draws $A_i$ as $t\cdot i+s_i$ for a sequence $A_{1\dots n}$, where $s_i$ is an independent random variable chosen from a uniform distribution.
We vary $t$ and $s_i$ to achieve different ranks.
For the weighted LIS problem, we always use random weights from a uniform distribution.


\myparagraph{Baseline Algorithms.} We compare to standard sequential LIS algorithms and the existing parallel LIS implementation from \swgsimp{}~\cite{shen2022many}.
We also show the running time of our algorithm on one core to indicate the work of the algorithm.
\swgsimp{} works on both LIS and WLIS problems with $O(n\log^3 n)$ work and $\tilde{O}(\lislength)$ span,
and we compare both of our algorithms (\cref{algo:lis,algo:lis-weighted}) with it.

\input{fig-scal}
For the LIS problem, we also use a highly-optimized sequential algorithm from~\cite{Knuth73vol3}, and call it \seqbs{}.
\seqbs{} maintains an array~$B$, where $B[r]$ is the smallest value of $A_i$ with rank $r$.
Note that $B$ is monotonically increasing.
Iterating $i$ from $1$ to $n$, we binary search $A_i$ in $B$, and if $B[r] < A_i \le B[r+1]$, we set $dp[i]$ as $r+1$.
By the definition of $B[\cdot]$, we then update the value $B[r+1]$ to $A_i$ if $A_i$ is smaller than the current value in $B[r+1]$.
The size of $B$ is at most $\lislength$, and thus this algorithm has work $O(n\log \lislength)$.
This algorithm only works on the unweighted LIS problem.

For WLIS, we implement a sequential algorithm and call it \seqavl{}.
This algorithm maintains an augmented search tree, which stores all input objects ordered by their values,
and supports range-max queries. 
Iterating $i$ from $1$ to $n$, we simply query the maximum \dpvalue{} in the tree among all objects with values less than $A_i$, and update $\mathdp[i]$.
We then insert $A_i$ (with $\mathdp[i]$) into the tree and continue to the next object.
This algorithm takes $O(n\log n)$ work, and we implement it with an AVL tree.

\smallskip

Due to better work and span bounds, our algorithms are always faster than the existing parallel implementation \swgsimp{}.
Our algorithms also outperform highly-optimized sequential algorithms up to reasonably large ranks (e.g., up to $k= 3\times 10^5$ for $n=10^9$).
For our tests on $10^8$ and $10^9$ input sizes, our algorithm outperforms the sequential algorithm on ranks from 1 to larger than $2\sqrt{n}$.
We believe this is the \defn{first parallel LIS implementation that can outperform the efficient sequential algorithm in a large input parameter space}.


\myparagraph{{Longest Increasing Subsequence (LIS)}.} \cref{fig:result}(a) shows the results on input size $n=10^8$ with ranks from $1$ to $10^7$ using the \linepattern{} generator.
For our algorithm and \seqbs{}, the running time first increases with $\lislength$ getting larger because both algorithms have work $O(n\log \lislength)$.
When $k$ is sufficiently large, the running time drops slightly---larger ranks bring up better cache locality, as each object is likely to extend its LIS from an object nearby.
Our parallel algorithm is faster than the sequential algorithm for $\lislength\le 3\times 10^4$ and gets slower afterward.
The slowdown comes from the lack of parallelism ($\tilde{O}(\lislength)$ span).
Our algorithm running on one core is only 1.4--5.5$\times$ slower than \seqbs{} due to work-efficiency.
With sufficient parallelism (e.g., on low-rank inputs), our performance is better than \seqbs{} by up to 16.8$\times$.


We only test \swgsimp{} on ranks up to $10^4$ because it costs too much time for larger ranks.
In the existing results, our algorithm is always faster than \swgsimp{} (up to 188$\times$) because of better work and span.
We believe the simplicity of code also contributes to the improvement.

We evaluate our algorithm on input size $n=10^9$ with varied ranks from $1$ to $10^8$ using \linepattern{} the generator (see \cref{fig:result}(b)) and with varied ranks from $1$ to $6 \times 10^4$ using the \randompattern{} generator (see \cref{fig:result}(c)).
We exclude \swgsimp{} in the comparison due to its space-inefficiency, since it ran out of memory to construct the range tree on $10^9$ elements.  
For $\lislength\le 3\times 10^5$, our algorithm is consistently faster than \seqbs{} (up to 9.1$\times$).
When the rank is large, the work in each round is not sufficient to get good parallelism, and the algorithm behaves as if it runs sequentially.
Because of work-efficiency, even with large ranks, our parallel algorithm introduces limited overheads, and its performance is comparable to \seqbs{} (at most 3.4$\times$ slower).
We also evaluate the self-relative speedup of our algorithm on input size $n=10^9$ with rank $10^2$ and rank $10^4$ using both \linepattern{} and \randompattern{} generators.
In all settings from \cref{fig:scal}, our algorithm scales well to 96 cores with hyperthreads, reaching the self-speedup of up to 25.6$\times$ for $k=10^2$ and up to 37.0$\times$ for $k=10^4$.
With the same rank, our algorithm has almost identical speedup for both patterns in all scales.
Our algorithm outperforms \seqbs{} (denoted as dash lines in \cref{fig:scal}) when using 8 or 16 cores.

Overall, our LIS algorithm performs well with reasonable ranks,
achieving up to 41$\times$ self-speedup with $n=10^8$ and up to 70$\times$ self-speedup with $n=10^9$.
Due to work-efficiency, our algorithm is scalable and performs especially well on large data because
larger input sizes result in more work to utilize parallelism better.

\myparagraph{Weighted LIS.}
We compare our WLIS algorithm (\cref{algo:lis-weighted}) with \swgsimp{} and \seqavl{} on input size $n=10^8$.
We vary the rank from $1$ to $3000$, 
and show the results in \cref{fig:result}(d).
Our algorithm is always faster than \swgsimp{} (up to 2.5$\times$).
Our improvement comes from better work bound (a factor of $O(\log n)$ better, although in many cases \swgsimp{}'s work bound is not tight).
Our algorithm also outperforms the sequential algorithm \seqavl{} with ranks up to $100$.
The running time of the sequential algorithm decreases with increasing ranks $k$ because of better locality. In contrast, our algorithm performs worse with increasing $k$ because of the larger span.

The results also imply the importance of work-efficiency in practice.
To get better performance, we believe an interesting direction is to design a work-efficient parallel algorithm for WLIS.

\hide{
	\begin{figure*}[ht]
		\centering
		\small
		\includegraphics[width=.35\columnwidth]{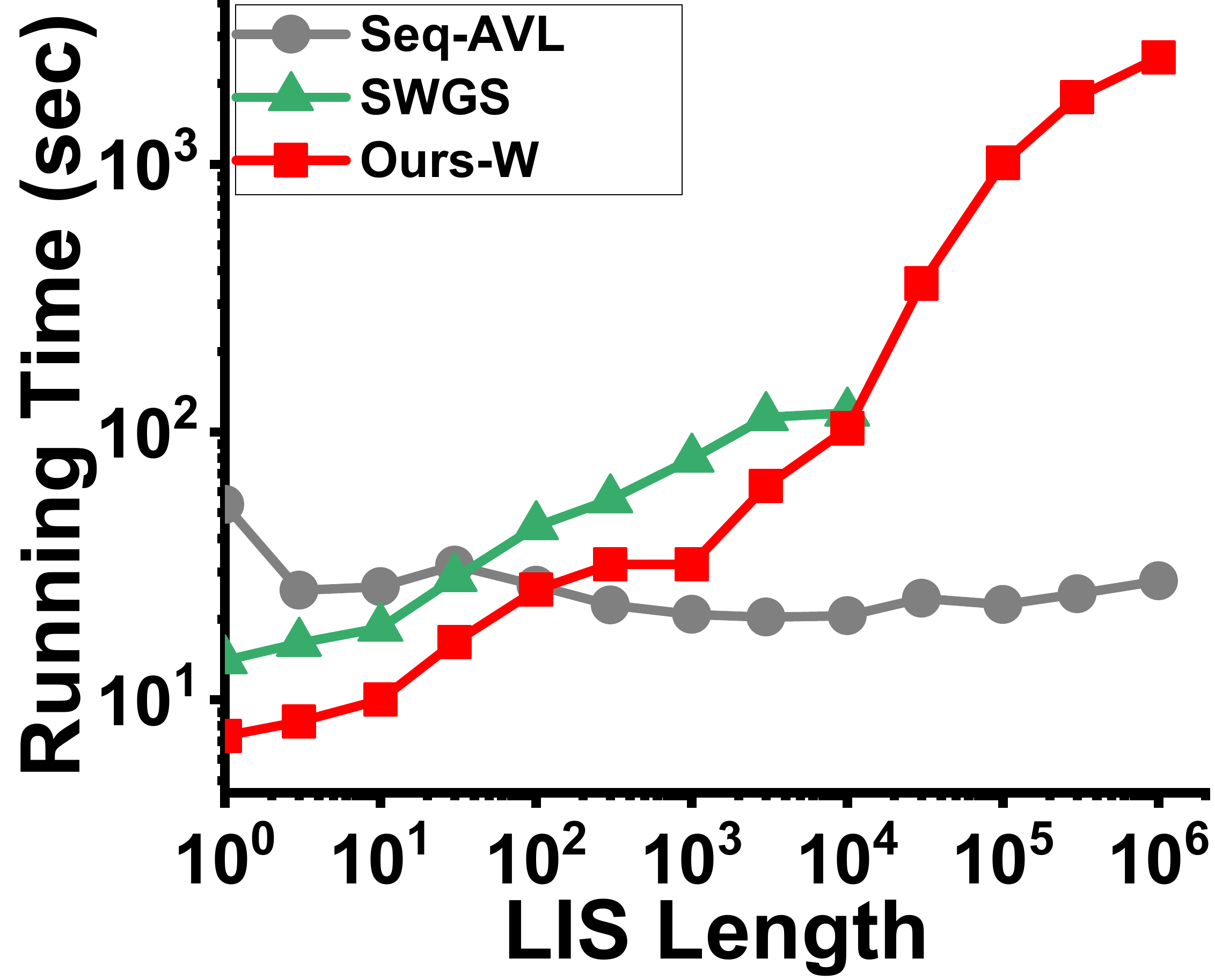}
		\caption{
			\small\textbf{Experiments on the weighted LIS.} \\
			We vary the output size, with the input size $n=10^8$.
			\label{fig:lis-weighted}
		}
	\end{figure*}
}

%% file: fig-exp.tex
\begin{figure*}[t]
	\centering
	\small
    \vspace{-1em}
	\begin{minipage}{\textwidth}
		\begin{tabular}{cccc}
			\includegraphics[width=.23\textwidth]{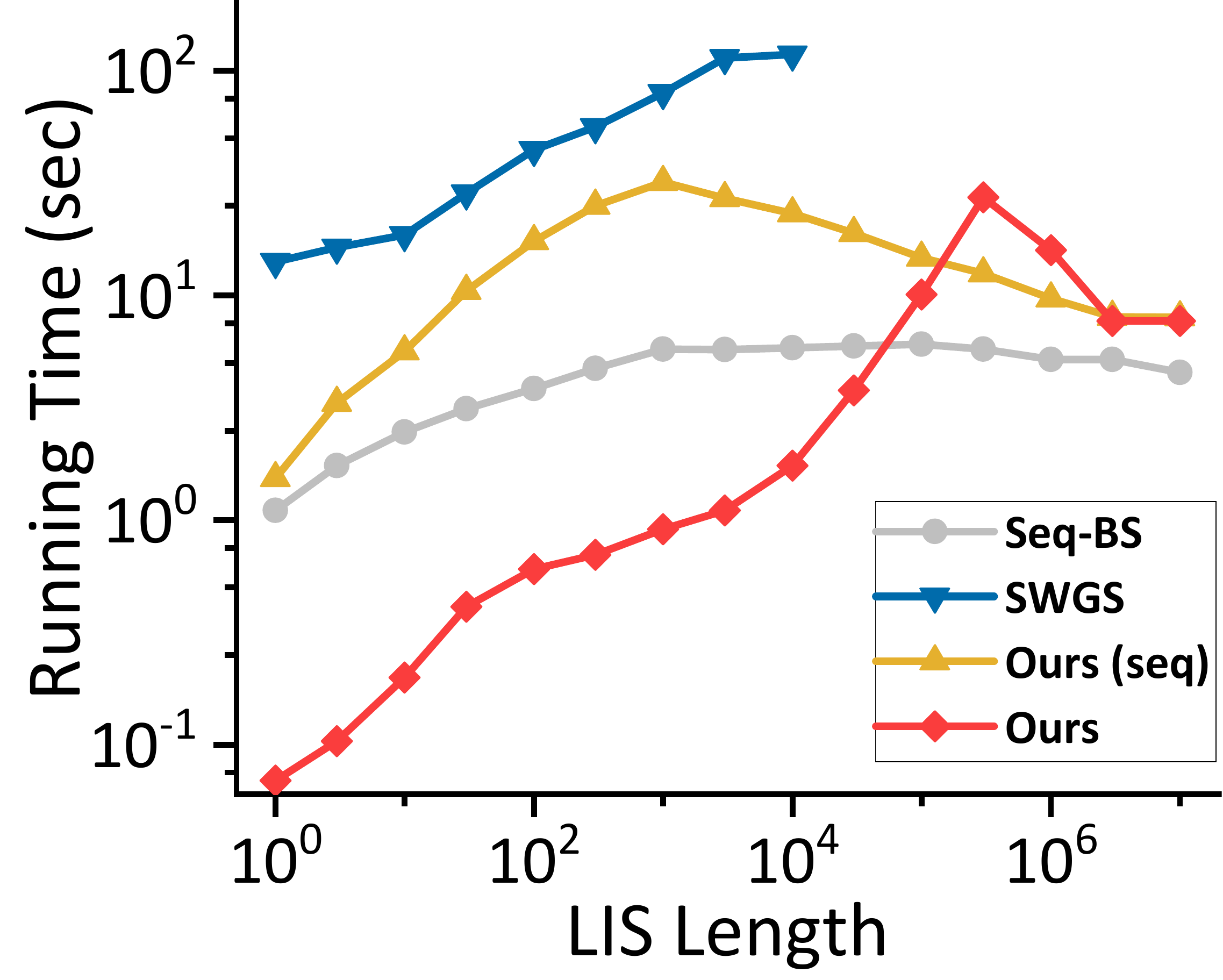}&
			\includegraphics[width=.23\textwidth]{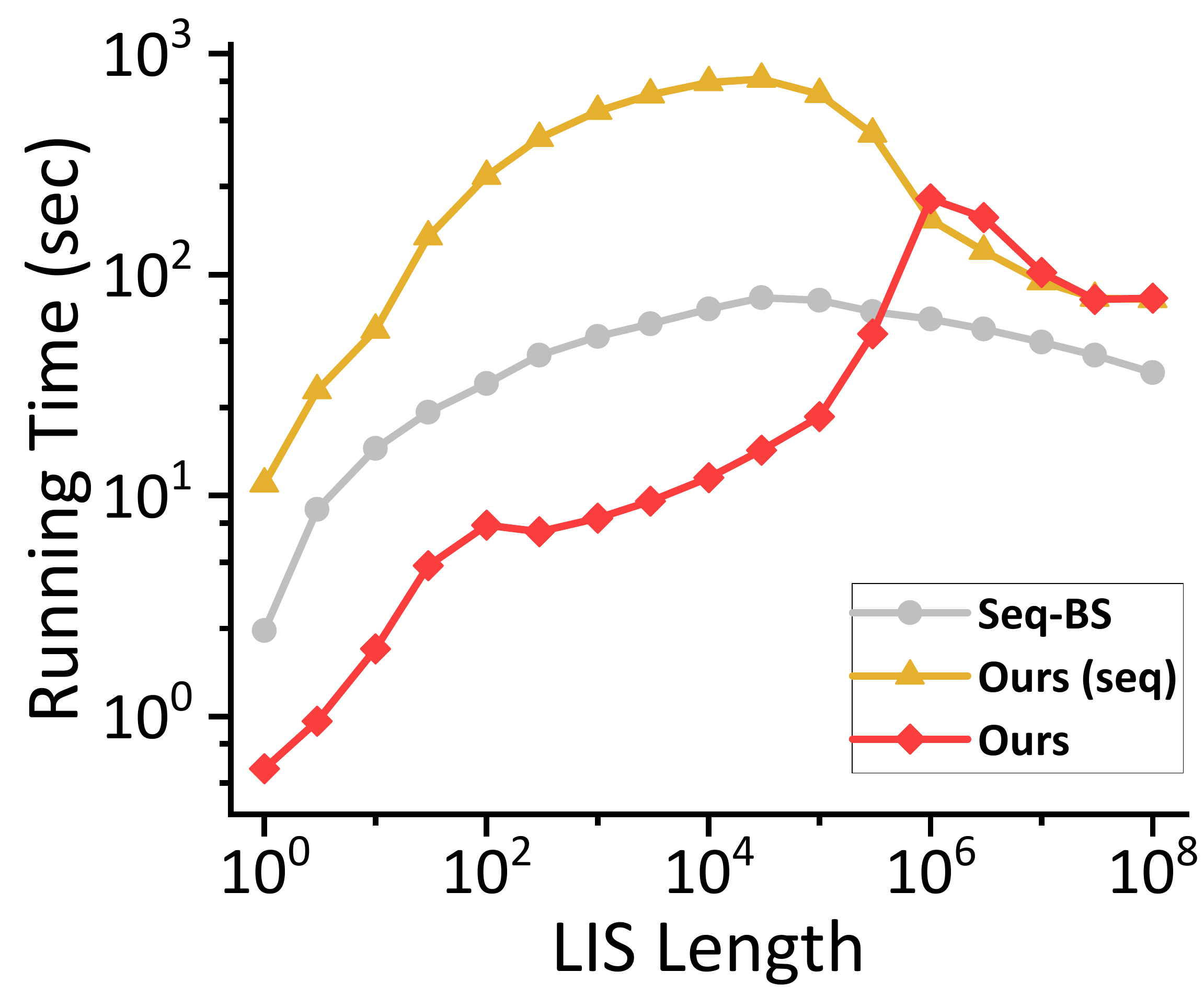}&
			\includegraphics[width=.23\textwidth]{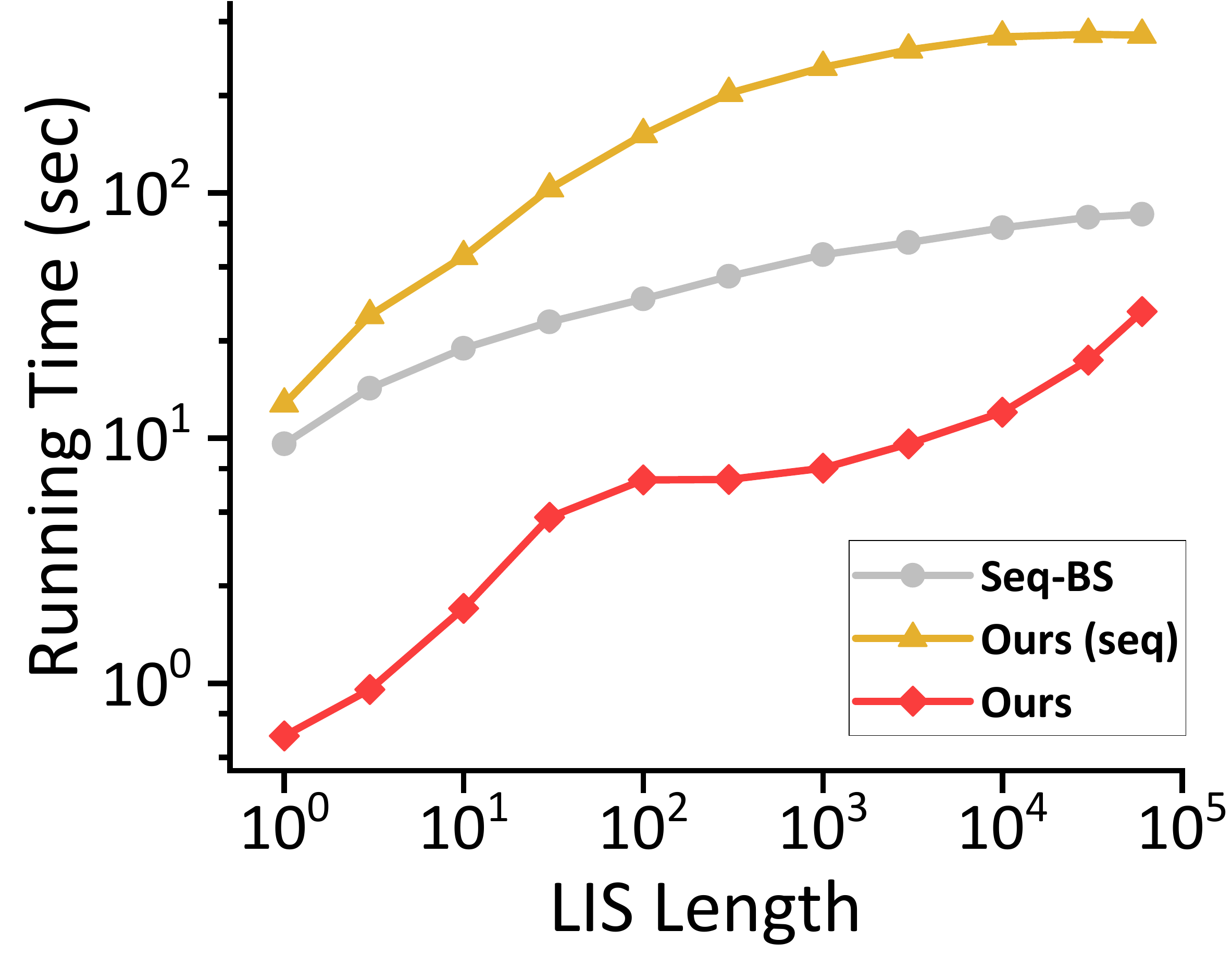}&
			\includegraphics[width=.23\textwidth]{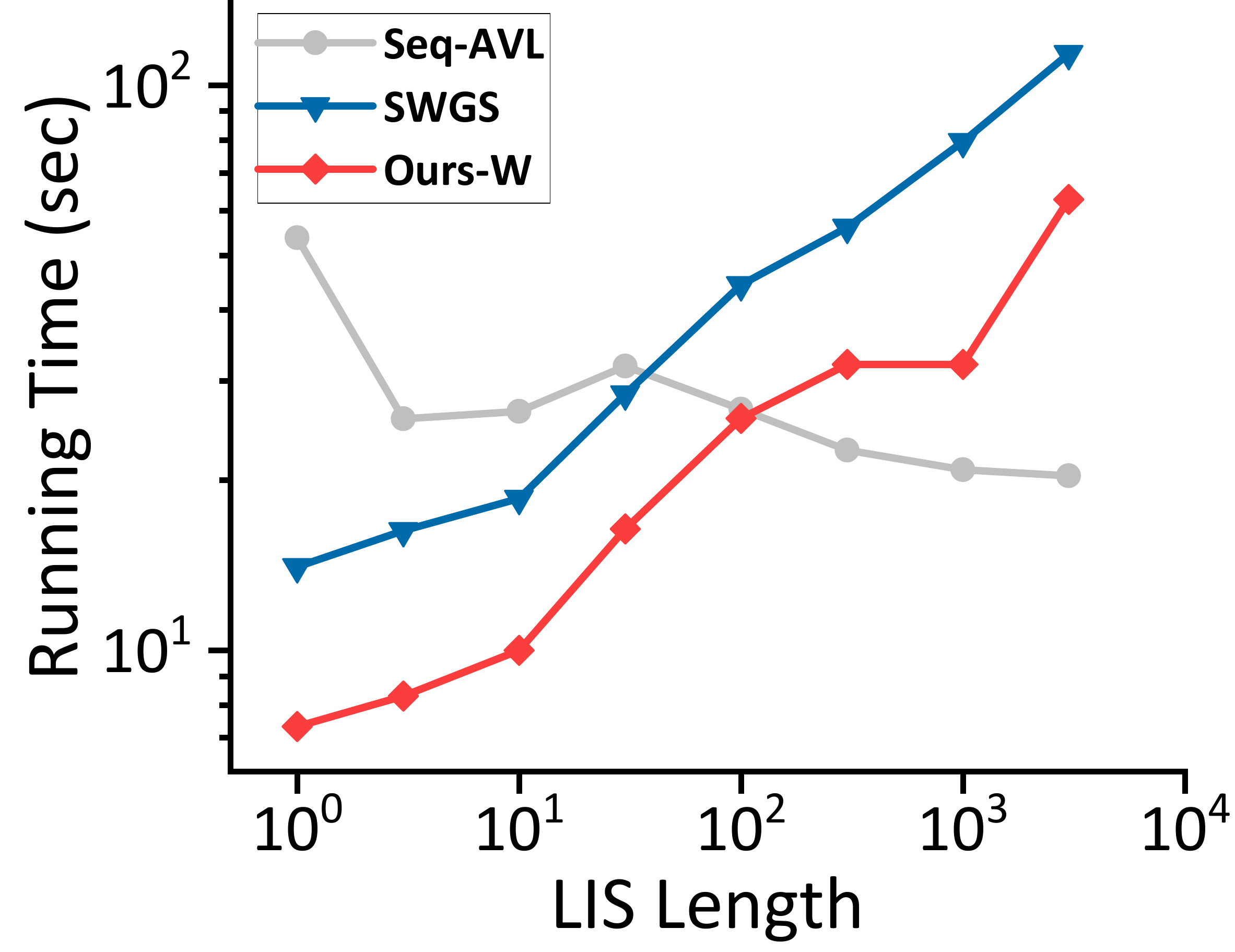}\\
			\bf (a). LIS. $\boldsymbol{n=10^8}$.  & \bf (b). LIS. $\boldsymbol{n=10^9}$. &\bf (c). LIS. $\boldsymbol{n=10^9}$. &\bf (d). Weighted LIS. $\boldsymbol{n=10^8}$.\\
			\bf Line pattern. & \bf Line pattern.&\bf Range  pattern.&\bf Line Pattern.\\
		\end{tabular}
	\end{minipage}
	\begin{minipage}{\textwidth}
		\caption{
			\small\textbf{Experimental results on the LIS and WLIS.}
			We vary the output size for each test.
			``Ours''$=$ our LIS algorithm in \cref{algo:lis} using 96 cores.
			``Ours (seq)''$=$ our LIS algorithm in \cref{algo:lis} using one core.
			``Ours-W''$=$our WLIS algorithm in \cref{algo:lis-weighted} using 96 cores.
			``Seq-BS''$=$ the sequential \seqbs{} algorithm based on binary search.
			``Seq-AVL''$=$ the sequential \seqavl{} algorithm based on the AVL tree.
			``\swgs{}''$=$ the parallel algorithm \swgsimp{} from~\cite{shen2022many}. See more details in \cref{sec:exp}.
			\label{fig:result}
		}
	\end{minipage}
    \vspace{-1.5em}
\end{figure*} 

%% file: fig-scal.tex
\begin{figure}[t]
	\begin{minipage}{\columnwidth}
		\begin{tabular}{cc}
			\includegraphics[width=0.46\columnwidth]{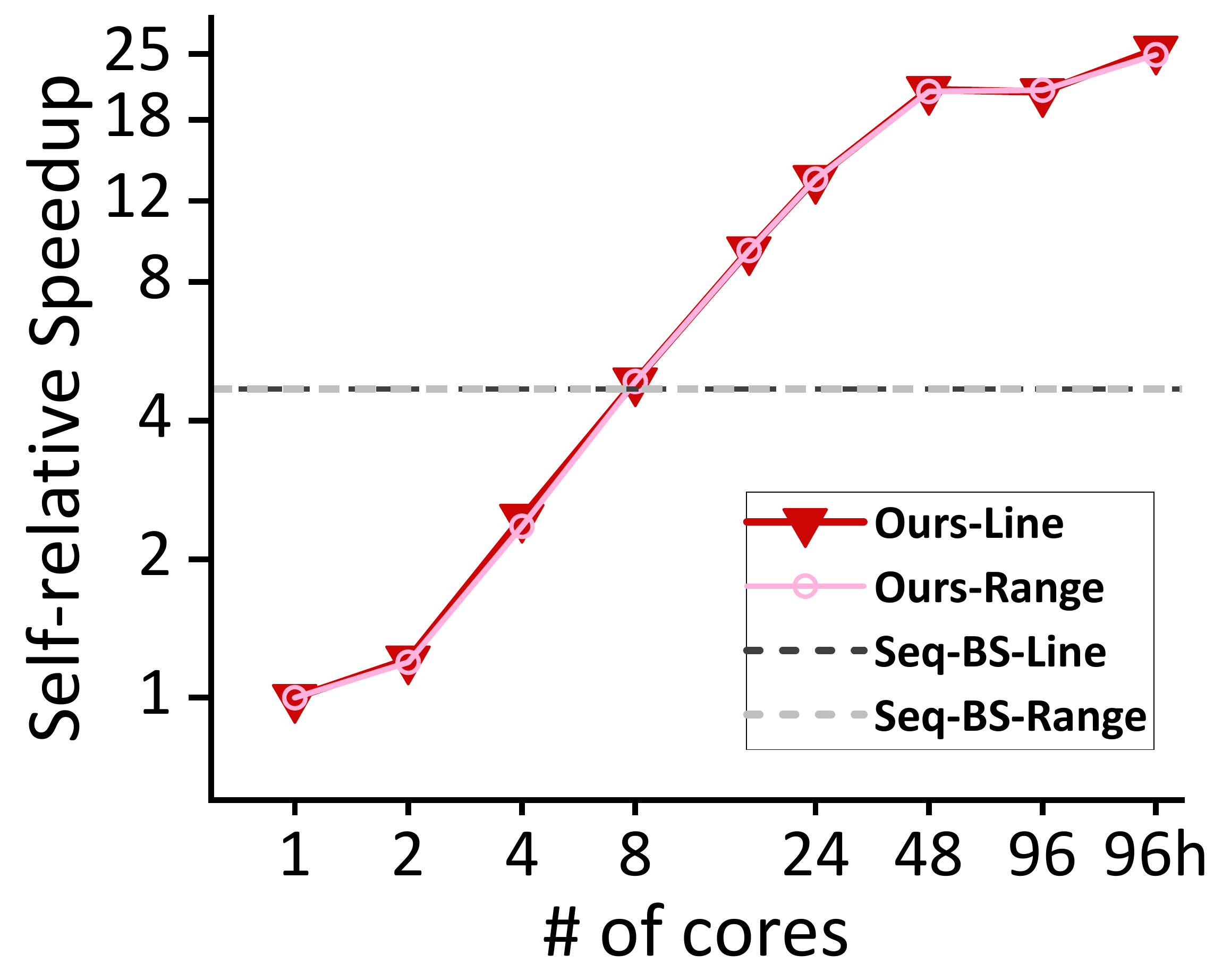}&
			\includegraphics[width=0.46\columnwidth]{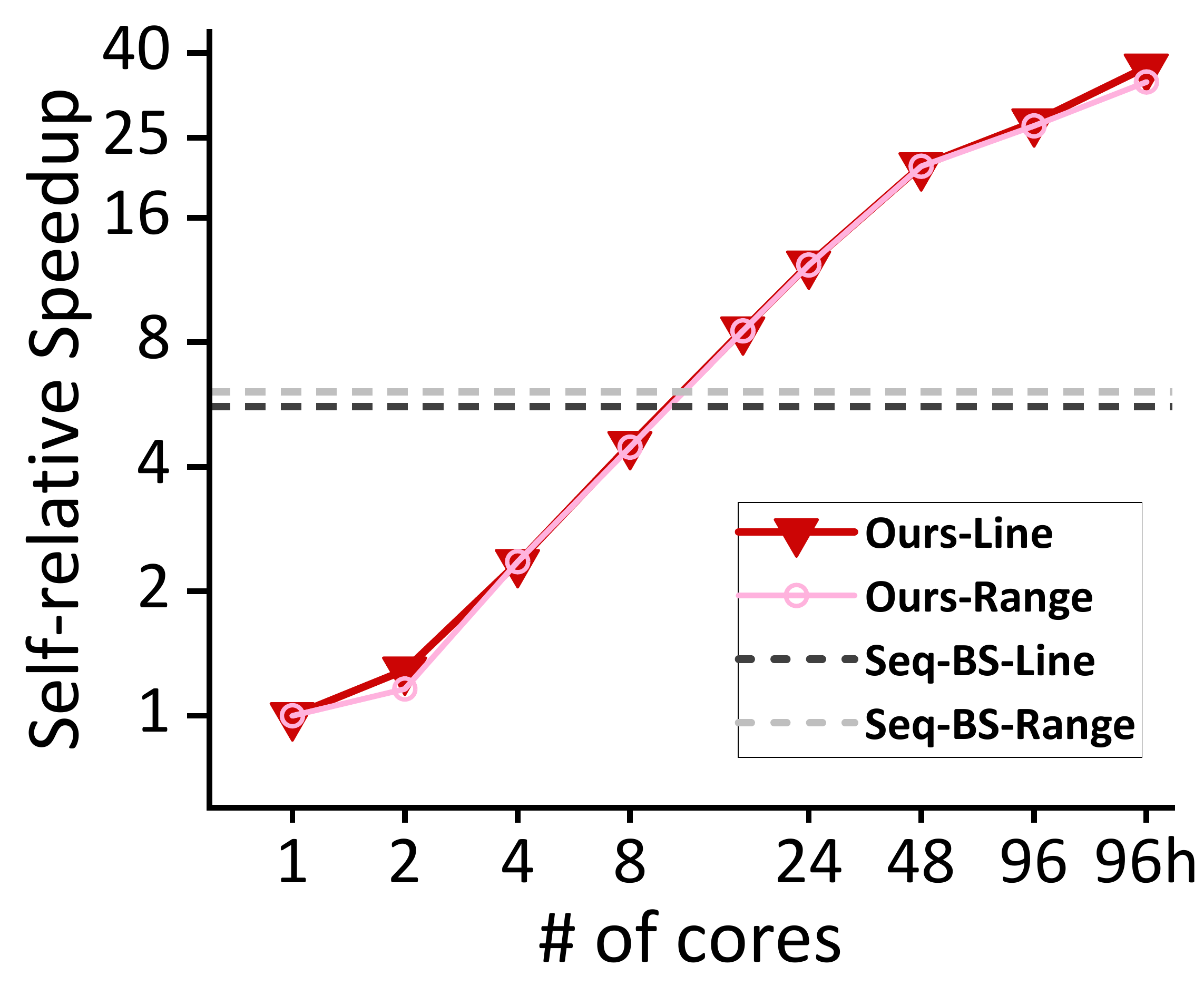}\\
			\bf (a). LIS. $\boldsymbol{k=10^2}$. & \bf (b). LIS. $\boldsymbol{k=10^4}$.\\
		\end{tabular}
	\end{minipage}

	\begin{minipage}{\columnwidth}
			\caption{
                \small \textbf{Experimental results of Self-relative Speedup.}
    			``Ours-Line''$=$ our LIS algorithm in \cref{algo:lis} using a \linepattern{} pattern generator.
    			``Ours-Range''$=$ our LIS algorithm in \cref{algo:lis} using a \randompattern{} pattern generator.
    			``Seq-BS-Line''$=$ \seqbs{} algorithm using a \linepattern{} pattern generator.
    			``Seq-BS-Range''$=$ \seqbs{} algorithm using a \randompattern{} pattern generator.
    The data generators are described at the beginning of \cref{sec:exp}.
                \label{fig:scal}
	}
	\end{minipage}%
\end{figure} 

%% file: related.tex
\section{Related Work}\label{sec:related}

LIS is widely studied both sequentially and in parallel.
Sequentially, various algorithms have been proposed~\cite{yang2005fast,bespamyatnikh2000enumerating,fredman1975computing,Knuth73vol3,schensted1961longest,crochemore2010fast}.
and the classic solution uses $O(n\log n)$ work.
This is also the lower bound~\cite{fredman1975computing} w.r.t.\ the number of comparisons.
In the parallel setting, LIS is studied both as general dynamic programming~\cite{BG2020,chowdhury2008cache,tang2015cache,galil1994parallel} or on its own~\cite{krusche2009parallel,seme2006cgm,thierry2001work,nakashima2002parallel,nakashima2006cost,alam2013divide,tiskin2015fast}.
However, we are unaware of any work-efficient LIS algorithm with non-trivial parallelism ($o(n)$ or $\tilde{O}(\lislength)$ span).
Most existing parallel LIS algorithms introduced a polynomial overhead in work~\cite{galil1994parallel,krusche2009parallel,seme2006cgm,thierry2001work,nakashima2002parallel,nakashima2006cost},
and/or have $\tilde{\Theta}(n)$ span~\cite{alam2013divide,BG2020,chowdhury2008cache,tang2015cache} (many of them~\cite{BG2020,chowdhury2008cache,tang2015cache} focused on improving the I/O bounds).
The algorithm in~\cite{krusche2010new} translates to $O(n\log^2 n)$ work and $\tilde{O}(n^{2/3})$ span, but it relies on complicated techniques for Monge Matrices~\cite{tiskin2015fast}.
Most of the parallel LIS algorithms are complicated and have no implementations.
We are unaware of any parallel LIS \emph{implementation} with competitive performance to the sequential $O(n\log \lislength)$ or $O(n\log n)$ algorithm.

Many previous papers
propose general frameworks to study dependencies in sequential iterative algorithms to achieve parallelism~\cite{blelloch2020optimal,blelloch2016parallelism,blelloch2012internally,shen2022many}.
Their common idea is to (implicitly or explicitly) traverse the \DG{}.
There are two major approaches, and both have led to many efficient algorithms. The first one is edge-centric~\cite{blelloch2016parallelism,blelloch2018geometry,blelloch2020optimal,blelloch2020randomized,jones1993parallel,hasenplaugh2014ordering,BFS12,fischer2018tight,shen2022many}, which identifies the ready objects  by 
processing the successors of the newly-finished objects.
The second approach is vertex-centric~\cite{shun2015sequential,blelloch2012internally,pan2015parallel,tomkins2014sccmulti,shen2022many},
which checks all unfinished objects in each round to process the ready ones.
However, none of these frameworks directly enables work-efficiency for parallel LIS.
The edge-centric algorithms evaluate all edges in the \dg{}, giving $\Theta(n^2)$ worst-case work for LIS.
The vertex-centric algorithms check the readiness of all remaining objects in each round and require $\lislength$ rounds,
meaning $\Omega(n\lislength)$ work for LIS.
The \swgs{} algorithm~\cite{shen2022many} combines the ideas in edge-centric and vertex-centric algorithms.
\swgs{} has $O(n\log^3 n)$ work \whp{} and is round-efficient ($\tilde{O}(k)$ span) using $O(n\log n)$ space.
It is sub-optimal in work and space.
Our algorithm improves the work and space bounds of \swgs{} in both LIS and WLIS.
Our algorithm is also simpler and performs much better than \swgs{} in practice.

The \veb{} tree was proposed by van Emde Boas in 1977~\cite{van1977preserving},  and has been widely used in sequential algorithms, such as dynamic programming~\cite{eppstein1988speeding, galil1992dynamic, hunt1977fast, chan2007efficient,inoue2018computing,narisada2017computing}, computational geometry~\cite{snoeyink1992two,chiang1992dynamic, claude2010range, afshani2017independent}, data layout~\cite{bender2000cache, ha2014models, umar2013deltatree,van1976design}, and others~\cite{lipski1981efficient, gawrychowski2015efficiently, koster2001treewidth, lipski1983finding, akbarinia2011best}. 
However, to the best of our knowledge, there was no prior work on supporting parallelism on \veb{} trees.

%% file: conclusion.tex
\section{Conclusion}
In this paper, we present the first work-efficient parallel algorithm for the longest-increasing subsequence (LIS) problem
that has non-trivial parallelism ($\tilde{O}(k)$ span for an input sequence with LIS length~$k$).
Theoretical efficiency also enables a practical implementation with good performance.
We also present algorithms for parallel \veb trees and show how to use them to improve the bounds for the weight LIS problem.
As a widely-used data structure, we believe our parallel \veb tree is of independent interest, and we plan to explore other applications as future work.
Other interesting future directions include
achieving work-efficiency and good performance for WLIS in parallel
and designing a work-efficient parallel LIS algorithm with $o(n)$ or even a polylogarithmic span. 

%% file: appendix.tex
\section{Output the LIS from Algorithm \ref{algo:lis}}
\label{app:outputlis}
Following the terminology in DP algorithms, we call each $\mathdp[i]$ a \defn{state}.
When we attempt to compute state $i$ from state $j<i$,
we say $j$ is a \defn{decision} for state $i$.
We say $j$ is the \defn{best decision} of the state $i$ if $j=\arg \max_{t:t<i,A_{t}<A_i}\mathdp[t]$.

To report a specific LIS of the input sequence,
we can slightly modify \cref{algo:lis} to compute the \emph{best decision} $d[i]$ for the $i$-th object.  
Namely, $A_{d[i]}$ is $A_i$'s previous object in the LIS ending at $A_i$.
Then starting from an object with the largest rank, we can iteratively find an LIS in $O(\lislength)$ work and span.
Our idea is based on a simple observation:

\begin{lemma}\label{lem:findlis}
For an object $A_i$ with rank $r$,
let $A_{d[i]}$ be the smallest object with rank $r-1$ before $A_i$,
then $d[i]$ is the best decision for state $i$, i.e., $d[i]=\arg\max_{j:j<i,A_j<A_i}\mathdp[j]$.
\end{lemma}
\begin{proof}
  Since $\rank(A_i)=r$ and $\rank(A_{d[i]})=r-1$, clearly $\mathdp[d[i]]=\max_{j:j<i,A_j<A_i}\mathdp[j]$.
  We just need to show that $A_{d[i]}<A_i$,
  such that it is a candidate of $A_i$'s previous object in the LIS.
  Assume to the contrary that $A_{d[i]}\ge A_i$.
  Note that $A_{d[i]}$ is the smallest object with rank $r-1$ before $A_i$.
  Therefore, for any $A_j$ where $\rank(A_j)=r-1$ and $j<i$, we have $A_j\ge A_{d[i]}\ge A_i$.
  Hence, $A_i$ cannot obtain a DP value of $r$, which leads to a contradiction.
\end{proof}

We then show how to identify the best decision $d[i]$ for each $A_i$.
First of all, when executing \processfrontier{} in round $r$,
we can also output all objects of rank $r$ into an array in parallel.
This can be performed by traversing $T$ twice, similar to the function \prefixminfunc{}.
In the first traversal, we mark all \prefixmin{} objects to be extracted in the frontier.
On the way back of the recursive calls, we also compute the number of such \prefixmin{} objects in each subtree.
We call this the \defn{effective size} of this subtree.
The effective size at the root is exactly the frontier size $m_r$, and we can allocate an array $\ff_r[1..m_r]$ for the frontier.
Then we traverse the tree once again.
We recursively put objects in the left and right subtrees into $\ff_r[\cdot]$ in parallel.
Let the effective size of the left subtree be $s$.
Then the right tree can be processed in parallel to put objects in $\ff_r$ from the $(s+1)$-th slot.
We then show that the objects in each frontier $\ff_r$ are non-increasing.

\begin{lemma}\label{lem:frontierdec}
Given a sequence $A$ and any integer $r$, let $\ff_r$ be the subsequence of $A$ with all objects with rank $r$.
Then $\ff_r$ is non-increasing for all $r$.
\end{lemma}
\begin{proof}
Assume to the contrary that there exist $A_i$ and $A_j$,
s.t. $\rank(A_i)=\rank(A_j), i<j, A_i<A_j$.
This means that we can add $A_j$ after $A_i$ in an LIS, so $\mathdp[j]$ is at least $\mathdp[i]+1$.
This leads to a contradiction since $A_i$ and $A_j$ have the same rank (DP values).
Therefore, each frontier $\ff_r$ is non-increasing.
\end{proof}

Based on \cref{lem:frontierdec}, the \emph{smallest} object with rank $r-1$ before $A_i$ is also the \emph{last} object with rank $r-1$ before $A_i$.
Therefore, after we find the frontier $\ff_r$,
we can merge $\ff_r$ with $\ff_{r-1}$ based on the index, such that each object in $\ff_r$ can find the last object before it with rank $r-1$.
Using a parallel merge algorithm~\cite{JaJa92}, this part takes $O(\log n)$ span in each round and $O(n)$ total work in the entire algorithm.

\section{Additional Proofs}
\subsection{Proof of \cref{lem:prefixmin}}
\label{app:prefixminproof}

\begin{proof}
  We will prove the theorem inductively.

  We first show that the base case is true.
  We start with the ``if'' direction.
  For a \prefixmin{} object $A_i$, $A_i$ is the smallest object among $A_{1..i}$.
  Thus, there exists no $A_j$ such that $a_j<a_i, j<i$. Based on \cref{eqn:lis}, $\mathdp[i]=1$.

  For the ``only-if'' direction, note that if $\mathdp[i]=1$, the LIS ending at $A_i$ has length 1.
  Assume to the contrary that there exists $j<i$ such that $A_j<A_i$. Then the LIS ending at $A_i$ is at least 2,
  which contradicts the assumption. Therefore, $\mathdp[i]=1$ also indicates that $A_i$ is the smallest element among $A_{1..i}$.

  Assume for all $r<t$, Lemma \ref{lem:prefixmin} is true. We will prove that the lemma is true for $r=t$.
  We first show the ``if'' direction.
  Based on the inductive hypothesis, after removing all objects with rank smaller than $t$, a (remaining) object $A_i$ must have $\rank(A_i)\ge t$.
  Since $A_i$ is the smallest object among all \emph{remaining} objects in $A_{1..i}$,
  all objects in $A_{1..i}$ smaller than $A_i$ must have been removed and thus have rank at most $t-1$.
  From \cref{eqn:lis}, $\rank(A_i)\le t-1+1=t$. Therefore, $\rank(A_i)=t$.

  For the ``only-if'' direction, note that if $\mathdp[i]=t$, $A_i$ has rank $t$ and must be remaining after removing objects with ranks smaller than $t$.
  We will then prove that $A_i$ is a \prefixmin{} object.
  Let $S=\{A_j : A_j<A_i, j<i\}$.
  We first show that $\rank(x)<t$ for all $x\in S$.
  This is because $A_i$ depends on all objects in $S$---if any object $x\in S$ has $\rank(x)\ge t$, $A_i$ must have rank at least $t+1$.
  This means that all objects in $S$ must have been removed. In this case, $A_i$ must be no larger than all remaining objects before it.
  Therefore, $\mathdp[i]=t$ indicates that $A_i$ is a \prefixmin{} object after removing all objects with rank smaller than $t$.
\end{proof}

\subsection{Solving the Recurrences in Thm. \ref{insertTheorem}}
\label{app:solverecurrence}


\begin{lemma}
	Recurrence \ref{eqn:vebwork} solves to $O(m\log\log U)$.
\end{lemma}
\begin{proof}
    We will inductively prove that $w(u,m)=c\cdot m\log\log u$ for some constant $c$.
	When $m=1$, the recurrence is the same as the one for single element insertion, which solves to $O(\log \log U)$.

    We now consider $m>1$.
    We can easily check that the conclusion is true for $u=1$ or $u=2$.
    Now assume the solution holds for $\sqrt{u}\leq t$, where $t\leq U$, then it holds for $\sqrt{u}\leq t^2$ inductively, since:

	\begin{align*}
		W(u,m)&=c\cdot \sum_{i=0}^{\sqrt{u}}m_i\log\log\sqrt u + c\cdot m\\
		&= c\cdot m (\log\log \sqrt u +1)\\
		&= c\cdot m \log\log u
	\end{align*}

	Picking $u=U$ gives the solution to Recurrence \ref{eqn:vebwork}.
\end{proof}

\begin{lemma}
	Recurrence \ref{eqn:vebspan} solves to $O(\log U)$.
\end{lemma}
\begin{proof}

	Let $u=2^m$, so $m=\log u$. Since $S(u,\cdot)=S(\sqrt{u},\cdot)+O(\log u)$, we have $S(2^m,\cdot)=S\left(2^\frac{m}{2},\cdot\right)+O(m)$.
	Let $T(m,\cdot)=S(2^m,\cdot)$, we have $T(m,\cdot)=T\left(\frac{m}{2},\cdot\right)+O(m)$.
	Using Master Theorem, it can be deduced that $T(m,\cdot)=O(m)$. Thus, $S(2^m,\cdot)=O(m)$, and $S(u,\cdot)=O(\log u)$.

	Picking $u=U$ gives the solution to Recurrence \ref{eqn:vebspan}.
\end{proof}

\subsection{Proof of Thm. \ref{deletetheorem} }
\label{app:deleteproof}

\begin{proof}[Proof of \cref{deletetheorem}]
  First note that the process of \cref{vebdelete} is almost the same as \cref{vebinsert}, and the span recurrence (Eqn. $\ref{eqn:vebspan}$) still holds.
  Therefore the span of \cref{vebdelete} is also $O(\log U\log\log U)$.
  We then analyze the work.
  We will show that in each recursive call of \textsc{\batchdeletehelper}, we will spend constant work on each key in $B$,
  and each key in $B$ will be involved in $O(\log\log u)$ recursions.

  We first analyze the process to restore the $\min/\max$ values.
  If we need to update the $\min/\max$ values, we may need to delete a key $k$ sequentially from the current subtree (\cref{line:seqdelete}).
  Note that this costs $O(\log\log u)$ work, and must indicate that one key has been excluded from the batch (in \cref{line:remove}).
  Although this key $B.\mmin$ will not occur in any recursive calls, we will charge the $O(\log \log u)$ cost for sequential deletion to this key as if it is involved in
  $O(\log \log u)$ levels of recursion, and the conclusion holds.
  In function \textsc{SurvivorRedirect}, \cref{line:redirect:find} has $O(\log \log u)$ work. As mentioned, we will charge this cost to the previous sequential deletion.
  The rest part of \textsc{SurvivorRedirect} has work $O(m)$ total work, which is $O(1)$ work per key in $B$ (excluding the one that has been removed in \cref{line:seqdelete}).
  Lastly, we note that any removal of the keys from $B$ must be removing either the minimum or maximum key in $B$.
  Since $B$ is a sorted array, we can do this in constant time by moving the head or tail pointer of the array.

  We then analyze the cost to recursively dealing with high- and low-bits. Note that to construct the survivor mappings for low- and high-bits,
  the two functions \textsc{SurvivorLow} and \textsc{SurvivorHigh} cost $O(m)$ work in total, which is also $O(1)$ per key in $B$ (excluding the possible one that has been removed in \cref{line:seqdelete}). Therefore, in each recursion will spend constant work on each key in $B$.

  We then analyze the recursion depths.
  There are two types of recursive calls: \cref{line:recursionhigh} on the high-bits and \cref{line:recursionlow} ($O(\sqrt{u})$ of them in total) on the low-bits.
  We will carefully charge the work such that each key in $B$ is involved in at most one of the recursive calls.
  For the recursive call on high-bits (\cref{line:recursionhigh}), note that only $O(|H'|)$ high-bits are involved, where $H'$ is the set of high-bits
  which have all keys in their clusters deleted. For any $h\in H'$, we will charge this work on the minimum key in $L[h]$ (i.e., the minimum key with high-bit $h$).
  Therefore, each key $x\in B$ will satisfy one of the following two conditions:

  \begin{enumerate}
    \item $x$ is the minimum value among all keys in $B$ with the same high-bit $h=\high(x)$, and all keys with high-bit $h$ in $\vt$ will be deleted. In this case, $x$ will
    go through the recursive call on \cref{line:recursionhigh}. It will also be involved in the corresponding recursive call on \cref{line:recursionlow},
    but it will be handled by Lines \ref{line:updateminstart}--\ref{line:updateminend}, and no further recursive calls are involved.
    \item Otherwise, $x$ will only go through the corresponding recursive call on \cref{line:recursionlow}.
  \end{enumerate}

  Therefore, if we consider the current recursive call has the universe size $u$, for each key $x\in B$,
  the relevant recursive call of the next level must have universe size $\sqrt{u}$.
  Let $D(u)$ be the number of recursive calls a key $x\in B$ is involved in, we have the recurrence

  $$D(u)=D(\sqrt{u})+O(1)$$

  Solving this recurrence we can get $D(u)=\log\log u$.

  Therefore, in each recursive call of \batchdeletehelper{}, we spend $O(1)$ work on average per key in $B$, and each key in $B$ will be involved in $O(\log\log u)$ recursive calls.
  Combining them we know that the total work of \batchdelete{} itself is $O(m\log\log U)$.
\end{proof}

\input{dominatedByVEBproof}

\input{space.tex} 

%% file: dominatedByVEBproof.tex
\input{BatchRangeQuery}
\section{Range Query for \veb{} Trees}
\label{app:range}

\cref{vebBatchRangeQuery} finds a batch of sorted keys $B \subseteq \univ$ in vEB tree $\vt$ with keys in range $[k_L,k_R]$. The process is similar to binary search except that we need to store the result of each recursive call in a binary tree, call it a \defn{result tree}, and flatten it in the end.
We first set $\kl$ and $\kr$ as the first and last key in the range and start the binary search to build the result tree in function \textsc{BuildTree}.
In the base cases, if $k_L$ equals to $k_R$, the algorithm simply returns $k_L$ in a tree node.
If $k_L$ is greater than $k_R$, return an empty tree node \texttt{NIL}.
Otherwise, we first find the predecessor of the middle key $\midd=\vebpred(\vt,\lceil (\kl+\kr)/2\rceil)$ and include it in the result.
To do this, we create a tree node $\tau$ and store the value of $\midd$, and then recursively deal with the two ranges $[\kl,\midd)$ and $[\midd,\kr)$ in parallel, and store them as the left and right children of $\tau$.
Note that to handle the two sub-ranges, we need to set the ranges as $[k_L,\vebpred(\vt,\midd)]$ and $[\vt,\vebsucc(\vt,\midd),k_R]$,
to ensure that the endpoints must present in $\vt$.
Finally, the return value from \cref{line:range:returntree} is a tree node, which is a root
to a tree connecting all keys in the queried range.
We can simply flatten the tree to an array as the output (\cref{line:range:flatten}).

Next, we analyze the cost of the \rangequery{} function. We present the result in the following theorem.

\begin{theorem}\label{thm:veb:rangequery}
	Given a range $[k_L,k_R]$, finding a batch of sorted keys in a \veb{} tree with keys in this range can be finished in $\Theta((1+m)\log\log U)$ work and $\Theta(\log U\log\log U)$ span, where $m$ is the size of output array and $U=|\univ|$ is the universe size.
\end{theorem}
\begin{proof}
We first analyze the span. The cost of calling the predecessor and successor in \veb{} tree is $O(\log{\log{U}})$.
Since \textsc{BuildTree} uses divide-and-conquer (deal with the two subproblems in parallel) and binary search the range to find the subproblems,
the recursion depth for \textsc{BuildTree} is $O(\log U)$. Therefore the span of the algorithm is $O(\log U\log\log U)$

We then analyze the work. In each recursive call of \textsc{BuildTree}, we will call $\vebpred{}$ and $\vebsucc{}$ a constant number of times.
Therefore, we just need to analyze the number of invocations to the \textsc{BuildTree} function during the entire \rangequery{} algorithm.
Note that every recursive all at least create one tree node, either a tree node with a value stored in it (a value included in the queried range),
or a \texttt{NIL} node, which should be an external node of the result tree.
Considering that the result tree has size $m$, the number of invocations to \textsc{BuildTree} should be $O(m)$.
Therefore, the total work of the algorithm is $O(m\log\log U)$.
Considering the cost of calling the predecessor and successor in \cref{line:veb:range:start1,line:veb:range:start2} is $O(\log{\log{U}})$,
the total work is $O((1+m)\log\log U)$.
\hide{we can derive the following recurrence of the span of Range Query:
\begin{align}
	S(n)=S(n/2)+\Theta(\log{\log{U}})
\end{align}

Since the largest possible value of $n$ is $U$, we can derive that the span of \cref{vebBatchRangeQuery} is $\Theta(\log{U}\log{\log{U}})$. Moreover, because at least one node is inserted into the binary tree every time we call the node's predecessor in \textsc{BuildTree($\vt,\kl,\kr$)}, the work of \cref{vebBatchRangeQuery} is bounded by $\Theta(m'\log{\log{U}})$, where $m'$ is the size of the output.
}
\end{proof}

\input{dominatedByVEB}

\SetKw{MIN}{min}
\section{The \domby{} Function for \rangeveb{} Trees}\label{app:dominatedby}
\hide{
\begin{theorem}\label{BatchDominatedByTheorem}
	Finding a batch of sorted elements dominated by another batch of sorted elements in a \veb{} tree can be finished in $\Theta(m'\log\log U)$ work and $\Theta(\log m\log U\log\log U)$ span, where $m$ is the size of input batch, $m'$ is the size of output batch and $U=|\univ|$ is the universe size.
\end{theorem}
}

\begin{figure}[t]
	\centering
	\includegraphics[width=.85\columnwidth]{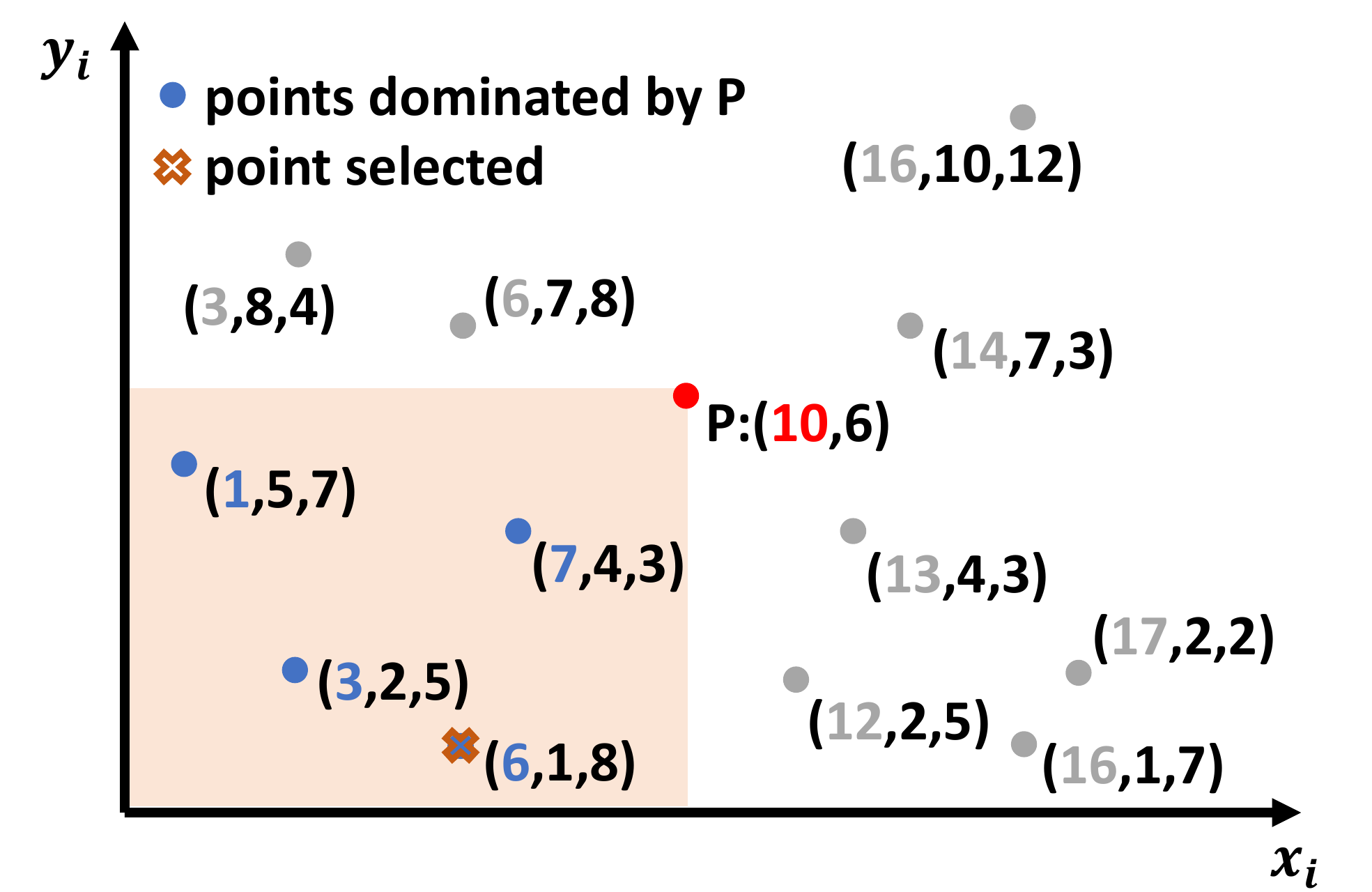}
	\caption{\small
		\textbf{An illustration of \dommax.} The red dot $P$ represents the query point. The blue dots represent the set of points within the query range of \dommax{}. The point noted by ``x'' represents the point with the highest score (\dpvalue{}) returned by this query. All points except $P$ are denoted as $\langle x_i,y_i,\mathdp_i\rangle$.
	}\label{fig:dominantmax}
\vspace{-.5em}
\end{figure}
\begin{figure*}[t]
\begin{minipage}{.85\columnwidth}
	\centering
	\includegraphics[width=\columnwidth]{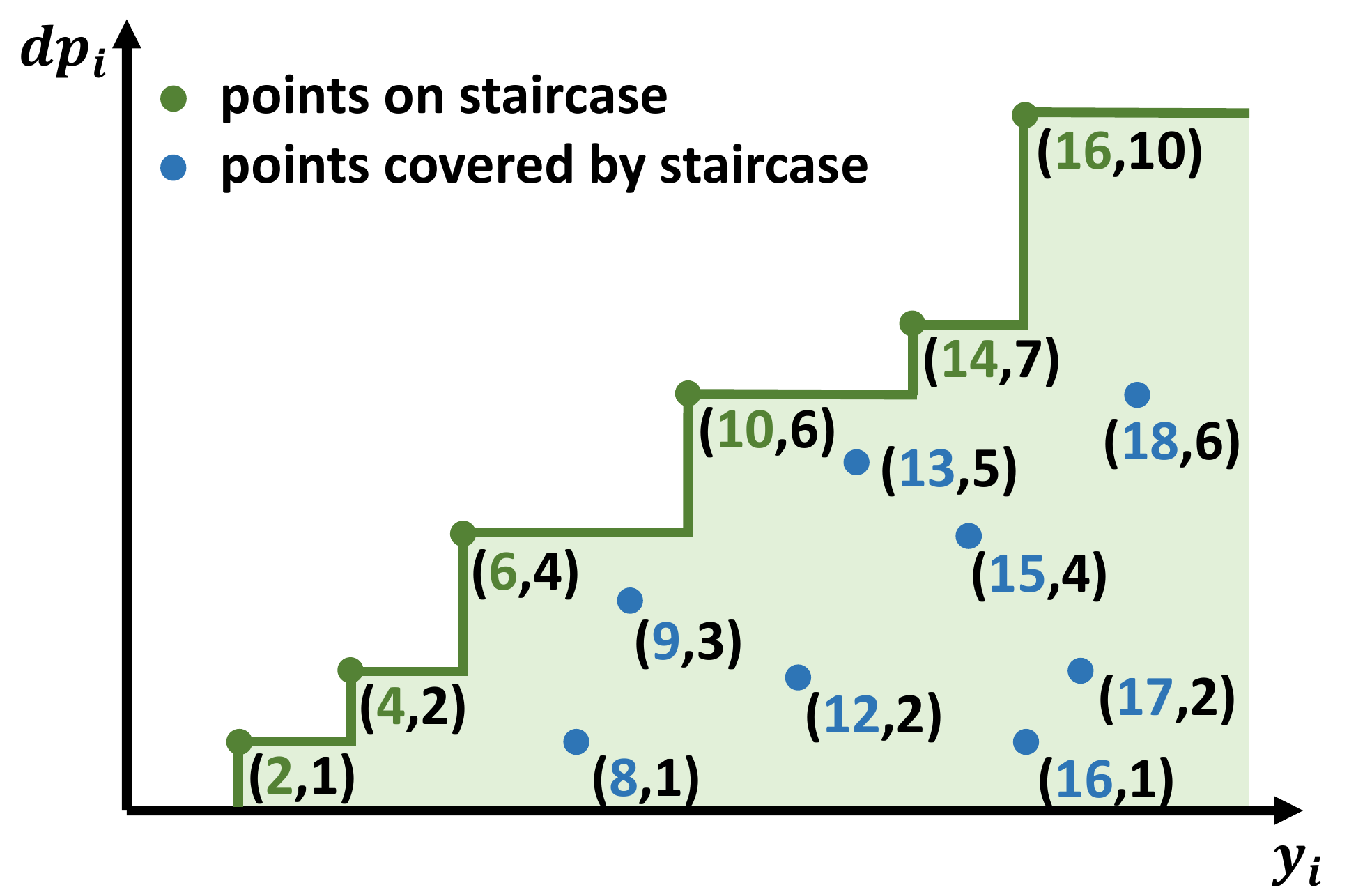}
	\caption{\small
		\textbf{An illustration of  staircase.} The green dots represent the set of points on staircase, which will be maintained by a \monoveb{} tree. The blue dots represent the set of points covered by the staircase, and will not be in the \monoveb{} tree.
	}\label{fig:staircase}
\end{minipage}\hfill
\begin{minipage}{1.15\columnwidth}
	\centering
	\includegraphics[width=.765\columnwidth]{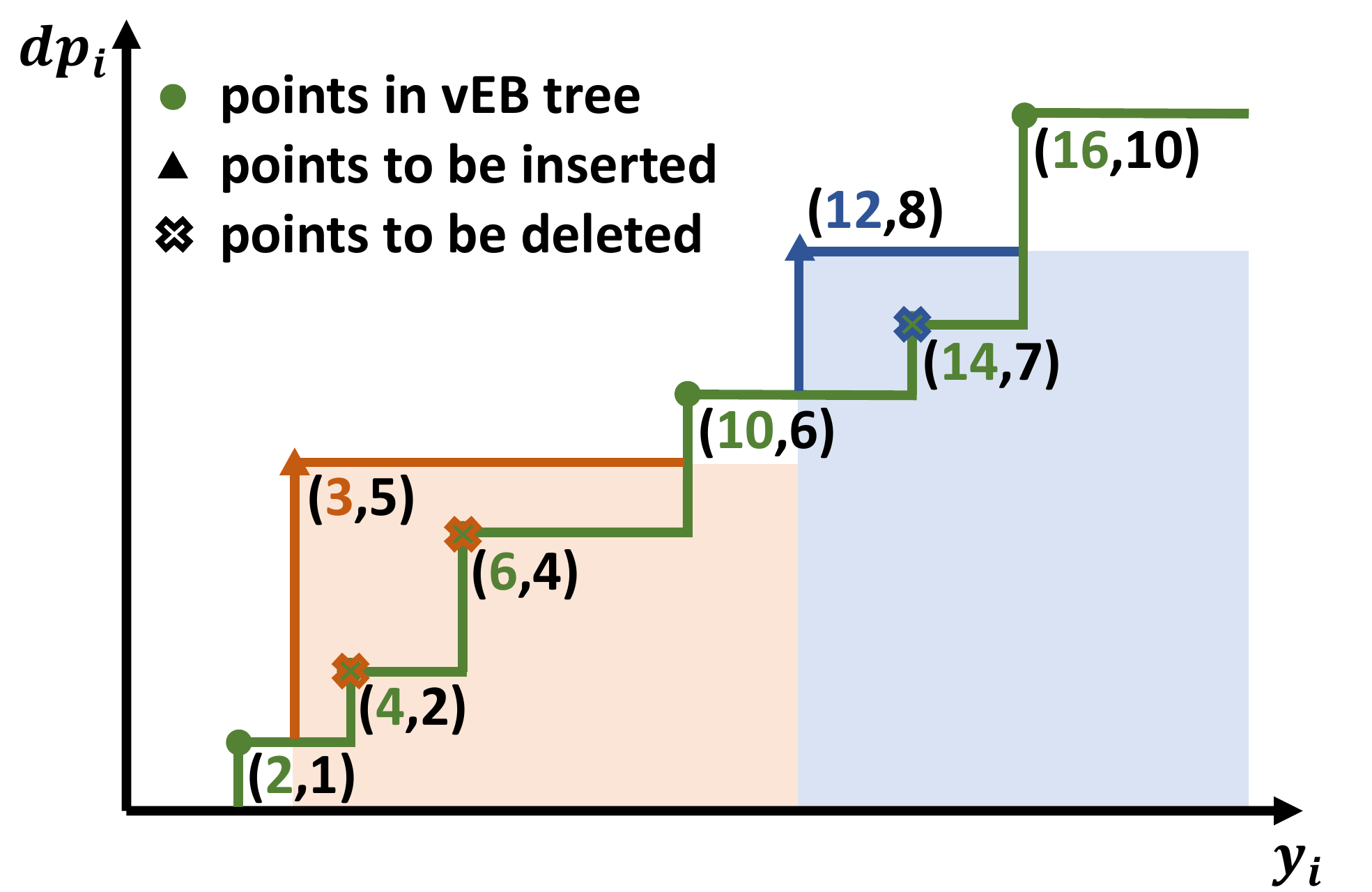}
	\caption{\small
		\textbf{An illustration of  \cref{vebdominatedby}.} The green dots represent the set of points $S_v$ in the \veb{} tree $\vt$ and the lines connecting them represents the staircase of $S_v$. The triangle points represent the points that will be inserted.  The shaded regions correspond to the points on the staircase that will be deleted from the \monoveb{} tree.}\label{fig:coveredby}
\end{minipage}
\vspace{-.5em}
\end{figure*}
Now we describe the \domby{} function. We will implement it using the \rangequery{} function introduced in \cref{app:range}.
We present the pseudocode in \cref{vebdominatedby}.
The function \textsc{\domby($\bin,\vt$)} finds a batch of sorted elements $\bout\subseteq \vt$, which are \dominated{} by a batch of sorted points $\bin\subseteq \univ$.
Recall that each point $\langle x_i, y_i, \mathdp_i\rangle$ is a triple. The \veb{} tree $\vt$ is a \monoveb{} tree. The points in $\vt$ are keyed on $y_i$.
For two points $p_1=\langle x_1,y_1,\mathdp_1\rangle$ and $p_2=\langle x_2,y_2,\mathdp_2\rangle$, we say $p_1$ \defn{\dominates} $p_2$
if $y_1<y_2$ and $\mathdp_1\geq \mathdp_2$.
The \monoveb{} tree maintains a set of points that do not \dominate{} each other. In other words, with increasing $y$ values, their \dpvalue{s} must also be increasing.

To find all points \dominated{} by the points in $\bin$, we first observe that each point in $\bin$ \dominates{} a set of consecutive points in $\vt$.
Take the $i$-th element $\bin[i]$ as an example, it \dominates{} all points in $\vt$ such that they have keys ($y$ coordinates) larger than $\bin[i].y$,
but have \dpvalue{s} smaller than or equal to $\bin[i].\mathdp$. Since \dpvalue{s} in $\vt$ are monotonically increasing with the key ($y$ coordinates),
such points \dominated{} by $\bin[i]$ must be a range of points. Such range (if any) starts with the successor $s$ of $\bin[i].y$ (if the \dpvalue{} of $s$ is no larger than $\bin[i].\mathdp$).
Also, for any point in $\vt$ that is \dominated{} by more than one points in $\bin$, we only need to consider it once.
Therefore, when finding the points \dominated{} by $\bin[i]$, we can just consider the range from $\bin[i]$ to $\bin[i+1]$.
More precisely, we should consider range $[s,e]$, where $s=\vebsucc(\vt,\bin[i].y)$, and $e=\vebpred(\vt,\bin[i+1].y)$.
For the last point in $\bin$, the upper bound of the range can be set to $U$.
This gives an initial range of points that may be \dominated{} by $\bin[i]$.
Note that the lower bound $s$ of the range is tight - either the range starts from $s$ (if the \dpvalue{} of $s$ is no larger than $\bin[i].\mathdp$), or $\bout$ is empty.
We only need to find a tight upper bound $e'\le e$, such that $e'$ is the last point in $\vt$ that is \dominated{} by $\bin[i]$.

For each point in $\bin[i]$, the range that it is responsible for searching is independent with other ranges. Therefore, we can deal with all such ranges in parallel.
For $\bin[i]$, we first use function \textsc{FindIndex} to find the tight upper bound $e'$,
which is the last point with \dpvalue{} smaller than $\bin[i].\mathdp$.
In particular, \textsc{FindIndex}$(\vt,\mathdp^*,s,e)$ means to search the key range $[s,e]$ in \monoveb{} tree $\vt$,
and report the last point in this range with \dpvalue{} no larger then $\mathdp^{*}$.
Note that the $\mathdp$ values in the \monoveb{} tree $\vt$ is monotonically increasing, we can use an binary-search-based algorithm to find $e'$.
Instead of directly starting a binary search, this function will first repeatedly apply $\vebsucc$ on the starting point $s$ for $\log U$ times. 
We do this to guarantee work-efficiency of the algorithm (see discussions later).
We stop this process of applying $\vebsucc{}$ to $s$ when a point with \dpvalue{} larger than $\mathdp^*$ is found, and the previous point must be the last point with \dpvalue{} smaller than $\mathdp^{*}$.
Another possible case is that when we call $\vebsucc$, we reach the upper bound $e$. In this case, $e$ is the last point with \dpvalue{} smaller than $\mathdp^{*}$, and we can return $e$.
If after chasing $\vebsucc{}$ for $\log U$ times, the \dpvalue{} of $s$ is still smaller than $\mathdp^{*}$, we can start a regular binary search to find $e'$, which is presented as function \textsc{BinarySearch}.
Finally, the function \textsc{FindInterval} will give the tight upper bound of the range of points \dominated{} by $\bin[i]$. We will store this return value as $e'$ (\cref{line:findupperbound}).

Finally, to get the list of points \dominated{} by $\bin[i]$, we can directly use the upper- and lower-bounds $s$ and $e'$, and call the \rangequery{} function to collect them in list $D[i]$. The final result will be obtained by combine all such lists $D[1..|\bin|]$.
\cref{fig:coveredby} present an example of the \domby{} function.

We now analyze the cost of \domby{}.

\begin{theorem}
  The \domby{}$(\bin, \vt)$ function has $O(m\log \log U)$ work and $O(\log U\log \log U)$ span, where $U$ is the universe size of the \veb{} tree $\vt$, and $m$ is the total size of the input and output batch $|\bin|+|\bout|$, i.e., the total number of points in the input $\bin$ and those \dominated{} by any point in $\bin$.
\end{theorem}
\begin{proof}
  We start with the analysis of the work. We will show that, when processing $\bin[i]$,
  the work of the \textsc{FindIndex} function and the corresponding \rangequery{} function is $O((1+m')\log\log U)$,
  where $m'=|D[i]|$ is number of points returned by the \rangequery{} function.

  From \cref{thm:veb:rangequery}, we know that \rangequery{} has $O(m'\log\log U)$ work, so we will focus on analyzing the work of \textsc{FindIndex}.
  We analyze two cases.
  \begin{itemize}
    \item $m'\le \log U$. In this case, \textsc{FindIndex} will return from \cref{line:domby:earlyreturn1} or \cref{line:domby:earlyreturn2}. In this case, we have called $\vebsucc{}$
    $m'$ times, and the total work is $O(m'\log\log U)$.
    \item $m'>\log U$. In this case, \textsc{FindIndex} will first call $\vebsucc{}$ for $\log U< m'$ times, and the total work is $O(\log U\log\log U)=O(m'\log\log U)$.
    The algorithm then performs a binary search. Note that each step in \textsc{BinarySearch} needs to call $\vebsucc{}$ and $\vebpred$, costing $O(\log\log U)$ work.
    In total, the \textsc{BinarySearch} function takes $O(\log U\log\log U)=O(m'\log\log U)$ work.
  \end{itemize}
  The work for calling $\vebpred$ and $\vebsucc$ is also $O(\log\log U)$.
  In summary, the work of one \textsc{FindInterval} has work $O((1+m')\log\log U)$ for output size $m'$.
  Adding all the $|\bin|$ loop rounds together, the total work is $O((|\bin|+|\bout|)\log\log U)$.

  We then analyze the span. The longest dependence chain appears in either the $O(\log U)$ invocations to $\vebsucc$ (the sequential for-loop
  on \cref{line:domby:chaselogu}), or the \textsc{BinarySearch} ($\log U$ steps, each calling $\vebsucc$ and $\vebpred$). Both of them and \rangequery{} have span $O(\log U\log\log U)$.
  Therefore, the span of the algorithm is $O(\log U\log\log U)$.
\end{proof}

\hide{
Denote the $i$-th element in $B_{in}$ as $B_{in}[i]$. Since the $dp$ value of the elements in batch $B_{in}$ is increasing, for element $B_{in}[i]$, the elements that it dominates in range $[\vebsucc(\vt,B_{in}[i+1].x), U]$ must also be dominated by element $B_{in}[i+1]$. Therefore, when finding the elements dominated by $B[i]$, we can only consider range $[\vebsucc(\vt,B_{in}[i].x),\\\vebpred(\vt,B_{in}[i+1].x)]$ (For the $m$-th element, the corresponding range is $[\vebsucc(\vt,B_{in}[i].x),U]$). Because the ranges that each element in $B_{in}$ needs to find are non-overlapping now, the corresponding elements dominated by each element in $B_{in}$ can be found in parallel.

The process of finding the corresponding elements dominated by $B_{in}[i]$ can be divided into two parts. Firstly, \textsc{FindIndex} finds the exact range of the elements dominated by $B_{in}[i]$. Then, \textsc{RangeQuery} finds the batch of sorted elements in this exact range. Assuming the number of elements dominated by $B_{in}[i]$ is $m'[i]$. Here we only need to analyze the cost of \textsc{FindIndex}. From \cref{beginSeqFindIndex} to \cref{endSeqFindIndex}, if $m'[i]\leq \log U$, we will look for this element in turn. Otherwise, we perform a binary search by calling \textsc{BinarySearch}. Therefore, the cost of \textsc{FindIndex} is $\Theta(\MIN(m'[i],\log U)\log\log U)$. According to \cref{BatchRangeQueryTheorem} and the cost of \textsc{FindIndex}, we can get that processing the $i$-th element can be finished in $\Theta(m'[i]\log\log U)$ work and $\Theta(\log U\log\log U)$ span. Since all elements in $B_{in}$ are processed in parallel, it can be obtained that the work of \cref{vebdominatedby} is $\Theta(m'\log\log U)$, and the span is $O(\log m\log U\log\log U)$.
}

%% file: BatchRangeQuery.tex
\begin{algorithm}[t]
	\fontsize{8pt}{9pt}\selectfont
	\caption{The \rangequery{} query for \veb{} tree\label{vebBatchRangeQuery}}
    \SetKwProg{MyFunc}{Function}{}{end}
	\SetKwFor{parForEach}{parallel\_for\_each}{do}{endfor}
    \SetKwFor{inParallel}{in parallel:}{}{}
    \SetKwInOut{Note}{Note}
	\KwIn{A \veb tree $\vt$, a range $[k_L,k_R]$}
	\KwOut{Batch of sorted elements $B$ }
    \Note{For simplicity, here we view the return values of functions $\vebsucc()$ and $\vebpred()$ 
    as just the key of the \veb{} tree, i.e., without the score (\dpvalue{}).
    }
	\vspace{0.5em}
	\SetKw{MIN}{min}
	\SetKw{MAX}{max}
	\SetKw{AND}{and}
    \SetKw{NOT}{not}
	\SetKw{OR}{or}
    \DontPrintSemicolon
	
    \tcp{\rangequery{} returns a batch of sorted keys $B$ in $\vt$ with keys in range $[s,e]$}
    \MyFunc{\upshape{\textsc{RangeQuery($\vt,\kl,\kr$)}}}{
    \label{beginRangeQuery}
        \lIf{$\kl\notin\vt$}{$\kl\gets\vebsucc(\vt,\kl)$\label{line:veb:range:start1}}
        \lIf{$\kr\notin\vt$}{$\kr\gets\vebpred(\vt,\kr)$\label{line:veb:range:start2}}
        $\treenode\gets$\Call{BuildTree}{$\vt,\kl,\kr$}\label{line:range:returntree}\\
        $B\gets$\Call{Flatten}{$\treenode$}\tcp*[f]{Flatten the binary tree into a sorted array}\label{line:range:flatten}\\
        \Return $B$
    }
    \smallskip
    \tcp{\textsc{BuildTree} returns a binary tree containing all keys in $\vt$ in range $[\kl,\kr]$}
    \MyFunc{\upshape{\textsc{BuildTree($\vt,\kl,\kr$)}}}{
        \lIf{$\kl > \kr$}{
            \Return $\texttt{NIL}$
        }
        Let $\treenode$ be a tree node\\
        \lIf{$\kl = \kr$}{
            $\treenode.\val\gets \kl$
        }
        \Else{
            $\midd \gets \vebpred(\vt,\lceil (\kl+\kr)/2\rceil)$\\
            $\treenode.\val{} \gets \midd$\\
            \inParallel{} {
            $\treenode.\mathit{left\_child} \gets$\Call{BuildTree}{$\vt,\kl,\vebpred(\vt,\midd)$}\\
            $\treenode.\mathit{right\_child} \gets$\Call{BuildTree}{$\vt,\vebsucc(\vt,\midd),\kr$}
            }
        }
        \Return $\treenode$
    }
    \label{endRangeQuery}
\end{algorithm}

\hide{
\begin{algorithm}
	\fontsize{8pt}{9pt}\selectfont
	\caption{Batch Range Query for vEB tree\label{vebBatchRangeQuery}}
    \SetKwProg{MyFunc}{Function}{}{end}
	\SetKwFor{parForEach}{parallel\_for\_each}{do}{endfor}
	\KwIn{A veb Tree $\vt$, $m$ non-overlapping ranges $\{[l[i],u[i]]|i\in 1..m\}$}
	\KwOut{Batch of sorted elements $B$ }
	\vspace{0.5em}
	\SetKw{MIN}{min}
	\SetKw{MAX}{max}
	\SetKw{AND}{and}
	\SetKw{OR}{or}
	
    \MyFunc{\upshape{\textsc{BatchRangeQuery($\vt,m,l[1...m],u[1...m]$)}}}{
        \parForEach{$i\in 1...m$}{
            $D[i]\leftarrow$\Call{RangeQuery}{$\vt,l[i],u[i]$}
        }
        $B$ = \Call{Combine}{$D[1...m]$}\tcp*[f]{Combine $D[1...m]$ into array $B_{out}$}\\
        \Return $B$
    }
    \smallskip
    \tcp{\textsc{RangeQuery} returns a batch of sorted elements $B$ in $\vt$ with keys in range $[s,e]$}
    \MyFunc{\upshape{\textsc{RangeQuery($\vt,s,e$)}}}{
    \label{beginRangeQuery}
        $TreeNode\leftarrow$\Call{BuildTree}{$\vt,s,e$}\\
        \tcp{Flatten binary tree into sorted array}
        $B.insert$(\Call{Flatten}{$TreeNode$ })\\
        \Return $B$
    }
    \smallskip
    \tcp{\textsc{BuildTree} returns a binary tree which contains a batch of elements in $\vt$ with keys in range $[s,e]$}
    \MyFunc{\upshape{\textsc{BuildTree($\vt,s,e$)}}}{
        \lIf{$s > e$}{
            \Return $\texttt{NIL}$
        }
        \lIf{$s = e$}{
            $TreeNode.val\leftarrow s$
        }
        \Else{
            $mid \leftarrow \vebpred(\vt,\lceil (s+e)/2\rceil)$\\
            $TreeNode.val \leftarrow mid$\\
            $TreeNode.l \leftarrow$\Call{BuildTree}{$\vt,s,\vebpred(\vt,mid)$}$||$
            $TreeNode.r \leftarrow$\Call{BuildTree}{$\vt,\vebsucc(\vt,mid),e$}\\
        }
        \Return $TreeNode$
    }
    \label{endRangeQuery}
\end{algorithm}
} 

%% file: dominatedByVEB.tex
\begin{algorithm}[t]
	\fontsize{8pt}{9pt}\selectfont
	\caption{The \domby{} Algorithm for vEB tree\label{vebdominatedby}}
    \SetKwProg{MyFunc}{Function}{}{end}
	\SetKwFor{parForEach}{parallel-foreach}{do}{endfor}
    \SetKwInput{Note}{Note}
	\KwIn{Batch of sorted keys $\bin$, A \veb{} Tree $\vt$}
	\KwOut{Batch of sorted keys $\bout$ \dominated{} by $\bin$}
    \Note{Recall that each point in $\bin$ is a point $\langle x_i,y_i,\mathdp_i \rangle$, ordered by $y_i$.
    The \veb{} tree $\vt$ is a \monoveb{} tree keyed on $y$ coordinates of the points, and the $\mathdp$ values in $\vt$ are increasing. 
    For simplicity, here we view the return values of functions $\vebsucc()$ and $\vebpred()$
    as just the key of the \monoveb{} tree, i.e., without the score (\dpvalue{}). 
    Recall that $\mathdp[s]$ is the \dpvalue{} of object with index $s$, which is also the score of the key $s$ in \monoveb{} tree. 
    }
	\vspace{0.5em}
	\SetKw{MIN}{min}
	\SetKw{MAX}{max}
	\SetKw{AND}{and}
	\SetKw{OR}{or}
    \DontPrintSemicolon
	
    \MyFunc{\upshape{\textsc{\domby($\bin,\vt$)}}}{
        $b\leftarrow |\bin|$\\
        $\bin[b+1]\gets U$\\
        \parForEach{$i\gets 1$ to $b$}{
            $s\gets\vebsucc(\vt,\bin[i].y)$\\
            $e\gets\vebpred(\vt,\bin[i+1].y)$\\
            $e'\gets$\Call{FindIndex}{$\vt, \bin[i].\mathdp{},s,e$}\label{line:findupperbound}\\
            $D[i]\gets$\Call{Range}{$\vt,s,e'$}
        }
        $\bout$ = {$\bigcup_{i=1}^{b} D[i]$}\tcp*[f]{Combine $D[1...b]$ into array $\bout$}\\
        \Return $\bout$
    }
    \smallskip
    \tcp{\textsc{FindIndex} returns the index of the last key in range $[s,e]$ whose $\mathdp{}$ value is smaller than $\mathdp{}^{*}$}
    \MyFunc{\upshape{\textsc{FindIndex($\vt,\mathdp^{*},s,e$)}}}{
        \lIf{$\mathdp[s]>\mathdp^{*}$}{
            \Return $\texttt{NIL}$
        }
        \lIf{$s = e$}{
            \Return $s$
        }
        \label{beginSeqFindIndex}
        \For{$i\gets 1$ to $\log{U}$\label{line:domby:chaselogu}}{
            $s\leftarrow\vebsucc(\vt,s)$\\
            \lIf{$\mathdp[s]>\mathdp^{*}$}{
                \Return $\vebpred(\vt,s)$\label{line:domby:earlyreturn1}
            }
            \lIf{$s = e$}{
                \Return $s$\label{line:domby:earlyreturn2}
            }
        }
        \label{endSeqFindIndex}
        \Return {\Call{BinarySearch}{$\vt,\mathdp,s,e$}}
    }
    \smallskip
    \tcp{\textsc{BinarySearch} uses binary search to find the index of the last key in range $[s,e]$ whose $\mathdp{}$ value is smaller than $\mathdp{}^{*}$}
    \MyFunc{\upshape{\textsc{BinarySearch($\vt,dp^{*},s,e$)}}}{
        \lIf{$s = e$}{
            \Return $s$
        }
        $\midd \leftarrow \vebpred(\vt,\lceil (s+e)/2\rceil)$\\
        \If{$\mathdp[\midd]\leq \mathdp^{*}$ }{
            \Return{\Call{BinarySearch}{$\vt,\mathdp^{*},\midd,e$}}
        }
        \Else{
            \Return{\Call{BinarySearch}{$\vt,\mathdp^{*},s,\vebpred(\vt,\midd)$}}
        }
    }
\end{algorithm} 

%% file: space.tex
\section{Making \rangeveb{} Space-Efficient}
\label{app:space}

The space bound for a plain \veb{} tree is $O(U)$.
For a \rangeveb{} tree, since the universe size for each inner \veb{} tree is $O(U)=O(n)$ and $n$ inner trees in total, the total size will be $O(n^2)$.
This is space-inefficient and even initialize the tree may make the WLIS algorithm itself work-inefficient.
However, note that the total number of possible elements in \rangeveb{} is at most $O(n\log n)$, the same as a regular range tree.
In this section, we discuss solutions to make our \rangeveb{} space-efficient.

Many approaches can make a \veb{} tree space-efficient sequentially, i.e., using $O(n)$ space when there are $n$ keys present in the tree.
One of the approaches is to store the cluster arrays using a size-varying hash table instead of an array.
As such, the total space consumption will be proportional to the number of existing keys in the \veb tree.
However, this makes our algorithms and bounds randomized, and also complicates our parallel algorithms.

Here we propose a solution based on \emph{relabeling} all keys in each inner tree.
In particular, we will make the outer tree a perfectly balanced tree, and its shape and organization are fixed and never change during the algorithm.
Once we build the outer tree of the \rangeveb{} tree, 
we know exactly which keys should present in each of the inner tree, and thus its size upper bound.
For a specific tree node $\tau$ in the outer tree with subtree size $n^*$,
its inner tree size is at most $n^*$.
Let $S$ be the set of the $n^*$ points in $\tau$'s subtree.
Note that unlike a range tree, where the inner tree of $\tau$ is an index on the entire set $S$,
the inner tree in \rangeveb{} tree is a \monoveb{} tree.
It maintains the staircase of $S$, which can only be a subset of $S$.
In this case, we will relabel the keys in $S$ to be $0$ to $n^*-1$, and the existing keys in the inner tree.
To do this, we have to preprocess all points, and build the outer tree in the \rangeveb{} tree first.
Then for each point (each object in the input sequence $A$), we find all the $O(\log n)$ inner trees it appears in.
In turn, this gives the set of points for each inner tree, and we relabel each of them.
Based on the results, we will maintain $O(\log n)$ lookup tables for each input object,
which is the new label in each of the inner tree containing it.
For each of the inner tree, we also maintain a lookup table that maps a label to the original input object.
In all the algorithms described on the inner trees of \rangeveb{} tree, we use the new labels as the key, instead of the original key in universe $O(n)$.

For an inner tree at level $l$ (the root is at level $0$), the universe of the inner tree is $O(n/2^{l})$, i.e., at most $O(n/2^{l})$ points are included in this inner tree.
Therefore, the total size of the inner trees is $O(n\log n)$.


%% file: main.bbl

\begin{thebibliography}{80}


\ifx \showCODEN    \undefined \def \showCODEN     #1{\unskip}     \fi
\ifx \showDOI      \undefined \def \showDOI       #1{#1}\fi
\ifx \showISBNx    \undefined \def \showISBNx     #1{\unskip}     \fi
\ifx \showISBNxiii \undefined \def \showISBNxiii  #1{\unskip}     \fi
\ifx \showISSN     \undefined \def \showISSN      #1{\unskip}     \fi
\ifx \showLCCN     \undefined \def \showLCCN      #1{\unskip}     \fi
\ifx \shownote     \undefined \def \shownote      #1{#1}          \fi
\ifx \showarticletitle \undefined \def \showarticletitle #1{#1}   \fi
\ifx \showURL      \undefined \def \showURL       {\relax}        \fi
\providecommand\bibfield[2]{#2}
\providecommand\bibinfo[2]{#2}
\providecommand\natexlab[1]{#1}
\providecommand\showeprint[2][]{arXiv:#2}

\bibitem[Afshani and Wei(2017)]%
        {afshani2017independent}
\bibfield{author}{\bibinfo{person}{Peyman Afshani} {and}
  \bibinfo{person}{Zhewei Wei}.} \bibinfo{year}{2017}\natexlab{}.
\newblock \showarticletitle{Independent range sampling, revisited}. In
  \bibinfo{booktitle}{\emph{esa}}. Schloss Dagstuhl-Leibniz-Zentrum fuer
  Informatik.
\newblock


\bibitem[Akbarinia et~al\mbox{.}(2011)]%
        {akbarinia2011best}
\bibfield{author}{\bibinfo{person}{Reza Akbarinia}, \bibinfo{person}{Esther
  Pacitti}, {and} \bibinfo{person}{Patrick Valduriez}.}
  \bibinfo{year}{2011}\natexlab{}.
\newblock \showarticletitle{Best position algorithms for efficient top-k query
  processing}.
\newblock \bibinfo{journal}{\emph{Information Systems}} \bibinfo{volume}{36},
  \bibinfo{number}{6} (\bibinfo{year}{2011}), \bibinfo{pages}{973--989}.
\newblock


\bibitem[Akhremtsev and Sanders(2016)]%
        {akhremtsev2016fast}
\bibfield{author}{\bibinfo{person}{Yaroslav Akhremtsev} {and}
  \bibinfo{person}{Peter Sanders}.} \bibinfo{year}{2016}\natexlab{}.
\newblock \showarticletitle{Fast Parallel Operations on Search Trees}. In
  \bibinfo{booktitle}{\emph{{IEEE} International Conference on High Performance
  Computing (HiPC)}}.
\newblock


\bibitem[Alam and Rahman(2013)]%
        {alam2013divide}
\bibfield{author}{\bibinfo{person}{Muhammad~Rashed Alam} {and}
  \bibinfo{person}{M~Sohel Rahman}.} \bibinfo{year}{2013}\natexlab{}.
\newblock \showarticletitle{A divide and conquer approach and a work-optimal
  parallel algorithm for the LIS problem}.
\newblock \bibinfo{journal}{\emph{Inform. Process. Lett.}}
  \bibinfo{volume}{113}, \bibinfo{number}{13} (\bibinfo{year}{2013}),
  \bibinfo{pages}{470--476}.
\newblock


\bibitem[Altschul et~al\mbox{.}(1990)]%
        {altschul1990basic}
\bibfield{author}{\bibinfo{person}{Stephen~F Altschul}, \bibinfo{person}{Warren
  Gish}, \bibinfo{person}{Webb Miller}, \bibinfo{person}{Eugene~W Myers}, {and}
  \bibinfo{person}{David~J Lipman}.} \bibinfo{year}{1990}\natexlab{}.
\newblock \showarticletitle{Basic local alignment search tool}.
\newblock \bibinfo{journal}{\emph{Journal of molecular biology}}
  \bibinfo{volume}{215}, \bibinfo{number}{3} (\bibinfo{year}{1990}),
  \bibinfo{pages}{403--410}.
\newblock


\bibitem[Arora et~al\mbox{.}(2001)]%
        {arora2001thread}
\bibfield{author}{\bibinfo{person}{Nimar~S Arora}, \bibinfo{person}{Robert~D
  Blumofe}, {and} \bibinfo{person}{C~Greg Plaxton}.}
  \bibinfo{year}{2001}\natexlab{}.
\newblock \showarticletitle{Thread scheduling for multiprogrammed
  multiprocessors}.
\newblock \bibinfo{journal}{\emph{Theory of Computing Systems (TOCS)}}
  \bibinfo{volume}{34}, \bibinfo{number}{2} (\bibinfo{year}{2001}),
  \bibinfo{pages}{115--144}.
\newblock


\bibitem[Bender et~al\mbox{.}(2000)]%
        {bender2000cache}
\bibfield{author}{\bibinfo{person}{Michael~A Bender}, \bibinfo{person}{Erik~D
  Demaine}, {and} \bibinfo{person}{Martin Farach-Colton}.}
  \bibinfo{year}{2000}\natexlab{}.
\newblock \showarticletitle{Cache-oblivious B-trees}. In
  \bibinfo{booktitle}{\emph{focs}}. IEEE, \bibinfo{pages}{399--409}.
\newblock


\bibitem[Bentley and Friedman(1979)]%
        {bentley1979data}
\bibfield{author}{\bibinfo{person}{Jon~Louis Bentley} {and}
  \bibinfo{person}{Jerome~H Friedman}.} \bibinfo{year}{1979}\natexlab{}.
\newblock \showarticletitle{Data structures for range searching}.
\newblock \bibinfo{journal}{\emph{Comput. Surveys}} \bibinfo{volume}{11},
  \bibinfo{number}{4} (\bibinfo{year}{1979}), \bibinfo{pages}{397--409}.
\newblock


\bibitem[Bespamyatnikh and Segal(2000)]%
        {bespamyatnikh2000enumerating}
\bibfield{author}{\bibinfo{person}{Sergei Bespamyatnikh} {and}
  \bibinfo{person}{Michael Segal}.} \bibinfo{year}{2000}\natexlab{}.
\newblock \showarticletitle{Enumerating longest increasing subsequences and
  patience sorting}.
\newblock \bibinfo{journal}{\emph{Inform. Process. Lett.}}
  \bibinfo{volume}{76}, \bibinfo{number}{1-2} (\bibinfo{year}{2000}),
  \bibinfo{pages}{7--11}.
\newblock


\bibitem[Blelloch et~al\mbox{.}(2022)]%
        {blelloch2022joinable}
\bibfield{author}{\bibinfo{person}{Guy Blelloch}, \bibinfo{person}{Daniel
  Ferizovic}, {and} \bibinfo{person}{Yihan Sun}.}
  \bibinfo{year}{2022}\natexlab{}.
\newblock \showarticletitle{Joinable Parallel Balanced Binary Trees}.
\newblock \bibinfo{journal}{\emph{{ACM} Transactions on Parallel Computing
  (TOPC)}} \bibinfo{volume}{9}, \bibinfo{number}{2} (\bibinfo{year}{2022}),
  \bibinfo{pages}{1--41}.
\newblock


\bibitem[Blelloch et~al\mbox{.}(2020a)]%
        {blelloch2020parlaylib}
\bibfield{author}{\bibinfo{person}{Guy~E. Blelloch}, \bibinfo{person}{Daniel
  Anderson}, {and} \bibinfo{person}{Laxman Dhulipala}.}
  \bibinfo{year}{2020}\natexlab{a}.
\newblock \showarticletitle{ParlayLib - a toolkit for parallel algorithms on
  shared-memory multicore machines}. In \bibinfo{booktitle}{\emph{{ACM}
  Symposium on Parallelism in Algorithms and Architectures (SPAA)}}.
  \bibinfo{pages}{507--509}.
\newblock


\bibitem[Blelloch et~al\mbox{.}(2016)]%
        {blelloch2016just}
\bibfield{author}{\bibinfo{person}{Guy~E. Blelloch}, \bibinfo{person}{Daniel
  Ferizovic}, {and} \bibinfo{person}{Yihan Sun}.}
  \bibinfo{year}{2016}\natexlab{}.
\newblock \showarticletitle{Just Join for Parallel Ordered Sets}. In
  \bibinfo{booktitle}{\emph{{ACM} Symposium on Parallelism in Algorithms and
  Architectures (SPAA)}}.
\newblock


\bibitem[Blelloch et~al\mbox{.}(2012b)]%
        {blelloch2012internally}
\bibfield{author}{\bibinfo{person}{Guy~E. Blelloch}, \bibinfo{person}{Jeremy~T.
  Fineman}, \bibinfo{person}{Phillip~B. Gibbons}, {and} \bibinfo{person}{Julian
  Shun}.} \bibinfo{year}{2012}\natexlab{b}.
\newblock \showarticletitle{Internally deterministic parallel algorithms can be
  fast}. In \bibinfo{booktitle}{\emph{{ACM} Symposium on Principles and
  Practice of Parallel Programming (PPOPP)}}.
\newblock


\bibitem[Blelloch et~al\mbox{.}(2020b)]%
        {blelloch2020optimal}
\bibfield{author}{\bibinfo{person}{Guy~E. Blelloch}, \bibinfo{person}{Jeremy~T.
  Fineman}, \bibinfo{person}{Yan Gu}, {and} \bibinfo{person}{Yihan Sun}.}
  \bibinfo{year}{2020}\natexlab{b}.
\newblock \showarticletitle{Optimal parallel algorithms in the binary-forking
  model}. In \bibinfo{booktitle}{\emph{{ACM} Symposium on Parallelism in
  Algorithms and Architectures (SPAA)}}.
\newblock


\bibitem[Blelloch et~al\mbox{.}(2012a)]%
        {BFS12}
\bibfield{author}{\bibinfo{person}{Guy~E. Blelloch}, \bibinfo{person}{Jeremy~T.
  Fineman}, {and} \bibinfo{person}{Julian Shun}.}
  \bibinfo{year}{2012}\natexlab{a}.
\newblock \showarticletitle{Greedy sequential maximal independent set and
  matching are parallel on average}. In \bibinfo{booktitle}{\emph{{ACM}
  Symposium on Parallelism in Algorithms and Architectures (SPAA)}}.
\newblock


\bibitem[Blelloch and Gu(2020)]%
        {BG2020}
\bibfield{author}{\bibinfo{person}{Guy~E. Blelloch} {and} \bibinfo{person}{Yan
  Gu}.} \bibinfo{year}{2020}\natexlab{}.
\newblock \showarticletitle{Improved Parallel Cache-Oblivious Algorithms for
  Dynamic Programming}. In \bibinfo{booktitle}{\emph{{SIAM} Symposium on
  Algorithmic Principles of Computer Systems (APOCS)}}.
\newblock


\bibitem[Blelloch et~al\mbox{.}(2018)]%
        {blelloch2018geometry}
\bibfield{author}{\bibinfo{person}{Guy~E. Blelloch}, \bibinfo{person}{Yan Gu},
  \bibinfo{person}{Julian Shun}, {and} \bibinfo{person}{Yihan Sun}.}
  \bibinfo{year}{2018}\natexlab{}.
\newblock \showarticletitle{Parallel Write-Efficient Algorithms and Data
  Structures for Computational Geometry}. In \bibinfo{booktitle}{\emph{{ACM}
  Symposium on Parallelism in Algorithms and Architectures (SPAA)}}.
\newblock


\bibitem[Blelloch et~al\mbox{.}(2020c)]%
        {blelloch2016parallelism}
\bibfield{author}{\bibinfo{person}{Guy~E. Blelloch}, \bibinfo{person}{Yan Gu},
  \bibinfo{person}{Julian Shun}, {and} \bibinfo{person}{Yihan Sun}.}
  \bibinfo{year}{2020}\natexlab{c}.
\newblock \showarticletitle{Parallelism in Randomized Incremental Algorithms}.
\newblock \bibinfo{journal}{\emph{J. {ACM}}} (\bibinfo{year}{2020}).
\newblock


\bibitem[Blelloch et~al\mbox{.}(2020d)]%
        {blelloch2020randomized}
\bibfield{author}{\bibinfo{person}{Guy~E. Blelloch}, \bibinfo{person}{Yan Gu},
  \bibinfo{person}{Julian Shun}, {and} \bibinfo{person}{Yihan Sun}.}
  \bibinfo{year}{2020}\natexlab{d}.
\newblock \showarticletitle{Randomized Incremental Convex Hull is Highly
  Parallel}. In \bibinfo{booktitle}{\emph{{ACM} Symposium on Parallelism in
  Algorithms and Architectures (SPAA)}}.
\newblock


\bibitem[Blelloch and Reid-Miller(1998)]%
        {Blelloch1998}
\bibfield{author}{\bibinfo{person}{Guy~E. Blelloch} {and}
  \bibinfo{person}{Margaret Reid-Miller}.} \bibinfo{year}{1998}\natexlab{}.
\newblock \showarticletitle{Fast Set Operations Using Treaps}. In
  \bibinfo{booktitle}{\emph{{ACM} Symposium on Parallelism in Algorithms and
  Architectures (SPAA)}}.
\newblock


\bibitem[Blumofe and Leiserson(1998)]%
        {BL98}
\bibfield{author}{\bibinfo{person}{Robert~D. Blumofe} {and}
  \bibinfo{person}{Charles~E. Leiserson}.} \bibinfo{year}{1998}\natexlab{}.
\newblock \showarticletitle{Space-Efficient Scheduling of Multithreaded
  Computations}.
\newblock \bibinfo{journal}{\emph{{SIAM} J. on Computing}}
  \bibinfo{volume}{27}, \bibinfo{number}{1} (\bibinfo{year}{1998}).
\newblock


\bibitem[Cao et~al\mbox{.}(2023)]%
        {cao2023nearly}
\bibfield{author}{\bibinfo{person}{Nairen Cao}, \bibinfo{person}{Shang-En
  Huang}, {and} \bibinfo{person}{Hsin-Hao Su}.}
  \bibinfo{year}{2023}\natexlab{}.
\newblock \showarticletitle{Nearly optimal parallel algorithms for longest
  increasing subsequence}. In \bibinfo{booktitle}{\emph{{ACM} Symposium on
  Parallelism in Algorithms and Architectures (SPAA)}}.
\newblock


\bibitem[Chan et~al\mbox{.}(2007)]%
        {chan2007efficient}
\bibfield{author}{\bibinfo{person}{Wun-Tat Chan}, \bibinfo{person}{Yong Zhang},
  \bibinfo{person}{Stanley~PY Fung}, \bibinfo{person}{Deshi Ye}, {and}
  \bibinfo{person}{Hong Zhu}.} \bibinfo{year}{2007}\natexlab{}.
\newblock \showarticletitle{Efficient algorithms for finding a longest common
  increasing subsequence}.
\newblock \bibinfo{journal}{\emph{Journal of Combinatorial Optimization}}
  \bibinfo{volume}{13}, \bibinfo{number}{3} (\bibinfo{year}{2007}),
  \bibinfo{pages}{277--288}.
\newblock


\bibitem[Chiang and Tamassia(1992)]%
        {chiang1992dynamic}
\bibfield{author}{\bibinfo{person}{Y-J Chiang} {and} \bibinfo{person}{Roberto
  Tamassia}.} \bibinfo{year}{1992}\natexlab{}.
\newblock \showarticletitle{Dynamic algorithms in computational geometry}.
\newblock \bibinfo{journal}{\emph{Proc. IEEE}} \bibinfo{volume}{80},
  \bibinfo{number}{9} (\bibinfo{year}{1992}), \bibinfo{pages}{1412--1434}.
\newblock


\bibitem[Chowdhury and Ramachandran(2008)]%
        {chowdhury2008cache}
\bibfield{author}{\bibinfo{person}{Rezaul~A. Chowdhury} {and}
  \bibinfo{person}{Vijaya Ramachandran}.} \bibinfo{year}{2008}\natexlab{}.
\newblock \showarticletitle{Cache-efficient dynamic programming algorithms for
  multicores}. In \bibinfo{booktitle}{\emph{{ACM} Symposium on Parallelism in
  Algorithms and Architectures (SPAA)}}. ACM.
\newblock


\bibitem[Claude et~al\mbox{.}(2010)]%
        {claude2010range}
\bibfield{author}{\bibinfo{person}{Francisco Claude}, \bibinfo{person}{J~Ian
  Munro}, {and} \bibinfo{person}{Patrick~K Nicholson}.}
  \bibinfo{year}{2010}\natexlab{}.
\newblock \showarticletitle{Range queries over untangled chains}. In
  \bibinfo{booktitle}{\emph{International Symposium on String Processing and
  Information Retrieval}}. Springer, \bibinfo{pages}{82--93}.
\newblock


\bibitem[Cormen et~al\mbox{.}(2009)]%
        {CLRS}
\bibfield{author}{\bibinfo{person}{Thomas~H. Cormen},
  \bibinfo{person}{Charles~E. Leiserson}, \bibinfo{person}{Ronald~L. Rivest},
  {and} \bibinfo{person}{Clifford Stein}.} \bibinfo{year}{2009}\natexlab{}.
\newblock \bibinfo{booktitle}{\emph{Introduction to Algorithms (3rd edition)}}.
\newblock \bibinfo{publisher}{MIT Press}.
\newblock


\bibitem[Crochemore and Porat(2010)]%
        {crochemore2010fast}
\bibfield{author}{\bibinfo{person}{Maxime Crochemore} {and}
  \bibinfo{person}{Ely Porat}.} \bibinfo{year}{2010}\natexlab{}.
\newblock \showarticletitle{Fast computation of a longest increasing
  subsequence and application}.
\newblock \bibinfo{journal}{\emph{Information and Computation}}
  \bibinfo{volume}{208}, \bibinfo{number}{9} (\bibinfo{year}{2010}),
  \bibinfo{pages}{1054--1059}.
\newblock


\bibitem[Dasgupta et~al\mbox{.}(2008)]%
        {dasgupta2008algorithms}
\bibfield{author}{\bibinfo{person}{Sanjoy Dasgupta},
  \bibinfo{person}{Christos~H Papadimitriou}, {and}
  \bibinfo{person}{Umesh~Virkumar Vazirani}.} \bibinfo{year}{2008}\natexlab{}.
\newblock \bibinfo{booktitle}{\emph{Algorithms}}.
\newblock \bibinfo{publisher}{McGraw-Hill Higher Education New York}.
\newblock


\bibitem[Deift(2000)]%
        {deift2000integrable}
\bibfield{author}{\bibinfo{person}{Percy Deift}.}
  \bibinfo{year}{2000}\natexlab{}.
\newblock \showarticletitle{Integrable systems and combinatorial theory}.
\newblock \bibinfo{journal}{\emph{Notices AMS}}  \bibinfo{volume}{47}
  (\bibinfo{year}{2000}), \bibinfo{pages}{631--640}.
\newblock


\bibitem[Delcher et~al\mbox{.}(1999)]%
        {delcher1999alignment}
\bibfield{author}{\bibinfo{person}{Arthur~L Delcher}, \bibinfo{person}{Simon
  Kasif}, \bibinfo{person}{Robert~D Fleischmann}, \bibinfo{person}{Jeremy
  Peterson}, \bibinfo{person}{Owen White}, {and} \bibinfo{person}{Steven~L
  Salzberg}.} \bibinfo{year}{1999}\natexlab{}.
\newblock \showarticletitle{Alignment of whole genomes}.
\newblock \bibinfo{journal}{\emph{Nucleic acids research}}
  \bibinfo{volume}{27}, \bibinfo{number}{11} (\bibinfo{year}{1999}),
  \bibinfo{pages}{2369--2376}.
\newblock


\bibitem[Dong et~al\mbox{.}(2021)]%
        {dong2021efficient}
\bibfield{author}{\bibinfo{person}{Xiaojun Dong}, \bibinfo{person}{Yan Gu},
  \bibinfo{person}{Yihan Sun}, {and} \bibinfo{person}{Yunming Zhang}.}
  \bibinfo{year}{2021}\natexlab{}.
\newblock \showarticletitle{Efficient Stepping Algorithms and Implementations
  for Parallel Shortest Paths}. In \bibinfo{booktitle}{\emph{{ACM} Symposium on
  Parallelism in Algorithms and Architectures (SPAA)}}.
\newblock


\bibitem[Eppstein et~al\mbox{.}(1988)]%
        {eppstein1988speeding}
\bibfield{author}{\bibinfo{person}{David Eppstein}, \bibinfo{person}{Zvi
  Galil}, {and} \bibinfo{person}{Raffaele Giancarlo}.}
  \bibinfo{year}{1988}\natexlab{}.
\newblock \showarticletitle{Speeding up dynamic programming}. In
  \bibinfo{booktitle}{\emph{{IEEE} Symposium on Foundations of Computer Science
  (FOCS)}}. \bibinfo{pages}{488--496}.
\newblock


\bibitem[Fischer and Noever(2018)]%
        {fischer2018tight}
\bibfield{author}{\bibinfo{person}{Manuela Fischer} {and}
  \bibinfo{person}{Andreas Noever}.} \bibinfo{year}{2018}\natexlab{}.
\newblock \showarticletitle{Tight analysis of parallel randomized greedy MIS}.
  In \bibinfo{booktitle}{\emph{{ACM-SIAM} Symposium on Discrete Algorithms
  (SODA)}}. \bibinfo{pages}{2152--2160}.
\newblock


\bibitem[Fredman(1975)]%
        {fredman1975computing}
\bibfield{author}{\bibinfo{person}{Michael~L Fredman}.}
  \bibinfo{year}{1975}\natexlab{}.
\newblock \showarticletitle{On computing the length of longest increasing
  subsequences}.
\newblock \bibinfo{journal}{\emph{Discrete Mathematics}} \bibinfo{volume}{11},
  \bibinfo{number}{1} (\bibinfo{year}{1975}), \bibinfo{pages}{29--35}.
\newblock


\bibitem[Galil and Park(1992)]%
        {galil1992dynamic}
\bibfield{author}{\bibinfo{person}{Zvi Galil} {and} \bibinfo{person}{Kunsoo
  Park}.} \bibinfo{year}{1992}\natexlab{}.
\newblock \showarticletitle{Dynamic programming with convexity, concavity and
  sparsity}.
\newblock \bibinfo{journal}{\emph{Theoretical Computer Science (TCS)}}
  \bibinfo{volume}{92}, \bibinfo{number}{1} (\bibinfo{year}{1992}).
\newblock


\bibitem[Galil and Park(1994)]%
        {galil1994parallel}
\bibfield{author}{\bibinfo{person}{Zvi Galil} {and} \bibinfo{person}{Kunsoo
  Park}.} \bibinfo{year}{1994}\natexlab{}.
\newblock \showarticletitle{Parallel algorithms for dynamic programming
  recurrences with more than {O}(1) dependency}.
\newblock \bibinfo{journal}{\emph{J. Parallel Distrib. Comput.}}
  \bibinfo{volume}{21}, \bibinfo{number}{2} (\bibinfo{year}{1994}).
\newblock


\bibitem[Gawrychowski et~al\mbox{.}(2015)]%
        {gawrychowski2015efficiently}
\bibfield{author}{\bibinfo{person}{Pawe{\l} Gawrychowski},
  \bibinfo{person}{Shunsuke Inenaga}, \bibinfo{person}{Dominik K{\"o}ppl},
  \bibinfo{person}{Florin Manea}, {et~al\mbox{.}}}
  \bibinfo{year}{2015}\natexlab{}.
\newblock \showarticletitle{Efficiently Finding All Maximal $\backslash \alpha
  $-gapped Repeats}.
\newblock \bibinfo{journal}{\emph{arXiv preprint arXiv:1509.09237}}
  (\bibinfo{year}{2015}).
\newblock


\bibitem[Goodrich and Tamassia(2015)]%
        {goodrich2015algorithm}
\bibfield{author}{\bibinfo{person}{Michael~T Goodrich} {and}
  \bibinfo{person}{Roberto Tamassia}.} \bibinfo{year}{2015}\natexlab{}.
\newblock \bibinfo{booktitle}{\emph{Algorithm design and applications}}.
\newblock \bibinfo{publisher}{Wiley Hoboken}.
\newblock


\bibitem[Gu et~al\mbox{.}(2022a)]%
        {gu2022analysis}
\bibfield{author}{\bibinfo{person}{Yan Gu}, \bibinfo{person}{Zachary Napier},
  {and} \bibinfo{person}{Yihan Sun}.} \bibinfo{year}{2022}\natexlab{a}.
\newblock \showarticletitle{Analysis of Work-Stealing and Parallel Cache
  Complexity}. In \bibinfo{booktitle}{\emph{{SIAM} Symposium on Algorithmic
  Principles of Computer Systems (APOCS)}}. SIAM, \bibinfo{pages}{46--60}.
\newblock


\bibitem[Gu et~al\mbox{.}(2022b)]%
        {gu2022parallel}
\bibfield{author}{\bibinfo{person}{Yan Gu}, \bibinfo{person}{Zachary Napier},
  \bibinfo{person}{Yihan Sun}, {and} \bibinfo{person}{Letong Wang}.}
  \bibinfo{year}{2022}\natexlab{b}.
\newblock \showarticletitle{Parallel Cover Trees and their Applications}. In
  \bibinfo{booktitle}{\emph{{ACM} Symposium on Parallelism in Algorithms and
  Architectures (SPAA)}}. \bibinfo{pages}{259--272}.
\newblock


\bibitem[Gusfield(1997)]%
        {gusfield1997algorithms}
\bibfield{author}{\bibinfo{person}{Dan Gusfield}.}
  \bibinfo{year}{1997}\natexlab{}.
\newblock \showarticletitle{Algorithms on stings, trees, and sequences:
  Computer science and computational biology}.
\newblock \bibinfo{journal}{\emph{Acm Sigact News}} \bibinfo{volume}{28},
  \bibinfo{number}{4} (\bibinfo{year}{1997}), \bibinfo{pages}{41--60}.
\newblock


\bibitem[Ha et~al\mbox{.}(2014)]%
        {ha2014models}
\bibfield{author}{\bibinfo{person}{Hoai~Phuong Ha}, \bibinfo{person}{Ngoc
  Nha~Vi Tran}, \bibinfo{person}{Ibrahim Umar}, \bibinfo{person}{Philippas
  Tsigas}, \bibinfo{person}{Anders Gidenstam}, \bibinfo{person}{Paul
  Renaud-Goud}, \bibinfo{person}{Ivan Walulya}, {and} \bibinfo{person}{Aras
  Atalar}.} \bibinfo{year}{2014}\natexlab{}.
\newblock \showarticletitle{Models for energy consumption of data structures
  and algorithms}.
\newblock  (\bibinfo{year}{2014}).
\newblock


\bibitem[Hasenplaugh et~al\mbox{.}(2014)]%
        {hasenplaugh2014ordering}
\bibfield{author}{\bibinfo{person}{William Hasenplaugh}, \bibinfo{person}{Tim
  Kaler}, \bibinfo{person}{Tao~B. Schardl}, {and} \bibinfo{person}{Charles~E.
  Leiserson}.} \bibinfo{year}{2014}\natexlab{}.
\newblock \showarticletitle{Ordering heuristics for parallel graph coloring}.
  In \bibinfo{booktitle}{\emph{{ACM} Symposium on Parallelism in Algorithms and
  Architectures (SPAA)}}.
\newblock


\bibitem[Hunt and Szymanski(1977)]%
        {hunt1977fast}
\bibfield{author}{\bibinfo{person}{James~W Hunt} {and}
  \bibinfo{person}{Thomas~G Szymanski}.} \bibinfo{year}{1977}\natexlab{}.
\newblock \showarticletitle{A fast algorithm for computing longest common
  subsequences}.
\newblock \bibinfo{journal}{\emph{Commun. ACM}} \bibinfo{volume}{20},
  \bibinfo{number}{5} (\bibinfo{year}{1977}), \bibinfo{pages}{350--353}.
\newblock


\bibitem[Inoue et~al\mbox{.}(2018)]%
        {inoue2018computing}
\bibfield{author}{\bibinfo{person}{Takafumi Inoue}, \bibinfo{person}{Shunsuke
  Inenaga}, \bibinfo{person}{Heikki Hyyr{\"o}}, \bibinfo{person}{Hideo Bannai},
  {and} \bibinfo{person}{Masayuki Takeda}.} \bibinfo{year}{2018}\natexlab{}.
\newblock \showarticletitle{Computing longest common square subsequences}. In
  \bibinfo{booktitle}{\emph{cpm}}. Schloss Dagstuhl-Leibniz-Zentrum f{\"u}r
  Informatik, Dagstuhl Publishing.
\newblock


\bibitem[J{\'a}J{\'a}(1992)]%
        {JaJa92}
\bibfield{author}{\bibinfo{person}{Joseph J{\'a}J{\'a}}.}
  \bibinfo{year}{1992}\natexlab{}.
\newblock \bibinfo{booktitle}{\emph{Introduction to Parallel Algorithms}}.
\newblock \bibinfo{publisher}{Addison-Wesley Professional}.
\newblock


\bibitem[Johansson(1998)]%
        {johansson1998longest}
\bibfield{author}{\bibinfo{person}{Kurt Johansson}.}
  \bibinfo{year}{1998}\natexlab{}.
\newblock \showarticletitle{The longest increasing subsequence in a random
  permutation and a unitary random matrix model}.
\newblock \bibinfo{journal}{\emph{Mathematical Research Letters}}
  \bibinfo{volume}{5}, \bibinfo{number}{1} (\bibinfo{year}{1998}),
  \bibinfo{pages}{68--82}.
\newblock


\bibitem[Jones and Plassmann(1993)]%
        {jones1993parallel}
\bibfield{author}{\bibinfo{person}{Mark~T. Jones} {and}
  \bibinfo{person}{Paul~E. Plassmann}.} \bibinfo{year}{1993}\natexlab{}.
\newblock \showarticletitle{A parallel graph coloring heuristic}.
\newblock  \bibinfo{volume}{14}, \bibinfo{number}{3} (\bibinfo{year}{1993}),
  \bibinfo{pages}{654--669}.
\newblock


\bibitem[Knuth(1973)]%
        {Knuth73vol3}
\bibfield{author}{\bibinfo{person}{Donald~E. Knuth}.}
  \bibinfo{year}{1973}\natexlab{}.
\newblock \bibinfo{booktitle}{\emph{The Art of Computer Programming, Volume
  III: Sorting and Searching}}.
\newblock \bibinfo{publisher}{Addison-Wesley}.
\newblock


\bibitem[Koster et~al\mbox{.}(2001)]%
        {koster2001treewidth}
\bibfield{author}{\bibinfo{person}{Arie~MCA Koster}, \bibinfo{person}{Hans~L
  Bodlaender}, {and} \bibinfo{person}{Stan~PM Van~Hoesel}.}
  \bibinfo{year}{2001}\natexlab{}.
\newblock \showarticletitle{Treewidth: computational experiments}.
\newblock \bibinfo{journal}{\emph{Electronic Notes in Discrete Mathematics}}
  \bibinfo{volume}{8} (\bibinfo{year}{2001}), \bibinfo{pages}{54--57}.
\newblock


\bibitem[Krusche and Tiskin(2009)]%
        {krusche2009parallel}
\bibfield{author}{\bibinfo{person}{Peter Krusche} {and}
  \bibinfo{person}{Alexander Tiskin}.} \bibinfo{year}{2009}\natexlab{}.
\newblock \showarticletitle{Parallel longest increasing subsequences in
  scalable time and memory}. In \bibinfo{booktitle}{\emph{International
  Conference on Parallel Processing and Applied Mathematics}}. Springer,
  \bibinfo{pages}{176--185}.
\newblock


\bibitem[Krusche and Tiskin(2010)]%
        {krusche2010new}
\bibfield{author}{\bibinfo{person}{Peter Krusche} {and}
  \bibinfo{person}{Alexander Tiskin}.} \bibinfo{year}{2010}\natexlab{}.
\newblock \showarticletitle{New algorithms for efficient parallel string
  comparison}. In \bibinfo{booktitle}{\emph{{ACM} Symposium on Parallelism in
  Algorithms and Architectures (SPAA)}}. \bibinfo{pages}{209--216}.
\newblock


\bibitem[Lim(2019)]%
        {lim2019optimal}
\bibfield{author}{\bibinfo{person}{Wei~Quan Lim}.}
  \bibinfo{year}{2019}\natexlab{}.
\newblock \showarticletitle{Optimal Multithreaded Batch-Parallel 2-3 Trees}.
\newblock \bibinfo{journal}{\emph{arXiv preprint arXiv:1905.05254}}
  (\bibinfo{year}{2019}).
\newblock


\bibitem[Lipski and Preparata(1981)]%
        {lipski1981efficient}
\bibfield{author}{\bibinfo{person}{Witold Lipski} {and}
  \bibinfo{person}{Franco~P Preparata}.} \bibinfo{year}{1981}\natexlab{}.
\newblock \showarticletitle{Efficient algorithms for finding maximum matchings
  in convex bipartite graphs and related problems}.
\newblock \bibinfo{journal}{\emph{Acta Informatica}} \bibinfo{volume}{15},
  \bibinfo{number}{4} (\bibinfo{year}{1981}), \bibinfo{pages}{329--346}.
\newblock


\bibitem[Lipski~Jr(1983)]%
        {lipski1983finding}
\bibfield{author}{\bibinfo{person}{Witold Lipski~Jr}.}
  \bibinfo{year}{1983}\natexlab{}.
\newblock \showarticletitle{Finding a Manhattan path and related problems}.
\newblock \bibinfo{journal}{\emph{Networks}} \bibinfo{volume}{13},
  \bibinfo{number}{3} (\bibinfo{year}{1983}), \bibinfo{pages}{399--409}.
\newblock


\bibitem[Nakashima and Fujiwara(2002)]%
        {nakashima2002parallel}
\bibfield{author}{\bibinfo{person}{Takaaki Nakashima} {and}
  \bibinfo{person}{Akihiro Fujiwara}.} \bibinfo{year}{2002}\natexlab{}.
\newblock \showarticletitle{Parallel algorithms for patience sorting and
  longest increasing subsequence}. In \bibinfo{booktitle}{\emph{International
  Conference in Networks, Parallel and Distributed Processing and
  Applications}}. \bibinfo{pages}{7--12}.
\newblock


\bibitem[Nakashima and Fujiwara(2006)]%
        {nakashima2006cost}
\bibfield{author}{\bibinfo{person}{Takaaki Nakashima} {and}
  \bibinfo{person}{Akihiro Fujiwara}.} \bibinfo{year}{2006}\natexlab{}.
\newblock \showarticletitle{A cost optimal parallel algorithm for patience
  sorting}.
\newblock \bibinfo{journal}{\emph{Parallel processing letters}}
  \bibinfo{volume}{16}, \bibinfo{number}{01} (\bibinfo{year}{2006}),
  \bibinfo{pages}{39--51}.
\newblock


\bibitem[Narisada et~al\mbox{.}(2017)]%
        {narisada2017computing}
\bibfield{author}{\bibinfo{person}{Shintaro Narisada},
  \bibinfo{person}{Kazuyuki Narisawa}, \bibinfo{person}{Shunsuke Inenaga},
  \bibinfo{person}{Ayumi Shinohara}, {et~al\mbox{.}}}
  \bibinfo{year}{2017}\natexlab{}.
\newblock \showarticletitle{Computing longest single-arm-gapped palindromes in
  a string}. In \bibinfo{booktitle}{\emph{International Conference on Current
  Trends in Theory and Practice of Informatics}}. Springer,
  \bibinfo{pages}{375--386}.
\newblock


\bibitem[O’Donnell and Wright({[n.\,d.]})]%
        {oprimer}
\bibfield{author}{\bibinfo{person}{Ryan O’Donnell} {and}
  \bibinfo{person}{John Wright}.} \bibinfo{year}{[n.\,d.]}\natexlab{}.
\newblock \showarticletitle{A primer on the statistics of longest increasing
  subsequences and quantum states}.
\newblock \bibinfo{journal}{\emph{SIGACT News}} (\bibinfo{year}{[n.\,d.]}).
\newblock


\bibitem[Pan et~al\mbox{.}(2015)]%
        {pan2015parallel}
\bibfield{author}{\bibinfo{person}{Xinghao Pan}, \bibinfo{person}{Dimitris
  Papailiopoulos}, \bibinfo{person}{Samet Oymak}, \bibinfo{person}{Benjamin
  Recht}, \bibinfo{person}{Kannan Ramchandran}, {and}
  \bibinfo{person}{Michael~I. Jordan}.} \bibinfo{year}{2015}\natexlab{}.
\newblock \showarticletitle{Parallel correlation clustering on big graphs}. In
  \bibinfo{booktitle}{\emph{Advances in Neural Information Processing Systems
  (NIPS)}}. \bibinfo{pages}{82--90}.
\newblock


\bibitem[Schensted(1961)]%
        {schensted1961longest}
\bibfield{author}{\bibinfo{person}{Craige Schensted}.}
  \bibinfo{year}{1961}\natexlab{}.
\newblock \showarticletitle{Longest increasing and decreasing subsequences}.
\newblock \bibinfo{journal}{\emph{Canadian Journal of Mathematics}}
  \bibinfo{volume}{13} (\bibinfo{year}{1961}), \bibinfo{pages}{179--191}.
\newblock


\bibitem[Sem{\'e}(2006)]%
        {seme2006cgm}
\bibfield{author}{\bibinfo{person}{David Sem{\'e}}.}
  \bibinfo{year}{2006}\natexlab{}.
\newblock \showarticletitle{A CGM algorithm solving the longest increasing
  subsequence problem}. In \bibinfo{booktitle}{\emph{International Conference
  on Computational Science and Its Applications}}. Springer,
  \bibinfo{pages}{10--21}.
\newblock


\bibitem[Shen et~al\mbox{.}(2022)]%
        {shen2022many}
\bibfield{author}{\bibinfo{person}{Zheqi Shen}, \bibinfo{person}{Zijin Wan},
  \bibinfo{person}{Yan Gu}, {and} \bibinfo{person}{Yihan Sun}.}
  \bibinfo{year}{2022}\natexlab{}.
\newblock \showarticletitle{Many Sequential Iterative Algorithms Can Be
  Parallel and (Nearly) Work-efficient}. In \bibinfo{booktitle}{\emph{{ACM}
  Symposium on Parallelism in Algorithms and Architectures (SPAA)}}.
\newblock


\bibitem[Shun et~al\mbox{.}(2015)]%
        {shun2015sequential}
\bibfield{author}{\bibinfo{person}{Julian Shun}, \bibinfo{person}{Yan Gu},
  \bibinfo{person}{Guy~E. Blelloch}, \bibinfo{person}{Jeremy~T. Fineman}, {and}
  \bibinfo{person}{Phillip~B Gibbons}.} \bibinfo{year}{2015}\natexlab{}.
\newblock \showarticletitle{Sequential random permutation, list contraction and
  tree contraction are highly parallel}. In
  \bibinfo{booktitle}{\emph{{ACM-SIAM} Symposium on Discrete Algorithms
  (SODA)}}. \bibinfo{pages}{431--448}.
\newblock


\bibitem[Snoeyink(1992)]%
        {snoeyink1992two}
\bibfield{author}{\bibinfo{person}{Jack Snoeyink}.}
  \bibinfo{year}{1992}\natexlab{}.
\newblock \showarticletitle{Two-and Three-Dimensional Point Location in
  Rectangular Subdivisions}. In \bibinfo{booktitle}{\emph{swat}},
  Vol.~\bibinfo{volume}{621}. Springer Verlag, \bibinfo{pages}{352}.
\newblock


\bibitem[Sun and Blelloch(2019)]%
        {sun2019parallel}
\bibfield{author}{\bibinfo{person}{Yihan Sun} {and} \bibinfo{person}{Guy~E
  Blelloch}.} \bibinfo{year}{2019}\natexlab{}.
\newblock \showarticletitle{Parallel Range, Segment and Rectangle Queries with
  Augmented Maps}. In \bibinfo{booktitle}{\emph{SIAM Symposium on Algorithm
  Engineering and Experiments (ALENEX)}}. \bibinfo{pages}{159--173}.
\newblock


\bibitem[Sun et~al\mbox{.}(2018)]%
        {sun2018pam}
\bibfield{author}{\bibinfo{person}{Yihan Sun}, \bibinfo{person}{Daniel
  Ferizovic}, {and} \bibinfo{person}{Guy~E Blelloch}.}
  \bibinfo{year}{2018}\natexlab{}.
\newblock \showarticletitle{{PAM}: Parallel Augmented Maps}. In
  \bibinfo{booktitle}{\emph{{ACM} Symposium on Principles and Practice of
  Parallel Programming (PPOPP)}}.
\newblock


\bibitem[Tang et~al\mbox{.}(2015)]%
        {tang2015cache}
\bibfield{author}{\bibinfo{person}{Yuan Tang}, \bibinfo{person}{Ronghui You},
  \bibinfo{person}{Haibin Kan}, \bibinfo{person}{Jesmin~Jahan Tithi},
  \bibinfo{person}{Pramod Ganapathi}, {and} \bibinfo{person}{Rezaul~A
  Chowdhury}.} \bibinfo{year}{2015}\natexlab{}.
\newblock \showarticletitle{Cache-oblivious wavefront: improving parallelism of
  recursive dynamic programming algorithms without losing cache-efficiency}. In
  \bibinfo{booktitle}{\emph{{ACM} Symposium on Principles and Practice of
  Parallel Programming (PPOPP)}}.
\newblock


\bibitem[Thierry et~al\mbox{.}(2001)]%
        {thierry2001work}
\bibfield{author}{\bibinfo{person}{Garcia Thierry}, \bibinfo{person}{Myoupo
  Jean-Fr{\'e}d{\'e}ric}, {and} \bibinfo{person}{Sem{\'e} David}.}
  \bibinfo{year}{2001}\natexlab{}.
\newblock \showarticletitle{A work-optimal CGM algorithm for the LIS problem}.
  In \bibinfo{booktitle}{\emph{{ACM} Symposium on Parallelism in Algorithms and
  Architectures (SPAA)}}. \bibinfo{pages}{330--331}.
\newblock


\bibitem[Tiskin(2015)]%
        {tiskin2015fast}
\bibfield{author}{\bibinfo{person}{Alexander Tiskin}.}
  \bibinfo{year}{2015}\natexlab{}.
\newblock \showarticletitle{Fast distance multiplication of unit-Monge
  matrices}.
\newblock \bibinfo{journal}{\emph{Algorithmica}} \bibinfo{volume}{71},
  \bibinfo{number}{4} (\bibinfo{year}{2015}), \bibinfo{pages}{859--888}.
\newblock


\bibitem[Tomkins et~al\mbox{.}(2014)]%
        {tomkins2014sccmulti}
\bibfield{author}{\bibinfo{person}{Daniel Tomkins}, \bibinfo{person}{Timmie
  Smith}, \bibinfo{person}{Nancy~M Amato}, {and} \bibinfo{person}{Lawrence
  Rauchwerger}.} \bibinfo{year}{2014}\natexlab{}.
\newblock \showarticletitle{SCCMulti: an improved parallel strongly connected
  components algorithm}.
\newblock \bibinfo{journal}{\emph{{ACM} Symposium on Principles and Practice of
  Parallel Programming (PPOPP)}} \bibinfo{volume}{49}, \bibinfo{number}{8}
  (\bibinfo{year}{2014}), \bibinfo{pages}{393--394}.
\newblock


\bibitem[Umar et~al\mbox{.}(2013)]%
        {umar2013deltatree}
\bibfield{author}{\bibinfo{person}{Ibrahim Umar}, \bibinfo{person}{Otto
  Anshus}, {and} \bibinfo{person}{Phuong Ha}.} \bibinfo{year}{2013}\natexlab{}.
\newblock \showarticletitle{Deltatree: A practical locality-aware concurrent
  search tree}.
\newblock \bibinfo{journal}{\emph{arXiv preprint arXiv:1312.2628}}
  (\bibinfo{year}{2013}).
\newblock


\bibitem[van Emde~Boas(1977)]%
        {van1977preserving}
\bibfield{author}{\bibinfo{person}{Peter van Emde~Boas}.}
  \bibinfo{year}{1977}\natexlab{}.
\newblock \showarticletitle{Preserving order in a forest in less than
  logarithmic time and linear space}.
\newblock \bibinfo{journal}{\emph{Inform. Process. Lett.}} \bibinfo{volume}{6},
  \bibinfo{number}{3} (\bibinfo{year}{1977}), \bibinfo{pages}{80--82}.
\newblock


\bibitem[van Emde~Boas et~al\mbox{.}(1976)]%
        {van1976design}
\bibfield{author}{\bibinfo{person}{Peter van Emde~Boas},
  \bibinfo{person}{Robert Kaas}, {and} \bibinfo{person}{Erik Zijlstra}.}
  \bibinfo{year}{1976}\natexlab{}.
\newblock \showarticletitle{Design and implementation of an efficient priority
  queue}.
\newblock \bibinfo{journal}{\emph{Mathematical systems theory}}
  \bibinfo{volume}{10}, \bibinfo{number}{1} (\bibinfo{year}{1976}),
  \bibinfo{pages}{99--127}.
\newblock


\bibitem[Wang et~al\mbox{.}(2021)]%
        {wang2020parallel}
\bibfield{author}{\bibinfo{person}{Yiqiu Wang}, \bibinfo{person}{Shangdi Yu},
  \bibinfo{person}{Yan Gu}, {and} \bibinfo{person}{Julian Shun}.}
  \bibinfo{year}{2021}\natexlab{}.
\newblock \showarticletitle{A Parallel Batch-Dynamic Data Structure for the
  Closest Pair Problem}. In \bibinfo{booktitle}{\emph{ACM Symposium on
  Computational Geometry (SoCG)}}.
\newblock


\bibitem[Williams et~al\mbox{.}(2021)]%
        {williams2021engineering}
\bibfield{author}{\bibinfo{person}{Marvin Williams}, \bibinfo{person}{Peter
  Sanders}, {and} \bibinfo{person}{Roman Dementiev}.}
  \bibinfo{year}{2021}\natexlab{}.
\newblock \showarticletitle{Engineering MultiQueues: Fast Relaxed Concurrent
  Priority Queues}. In \bibinfo{booktitle}{\emph{European Symposium on
  Algorithms (ESA)}}.
\newblock


\bibitem[Yang et~al\mbox{.}(2005)]%
        {yang2005fast}
\bibfield{author}{\bibinfo{person}{I-Hsuan Yang}, \bibinfo{person}{Chien-Pin
  Huang}, {and} \bibinfo{person}{Kun-Mao Chao}.}
  \bibinfo{year}{2005}\natexlab{}.
\newblock \showarticletitle{A fast algorithm for computing a longest common
  increasing subsequence}.
\newblock \bibinfo{journal}{\emph{Inform. Process. Lett.}}
  \bibinfo{volume}{93}, \bibinfo{number}{5} (\bibinfo{year}{2005}),
  \bibinfo{pages}{249--253}.
\newblock


\bibitem[Zhang(2003)]%
        {zhang2003alignment}
\bibfield{author}{\bibinfo{person}{Hongyu Zhang}.}
  \bibinfo{year}{2003}\natexlab{}.
\newblock \showarticletitle{Alignment of BLAST high-scoring segment pairs based
  on the longest increasing subsequence algorithm}.
\newblock \bibinfo{journal}{\emph{Bioinformatics}} \bibinfo{volume}{19},
  \bibinfo{number}{11} (\bibinfo{year}{2003}), \bibinfo{pages}{1391--1396}.
\newblock


\bibitem[Zhou et~al\mbox{.}(2019)]%
        {zhou2019practical}
\bibfield{author}{\bibinfo{person}{Tingzhe Zhou}, \bibinfo{person}{Maged
  Michael}, {and} \bibinfo{person}{Michael Spear}.}
  \bibinfo{year}{2019}\natexlab{}.
\newblock \showarticletitle{A Practical, Scalable, Relaxed Priority Queue}. In
  \bibinfo{booktitle}{\emph{International Conference on Parallel Processing
  (ICPP)}}. \bibinfo{pages}{1--10}.
\newblock


\end{thebibliography}
